%% file: main.tex
\documentclass[11pt]{article}
\usepackage{fullpage}

\usepackage{amsthm,amsmath,amssymb,enumerate}
\usepackage{complexity}
\usepackage{multirow}
\usepackage{mathtools}
\usepackage{tabto}
\usepackage{comment}
\usepackage{stmaryrd}
\usepackage{hyperref}
\usepackage{caption} 
\usepackage{thm-restate}
\usepackage{enumitem}
\usepackage{array}

\newtheorem{theorem}{Theorem}[section]
\newtheorem{definition}[theorem]{Definition}

\newtheorem{observation}[theorem]{Observation}
\newtheorem{lemma}[theorem]{Lemma}

\newtheorem{claim}[theorem]{Claim}
\newtheorem{fact}[theorem]{Fact}

\newcounter{note}[section]

\input{preamble}

\newcommand{\repeattheorem}[1]{%
  \begingroup
  \renewcommand{\thetheorem}{\ref{#1}}%
  \expandafter\expandafter\expandafter\theorem
  \csname reptheorem@#1\endcsname
  \endtheorem
  \endgroup
}

\title{Dynamic Geometric Independent Set}

\date{}

\author{Sujoy Bhore\thanks{Universit\'e libre de Bruxelles (ULB), Belgium. Emails: \texttt{sujoy.bhore@gmail.com, jcardin@ulb.ac.be, john@johniacono.com, gregkoumoutsos@gmail.com }} \thanks{Supported by the Fonds de la Recherche Scientifique-FNRS under Grant no MISU F 6001 1}
\qquad
Jean Cardinal\footnotemark[1] 
\qquad
John Iacono\footnotemark[1] \footnotemark[2]
\qquad
Grigorios Koumoutsos\footnotemark[1] \footnotemark[2]
}

\renewcommand{\thenote}{\thesection.\arabic{note}}

\newcommand{\sbnote}[1]{\textcolor[rgb]{0,0.5,0}{$\ll${\bf Sujoy~\thenote:} {\sf #1}$\gg$}}

\newcommand{\jinote}[1]{\textcolor{blue}{$\ll${\bf John~\thenote:} {\sf #1}$\gg$}}
\newcommand{\gknote}[1]{\textcolor{blue}{$\ll${\bf Greg~\thenote:} {\sf #1}$\gg$}}

\newcommand{\lrbrace}[1]{\lbrace#1\rbrace}

\input{00-newcommands}

\begin{document}

\maketitle

\thispagestyle{empty}
\begin{abstract}

We present fully dynamic approximation algorithms for the Maximum Independent Set problem on several types of geometric objects: intervals on the real line, arbitrary axis-aligned squares in the plane and axis-aligned $d$-dimensional hypercubes.\\

It is known that a maximum independent set of a collection of $n$ intervals can be found in $O(n\log n)$ time, while it is already \textsf{NP}-hard for a set of unit squares. Moreover, the problem is inapproximable on many important graph families, but admits a \textsf{PTAS} for a set of arbitrary pseudo-disks. Therefore, a fundamental question in computational geometry is whether it is possible to maintain an approximate maximum independent set in a set of dynamic geometric objects, in truly sublinear time per insertion or deletion. 
In this work, we answer this question in the affirmative for intervals, squares and hypercubes.\\

First, we show that for intervals a $(1+\varepsilon)$-approximate maximum independent set can be maintained with logarithmic worst-case update time. This is achieved by maintaining a locally optimal solution using a constant number of constant-size exchanges per update.\\

We then show how our interval structure can be used to design a data structure for maintaining an expected constant factor approximate maximum independent set of axis-aligned squares in the plane, with polylogarithmic amortized update time. Our approach generalizes to $d$-dimensional hypercubes, providing a $O(4^d)$-approximation with polylogarithmic update time.\\ 

Those are the first approximation algorithms for any set of dynamic arbitrary size geometric objects; previous results required bounded size ratios to obtain polylogarithmic update time. Furthermore, it is known that our results for squares (and hypercubes) cannot be improved to a $(1+\varepsilon)$-approximation with the same update time.  
\end{abstract}

\newpage
\thispagestyle{empty}
\tableofcontents

\clearpage
\setcounter{page}{1}
\input{1-introduction}

\clearpage
\input{1b-outline}

\clearpage
\input{2-dynamicIntervals}

\clearpage
\input{3-quadtreeApproach}

\clearpage
\input{4-dynamicSquares}

\clearpage
\bibliographystyle{plain}
\bibliography{MIS}

\end{document}

%% file: preamble.tex
\usepackage[table]{xcolor}
\usepackage{wrapfig}
\usepackage{float}
\usepackage[textsize=footnotesize]{todonotes}

\newcommand{\iab}{\ensuremath{(I \setminus A) \cup B}}
\newcommand{\jop}[1]{\textsc{#1}}

\newcommand{\np}{\textsf{NP}}
\newcommand{\old}[1]{{}}

\DeclareMathOperator{\NEXT}{NEXT}
\DeclareMathOperator{\NULL}{NULL}

\hypersetup{
hidelinks,
colorlinks=false,
citecolor=[rgb]{1 0 0},
linkcolor=[rgb]{1 0 0},
urlcolor=[rgb]{1 0 0}
}

%% file: 00-newcommands.tex
\newcommand{\jmark}[1]{#1}
\newcommand{\jparen}[1]{(#1)}

\newcommand{\lrp}[1]{\left( #1 \right)}
\newcommand{\paragraphh}[1]{\paragraph{{#1}} }

\newcommand{\manysquares}{\raisebox{1pt}{\resizebox{8pt}{!}{$_{\square \square}^{\square \square}$}}}

\newcommand{\Centered}{\jmark{\boxbox (\Sinput)}}
\newcommand{\Sinput}{\jmark{S}}
\newcommand{\Iopt}[1]{\jmark{I_\text{opt}}\jparen{#1}}
\newcommand{\OPT}[1]{\jmark{\alpha}\jparen{#1}}
\newcommand{\Iapprox}{\jmark{I}}
\newcommand{\Ipath}[1]{\jmark{I^{#1}}}
\newcommand{\Ippath}[1]{\jmark{\hat{I}^{#1}}}
\newcommand{\Iintopt}[1]{\jmark{I^{#1}_{\text{$\delta$-opt}}}}
\newcommand{\Iintapprox}[1]{\jmark{I^{#1}_{\text{$\delta$-approx}}}}

\newcommand{\node}{\jmark{\eta}}
\newcommand{\Qinf}{\jmark{\stackrel{\infty}{Q}}}
\newcommand{\Qtree}{\jmark{Q}}
\newcommand{\Qleaves}{\jmark{Q_{\text{leaves}}}}
\newcommand{\Qinternal}{\jmark{Q_{\text{internal}}}}
\newcommand{\Qpaths}{\jmark{Q_{\text{paths}}}}
\newcommand{\Interval}[2]{\jmark{\delta}_{#1}\jparen{#2}}
\newcommand{\Squarecenter}[1]{\jmark{\boxdot}\jparen{#1}}
\newcommand{\Squares}[1]{{\jmark{\manysquares}\jparen{#1}}}

\newcommand{\Path}{\jmark{P}}
\newcommand{\Pathtop}{\jmark{P_{\text{top}}}}
\newcommand{\Pathbottom}{\jmark{P_{\text{bottom}}}}
\newcommand{\Pathq}[1]{\jmark{{P}}_{#1}}
\newcommand{\Pathextend}[1]{\jmark{\hat{P}}_{#1}}
\newcommand{\Square}[1]{#1_{\jmark{\square}}}
\newcommand{\Depth}[1]{\jmark{d}\jparen{#1}}
\newcommand{\Depthmax}[2]{\jmark{d'}_{#1}\jparen{#2}}
\newcommand{\Depths}[1]{\jmark{D}\jparen{#1}}
\newcommand{\Node}[1]{\jmark{\node}\jparen{#1}}
\newcommand{\Marked}{\jmark{M}}
\newcommand{\qa}{\jmark{\textsc{I}}}
\newcommand{\qb}{\jmark{\textsc{II}}}
\newcommand{\qc}{\jmark{\textsc{III}}}
\newcommand{\qd}{\jmark{\textsc{IV}}}
\newcommand{\Half}[2]{\jmark{\moo}\jparen{#2}} 
\DeclareMathOperator{\opt}{OPT}
\DeclareMathOperator{\alg}{ALG}

%% file: 1-introduction.tex
\section{Introduction}

We consider the maximum independent set problem on dynamic collections of geometric objects. We wish to maintain, at any given time, an 
approximately maximum subset of pairwise nonintersecting objects, under the two natural update operations of insertion and deletion of 
an object. Before providing an outline of our results and the methods that we used, we briefly summarize the background and state of the art related to the independent set problem and dynamic algorithms on geometric inputs.

In the maximum independent set (MIS) problem, we are given a graph $G = (V,E)$ and we aim to produce a subset $I \subseteq V$ of maximum cardinality, such that no two vertices in $I$ are adjacent. This is one of the most well-studied algorithmic problems and it is among the Karp's 21 classic \textsf{NP}-complete problems \cite{karp1972reducibility}. Moreover, it is well-known to be hard to approximate: no polynomial time algorithm can achieve an approximation factor $n^{1-\epsilon}$, for any constant $\epsilon>0$, unless $\P = \np$~\cite{Zuck07,Hastad1999}.  

\paragraph{Geometric Independent Set.} 
Despite those strong hardness results, for several restricted cases of the MIS problem better results can be obtained. We focus on such cases with geometric structure, called \textit{geometric independent sets}. Here, we are given a set $S$ of geometric objects, and the graph $G$ is their intersection graph, where each vertex corresponds to an object, and two vertices form an edge if and only if the corresponding objects intersect.  

A fundamental and well-studied problem is the 1-dimensional case where all objects are intervals. This is also known as the \textit{interval scheduling} problem and has several applications in scheduling, resource allocation, etc. This is one of the few cases of the MIS problem which can be solved in polynomial time; it is a standard textbook result (see e.g.~\cite{KleinTardos}) that the greedy algorithm which sweeps the line from left to right and at each step picks the interval with the leftmost right endpoint produces always the optimal solution in time $O(n \log n)$. 

Independent sets of geometric objects in the plane such as axis-aligned squares or rectangles have been extensively studied due to their various applications in e.g., VLSI design~\cite{HM85}, map labeling~\cite{AKS97} and data mining~\cite{KMP98,BDMR01}. However, even the case of independent set of unit squares is \textsf{NP}-complete~\cite{fowler1981optimal}. On the positive side several geometric cases admit a polynomial time approximation scheme (PTAS). One of the first results was due to Hochbaum and Maass who gave a PTAS for unit $d$-cubes in $R^d$~\cite{HM85} (therefore also for unit squares in 2-d). Later, PTAS were also developed for arbitrary squares and more generally hypercubes and fat objets~\cite{chan2003polynomial, erlebach2005polynomial}. More recently, Chan and Har-Peled~\cite{chan2012approximation} showed that for all pseudodisks (which include squares) a PTAS can be achieved using local search.

Despite this remarkable progress, even seemingly simple cases such as axis-parallel rectangles in the plane, are notoriously hard and no PTAS is known. For rectangles, the best known approximation is $O(\log \log n)$ due to the breakthrough result of Chalermsook and Chuzhoy~\cite{Parinya09}. Recently, several QPTAS were designed~\cite{aw-asmwir-13,ce-amir-16}, but still  no polynomial $o(\log \log n)$-approximation is known. 

\paragraph{Dynamic Independent Set.} In the dynamic version of the Independent Set problem, nodes of $V$ are inserted and deleted over time. The goal is to achieve (almost) the same approximation ratio as in the offline (static) case while keeping the update time, i.e.,  the running time required to compute the new solution after insertion/deletion, as small as possible. Dynamic algorithms have been a very active area of research and several fundamental problems, such as Set-Cover have been studied in this setting (we discuss some of those results in Section~\ref{sec:related}).

\paragraph{Previous Work.} Very recently, Henzinger et al.~\cite{Henz20} studied geometric independent set for intervals, hypercubes and hyperrectangles. They obtained several results, many of which extend to the substantially more general weighted independent set problem where objects have weights (we discuss this briefly in Section~\ref{sec:related}). Here we discuss only the results relevant to our context.

Based on a lower bound of Marx~\cite{Marx07} for the offline problem, Henzinger et al.~\cite{Henz20} showed that any dynamic $(1+\epsilon)$-approximation for squares requires $\Omega(n^{1/\epsilon})$ update time, ruling out the possibility of sublinear dynamic approximation schemes.

As for upper bounds, Henzinger et al.~\cite{Henz20} considered the setting where all objects are located in $[0,N]^d$ and have minimum length edge of 1, hence therefore also bounded size ratio of $N$. They presented dynamic algorithms with update time $\mathrm{polylog}(n,N)$. We note that in general, $N$ might be quite large such as $\exp{n}$ or even unbounded, thus those bounds are not sublinear in $n$ in the general case. In another related work, Gavruskin et al.~\cite{gavruskin2015dynamic} considered the interval case under the assumption that no interval is fully contained in other interval and obtained an optimal solution with $O(\log n )$ amortized update time.

Quite surprisingly, no other results are known. In particular, even the problem of efficiently maintaining an independent set of intervals, without any extra assumptions on the input, remained open.

\subsection{Our Results}

In this work, we present the first dynamic algorithms with polylogarithmic update time for geometric versions of the independent set problem.

First, we consider the 1-dimensional case of dynamic independent set of intervals. 

\begin{theorem}
\label{thm:intervals_result}
There exist algorithms for maintaining a $(1+\epsilon)$-approximate independent set of intervals under insertions and deletions of intervals, in $O_{\epsilon}(\log n)$ worst-case time per update, where $\varepsilon>0$ is any positive constant and $n$ is the total number of intervals.
\end{theorem}

This is the first algorithm yielding such a guarantee in the comparison model, in which the only operations allowed on the input are comparisons between endpoints of intervals.

To achieve this result we use a novel application of local search to dynamic algorithms, based on the paradigm of Chan and Har-Peled~\cite{chan2012approximation} for the static version of the problem. At a very high-level (and ignoring some details) our algorithms can be phrased as follows: Given our current independent set $I$ and the new (inserted/deleted) interval $x$, if there exists a subset of $t \leq k$ intervals which can be replaced by $t+1$, do this change. We show that using such a simple strategy, the resulting independent set has always size at least a fraction $(1 - \frac{c}{k})$ of the maximum. The main ideas and the description of our algorithms is in Section~\ref{sec:intervals_main}. The detailed analysis and proof of running time are in Section~\ref{sec:intervals_details}.

Next, we consider the problem of maintaining dynamically an independent set of squares. A natural question to ask is whether we can again apply local search. The problem is not with local search itself: an $(1+\epsilon)$-approximate MIS can be obtained if there are no local exchanges of certain size possible (due to the result of Chan and Har-Peled~\cite{chan2012approximation}); the problem is algorithmically implementing these local exchanges, which comes down to the issue that the 2-D generalization of maximum has linear size and not constant size. Note that the lower bound of Henzinger et. al.~\cite{Henz20} also implies that local search on squares cannot be implemented in polylogarithmic time. 

To circumvent this, we adopt a completely different technique, reducing the case of squares to intervals while losing a $O(1)$ factor in the approximation. We conjecture that one could implement local search to yield a $(1+\epsilon)$-approximation by using some kind of sophisticated range search to find local exchanges, at a cost of $O(n^c)$ for some $c>1$, which is another tradeoff that conforms to the lower bound of~\cite{Henz20}.

\begin{theorem}
\label{thm:squares_result}
There exist algorithms for maintaining an expected $O(1)$-approximate independent set of axis-aligned squares in the plane under insertions and deletions of squares, in $O(\log^5 n)$ amortized time per update, where $n$ is the total number of squares.
\end{theorem}

To obtain this result, we reduce the case of squares to intervals using a random quadtree and decomposing it carefully into relevant paths. First, we show that for the static case, given a $c$-approximate solution for intervals we can obtain a $O(c)$-approximate solution for squares (Section~\ref{s:quadtree}). To make this dynamic, more work is needed: we need a dynamic interval data structure supporting extra operations such as split, merge and some more. For that reason, we extend our structure from Theorem~\ref{thm:intervals_result} to support those additional operations while maintaining the same approximation ratio (Section~\ref{sec:intervals_extend}). Then, we dynamize our random quadtree approach to interact with the extended interval structure and obtain a $O(1)$-approximation for dynamic squares (Section~\ref{sec:squares_dynamic}).

We then show in Section~\ref{s:hyper} that our approach naturally extends to axis-aligned hypercubes in $d$ dimensions, providing a $O(4^d)$-approximate independent set in $O(2^d \log ^{2d+1} n)$ time.

\subsection{Other Related Work}
\label{sec:related}

\paragraph{Dynamic Algorithms.} Dynamic graph algorithms has been a continuous subject of investigation for many decades; see~\cite{EGI99}.  Over the last few years there has been a tremendous progress and various breakthrough results have been achieved for several fundamental problems. Among others, some of the recently studied problems are set cover~\cite{AAGPS19,BHN19,GKKP17}, geometric set cover and hitting set~\cite{DBLP:conf/compgeom/AgarwalCSXX20}, 
vertex cover~\cite{DBLP:conf/soda/BhattacharyaK19}, planarity testing~\cite{DBLP:conf/stoc/HolmR20} and graph coloring~\cite{DBLP:journals/algorithmica/BarbaCKLRRV19,DBLP:conf/soda/BhattacharyaCHN18,DBLP:conf/stacs/Henzinger020}.

A related problem to MIS is the problem of maintaining dynamically a \textit{maximal independent set}; this problem has numerous applications, especially in distributed and parallel computing. Since maximal is a local property, the problem is ``easier'' than MIS and allows for better approximation results even in general graphs. Very recently, several remarkable results have been obtained in the dynamic version of the problem~\cite{DBLP:conf/stoc/AssadiOSS18, DBLP:conf/soda/AssadiOSS19, DBLP:conf/focs/BehnezhadDHSS19, DBLP:conf/focs/ChechikZ19}.

\paragraph{Weighted Independent Set.} The maximum independent set problem we study here is special case of the more general weighted independent set (WIS) problem where each node has a weight and the goal is to produce an independent set of maximum weight. Clearly, MIS is the special case of WIS where all nodes have the same weight. 

The WIS problem has also been extensively studied. Usually stronger techniques that in MIS are needed. For instance, the greedy algorithm for intervals does not apply and obtaining the optimal solution in $O(n \log n)$ time requires a standard use of dynamic programming~\cite{KleinTardos}. Similarly, for squares the local-search technique of Chan and Har-Peled~\cite{chan2012approximation} does not provide a PTAS. This is the main reason that our approach here does not extend to the dynamic WIS problem.  


We note that dynamic WIS was studied in the recent work of Henzinger et al.~\cite{Henz20}. Authors provided dynamic algorithms for intervals, hypercubes and hyperrectangles lying in $[0,N]^d$ with minimum edge length 1, with update time polylog$(n,N,W)$, where $W$ is the maximum weight of an object.

%% file: 1b-outline.tex
\section{Outline of our Contributions}

In this section, we give a concise overview of the techniques involved in our algorithms.
The details of the proofs are deferred to the following sections.

\subsection{Intervals}
\label{sec:intervals_main}

We now give an overview of our dynamic algorithms for intervals achieving the bounds of Theorem~\ref{thm:intervals_result}.

\paragraph{Notation.} In what follows, $S$ denotes the current set of intervals, $n=|S|$ is the number of intervals, and $\alpha (S)$ denotes the size of a maximum independent set of $S$. We will show that our algorithms maintain a dynamic independent set $I$ such that $|I| \geq (1-\epsilon) \alpha(S)$ in time $O_{\varepsilon}(\log n)$, for $0 < \epsilon <1$. 

Note that while stating the results in the Introduction section, we used $a>1$ to denote the approximation ratio of an algorithm, meaning that $\opt / \alg \leq a$. Showing that $|I| \geq (1-\epsilon) \alpha(S)$ is equivalent to showing a $(1+\epsilon')$-approximation for $\epsilon' = \frac{1}{1-\epsilon} -1$ and Theorem~\ref{thm:intervals_result} follows.

\paragraph{Intuition.} We begin with some intuition and high-level ideas. Let us first mention, as observed by Henzinger et. al.~\cite{Henz20}, that trying to maintain maximum independent sets
exactly is hopeless, even in the case of intervals. Indeed, there are instances where $\Omega(n)$ changes are required, 
as illustrated in Figure~\ref{fig:maximum_global}.

\begin{figure}[ht]
    \centering
    \includegraphics[scale=1]{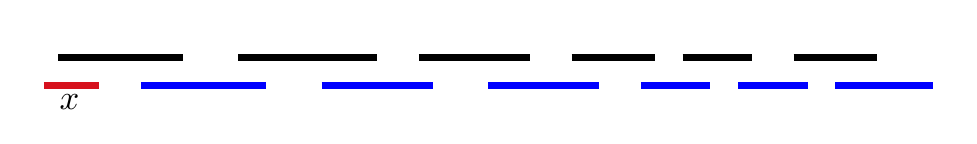}
    \caption{Example where a single insertion causes $\Omega(n)$ changes to the maximum independent set. If interval $x$ is not in the set, then the black intervals define a maximum independent set. Once $x$ gets inserted, then $x$ together with the blue intervals form the new maximum independent set.}
    \label{fig:maximum_global}
\end{figure}

Since we only aim at maintaining an approximate solution, we can focus on maintaining a $k$-maximal independent set. An independent set is $k$-maximal, if for any $t \leq k$, there is no set of $t$ intervals  that can be replaced by a set of $t+1$ intervals. Maintaining a $k$-maximal independent set implies that all changes will involve $O(k)$ intervals. 

\begin{definition}
\label{def:k-maximal}
  A $k$-maximal independent set $I \subseteq S$ for some integer $k\geq 0$ is a subcollection of disjoint intervals of $S$ such that for every positive integer $t \leq k$, there is no pair $A\in {\binom{I}{t}}$ and $B\in {S\setminus {\binom{I}{t+1}}}$ such that $(I\setminus A)\cup B$ is an independent set of $S$.
\end{definition}

Note that for $k=0$, this corresponds to the usual notion of inclusionwise maximality.
The following lemma states that local optimality provides an approximation guarantee. It is a special case of a much more general result of Chan and Har-Peled~\cite{chan2012approximation}~(Theorem 3.9).

\begin{lemma}
\label{lem:chan}
There exists a constant $c$ such that for any $k$-maximal independent set $I\subseteq S$, $|I|\geq (1-\frac ck)\cdot \alpha (S)$. 
\end{lemma}

Thus, we set as our goal the dynamic maintenance of a $k$-maximal independent set. 
It turns out, however, that even this is not easy and there might be cases where $\Omega(n/k)$ changes of $\Theta(k)$ intervals (therefore $\Omega(n)$ overall changes) are needed to maintain a $k$-maximal independent set. This is illustrated in Figure~\ref{fig:maximal_global}.

\begin{figure}[ht]
    \centering
    \includegraphics[scale=1]{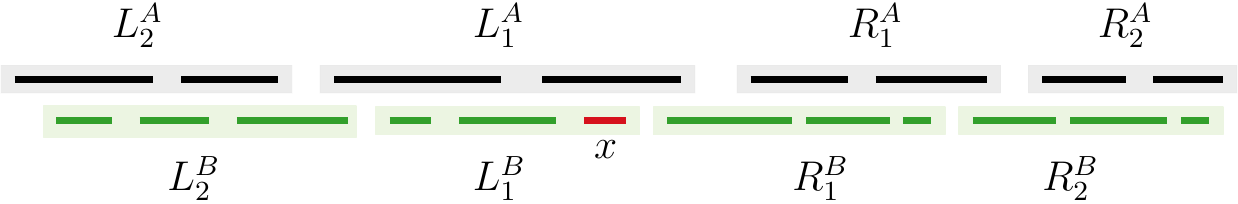}
    \caption{Example where a $k$-maximal independent set changes completely after a single insertion, for $k=2$. Before the insertion of $x$, the black intervals define a $k=2$-maximal independent set. Once $x$ gets inserted, then a 2-to-3 exchange is possible: the set $L^A_1$ of two intervals can be replaced by the set $L^B_1$ of three intervals. Once this exchange is made, other 2-to-3 exchanges are possible: the set $L^A_2$  can be replaced by $L^B_2$. Moreover,  the set $R^A_1$ can be replaced by $R^B_1$, which in turn enables the replacement of $R^A_2$ by $R^B_2$. The same changes percolate to the left and right for arbitrarily long instances. Observe that in all exchanges, one green interval is strictly contained in a black interval.}
    \label{fig:maximal_global}
\end{figure}

\paragraph{Our Approach.} To overcome those pathological instances, we observe that those occur because in a $k$-maximal independent set $I$, there might be intervals $y \in S \setminus I$ which are strictly contained in an interval $a \in I$. It turns out that if we eliminate this case, we can indeed maintain a $k$-maximal independent set in logarithmic update time. Thus our goal is to maintain a $k$-maximal independent set $I$, where there are no intervals of $S \setminus I$  that are strictly contained in intervals of $I$. We will call such independent sets \textit{$k$-valid}, as stated in the following definition.

\begin{definition}
\label{def:valid}
An independent set of intervals $I \subseteq S$ is called \textit{$k$-valid}, if it satisfies the following two properties:

\begin{enumerate}
\item \textit{No-containment:} No interval of $S \setminus I$ is completely contained in an interval of $I$.
\item \textit{$k$-maximality:} The independent set $I$ is $k$-maximal, according to definition~\ref{def:k-maximal}. 
\end{enumerate}
\end{definition}

Our main technical contribution is maintaining a $k$-valid independent set subject to insertions and deletions in time $O(k^2 \log n)$ (in fact for insertions our time is even better, $O(k \log n)$). 
Since by definition all $k$-valid independent sets are $k$-maximal, this combined with Lemma~\ref{lem:chan} implies the result. More precisely, for $\epsilon = c/k$, we get $|I| \geq (1 - \epsilon) \cdot \alpha(S)$ with update time  $O( \frac{\log n}{\epsilon})$ for insertions and $O( \frac{\log n}{\epsilon^2})$ for deletions.

\paragraph{Our Algorithm.} We now give the basic ideas behind our algorithm. 
Let $I$ be the current independent set we maintain. Suppose that there exists a pair $(A,B)$ of sizes $t$ and $t+1$, for $t \leq k$, such that $A \subseteq I$ and $B \cap I = \emptyset$ and $\iab$ is an independent set. Such a pair is a certificate that $I$ is not a $k$-maximal independent set. We call such a pair an \textit{alternating path}, since (as we show in Section~\ref{sec:intervals_details}, Lemma~\ref{lem:mes}) it induces an alternating path in the intersection graph of $I \cup B$. 


Our main algorithm is essentially based on searching alternating paths of size at most $k$. This can be done in time $O(k \log n)$ using our data structures (Section~\ref{sec:intervals_alg}).

\paragraph{Insertions.} 
Suppose a new interval $x$ gets inserted. 
Our insertion algorithm will be the following: 

\medskip

\noindent \textbf{Case 1:} $x$ is strictly contained in $a \in I$ (Figure~\ref{fig:insert_subset}). Then 
\begin{enumerate}
\item Replace $a$ by $x$.
\item Check on the left for an alternating path; if found, do the corresponding exchange. Same for right.  
\end{enumerate}

\noindent \textbf{Case 2:} $x$ is not contained. Then check if there exists alternating path $(A,B)$ involving $x$. If so, do this exchange.

\begin{figure}[t]
    \centering
    \includegraphics[scale=1]{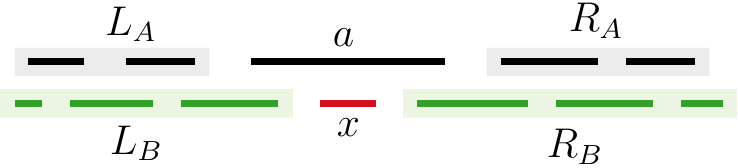}
    \caption{Case 1 of our insertion algorithm. Inserted interval $x$ is a subset of interval $a \in I$. After replacing $a$ by $x$, at most two alternating paths might be found, one to the left of $x$, namely $(L_A,L_B)$ and one to the right, $(R_A,R_B)$.}
    \label{fig:insert_subset}
\end{figure}

The proof of correctness (that means, showing that after this single exchange of the algorithm, we get a $k$-valid independent set) requires a more careful and strict characterization of the alternating paths that we choose. 
The details are deferred to Section~\ref{sec:intervals_details}.

\paragraph{Deletions.} We now describe the deletion algorithm. Suppose interval $x \in I$ gets deleted. We check for alternating paths to the left and to the right of $x$. Let $L$ be the alternating path found in the left and $R$ the one found in the right (if no such path is found, set $L$ or $R$ to $\emptyset$). We then check if they can be merged, that is, if the two corresponding exchanges can be performed simultaneously (see Figure~\ref{fig:del_alternating_both}).

\begin{figure}[ht]
    \centering
    \includegraphics[scale=1]{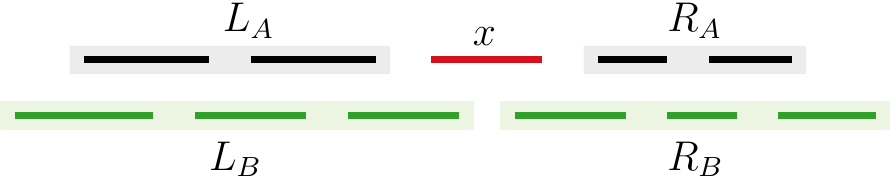}
    \caption{After deletion of interval $x$, alternating paths $L = (L_A,L_B)$ and $R = (R_A,R_B)$ are formed to the left and right of $x$ respectively. If they can be merged, we do both exchanges.}
    \label{fig:del_alternating_both}
\end{figure}

\begin{enumerate}
\item Both $L$ and $R$ are non-empty and they can be merged (Figure~\ref{fig:del_alternating_both}). We perform both exchanges. 
\item Both $L$ and $R$ are non-empty but cannot be merged. In this case perform only one of the two exchanges (details deferred to following sections). 
\item Only one of $L$ and $R$ are non-empty: Do this exchange.
\item Both $L$ and $R$ are empty. In this case, we check whether there exists an alternating path involving an interval $y$ containing $x$. If yes, then do the exchange. Otherwise do nothing.
\end{enumerate}

Again, proving correctness requires some effort. The important operation is to search for alternating paths starting from a point $x$, which can be done in time $O(k \log n)$. From this, the whole deletion algorithm can be implemented to run in time $O(k^2 \log n)$ in the worst case.


\subsection{Squares}
\label{sec:outline_squares}

Our presentation for how to maintain an approximate maximum dynamic independent set of squares is split into two sections. In Section~\ref{s:quadtree} we show how to do this statically, which is not new, but allows a clean presentation of our main novel ideas. In Section~\ref{sec:squares_dynamic}, we show how to make this dynamic, mostly using standard but cumbersome data structuring ideas.


We define a randomly scaled and shifted infinite quadtree. Associate each square with the smallest enclosing node of the quadtree. Call squares that intersect the center of their quadtree node \emph{centered} and discard all noncentered squares, see Figure~\ref{f:centered}. Nodes of the quadtree associated with squares are called the \textit{marked nodes} of the infinite quadtree, and call the quadtree $\Qtree$ the union of all the marked nodes and their ancestors. Note that multiple squares may be associated with one quadtree node.

\paragraph{High-level overview.} We will show that given a $c$-approximate solution for intervals, we can provide a $O(c)$-approximation for squares. To do that we proceed into a four-stage approach. We first focus on the static case and then discuss the modifications needed to support insertions/deletions. 

 \begin{enumerate}
     \item \label{item:centered_lose} We show that by losing a factor of $16$ in expectation, we can restrict our attention to centered squares (Lemma~\ref{l:bds}), thus we can indeed discard all non-centered squares.
    \item \label{item:path_lose} Then we focus on the quadtree $\Qtree$. We partition $\Qtree$ into leaves, internal nodes, and monochild paths, which will be stored in a compressed format. We show that given a linearly approximate solution for monochild paths of $\Qtree$, we can combine these solutions with a square from each leaf to obtain an $O(1)$-approximate solution for $\Qtree$ (Lemma~\ref{l:combine}). Roughly, if each monochild path our solution has size $\geq (1/d) \cdot \opt - \gamma $ (for some parameter $\gamma$), we get a $(2+d \cdot (\gamma+1))$-approximate solution for $\Qtree$. Thus, it suffices to solve the problem for squares stored in monochild paths. 
    
    To obtain intuition behind this, observe in Figure~\ref{f:thequad}(c) that each path has a pink region which corresponds to the region of the top quadtree node of the path minus the region of quadtree node which is the child of the bottom node of the path. We call a protected independent subset of the squares of path an independent set of squares that stays entirely within the protected region. All regions of protected paths and leaves (orange) are disjoint and thus their independent sets may be combined without risk of overlap.
This is what we do, we prove that a $O(1)$-approximate maximum independent set can be obtained with a single square associated with each leaf node and a linear approximate maximum protected independent set of each path. No squares associated with internal nodes of the quadtree form part of our independent set.

     \item \label{item:monochild_lose4} To obtain an approximate independent set in monochild paths, we partition each monochild path into four monotone subpaths, and show (Fact~\ref{f:max4}) that by loosing a factor of 4, it suffices to use only the independent set of monotone subpath with the independent set of maximum size.
     
     Let us see this a bit more closely. Figure~\ref{f:patha} illustrates such a path of length 30. Each node on the path has by definition only one child. The quadrant of a node is the quadrant where that single child lies.
   We partition the marked nodes of each path into four groups based on the quadrant's child, we call these monotone subpaths, each group is colored differently in the figure. 
   We observe that the the centers of the nodes on each monotone subpath are monotone.
   We proceed separately on each and use the one with largest independent set, losing a factor of four. 
   
     \item \label{item:monotone_lose2}  We show that independent set of centered squares in monotone subpaths reduces to the maximum independent set of intervals by losing roughly a factor of 2. More precisely, given a $c$-approximate solution for intervals, we can get a solution for monotone subpaths of size $\geq \frac{1}{2c} \opt - 1$.  (Lemma~\ref{l:indsetsize}).
     
     This is achieved as follows. As illustrated in Figure~\ref{f:pathb}, we associate each square on a monotone subpath with an interval which corresponds to the depth of its node to the depth of the deepest node on the subpath that it intersects the center of.
    For each subpath, we take the squares associated with the nodes of the subpath, and compute an independent set (red and orange intervals in the figure) with respect to the intervals associated with each square.
    We observe that while the set of squares in the previous step have independent intervals, the squares may nevertheless intersect, and may intersect the gray region, which would violate the protected requirement. However, only adjacent squares can intersect and thus by taking every other square from the independent set with respect to the intervals this new set of squares is independent with respect to the squares.
    
    By beginning this removal process with the deepest interval, the gray region in the figure, which is not part of the protected region, is also avoided. Observe the red set of squares in the figure is an independent set and avoids the gray region. 
 
 \end{enumerate}

\paragraph{Putting everything together.} Combining all those parts, we get that due to (\ref{item:monotone_lose2}), a $c$-approximate solution for intervals gives a solution for squares of monotone subpaths that is at least half the interval solution minus one. A factor of 4 is lost in the conversion from monotone paths to paths due to (\ref{item:monochild_lose4}), thus for monotone paths our solution has size $\frac{1}{8c} \opt - 1$. Consequently, due to (\ref{item:path_lose}), we get $d = 8c$ and $\gamma=1$ and this gives a $(2+16 \cdot c)$-approximation for centered squares. Finally, due to (\ref{item:centered_lose}), a $(2+16c)$-approximation for centered squares implies an $(2+16c)\cdot16=256c+32$-approximation for squares.


\begin{figure}[t]
\begin{minipage}[c]{0.4\textwidth}
    \includegraphics[width=\textwidth]{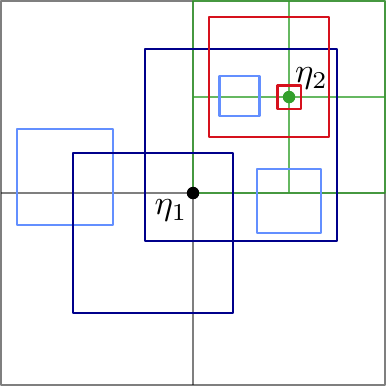}
    \end{minipage}\hspace{0.1\textwidth}
\begin{minipage}[c]{0.4\textwidth}
  \caption{Two quadtree nodes are illustrated, labelled at their center point. The dark blue squares are centered and have $\node_1$ as their node, the red squares are centered and have $\node_2$ as their node, and the light blue squares are not centered.} \label{f:centered}
\end{minipage}
\end{figure}



\begin{figure}[ht]
  \begin{center}
    \includegraphics[scale=.7]{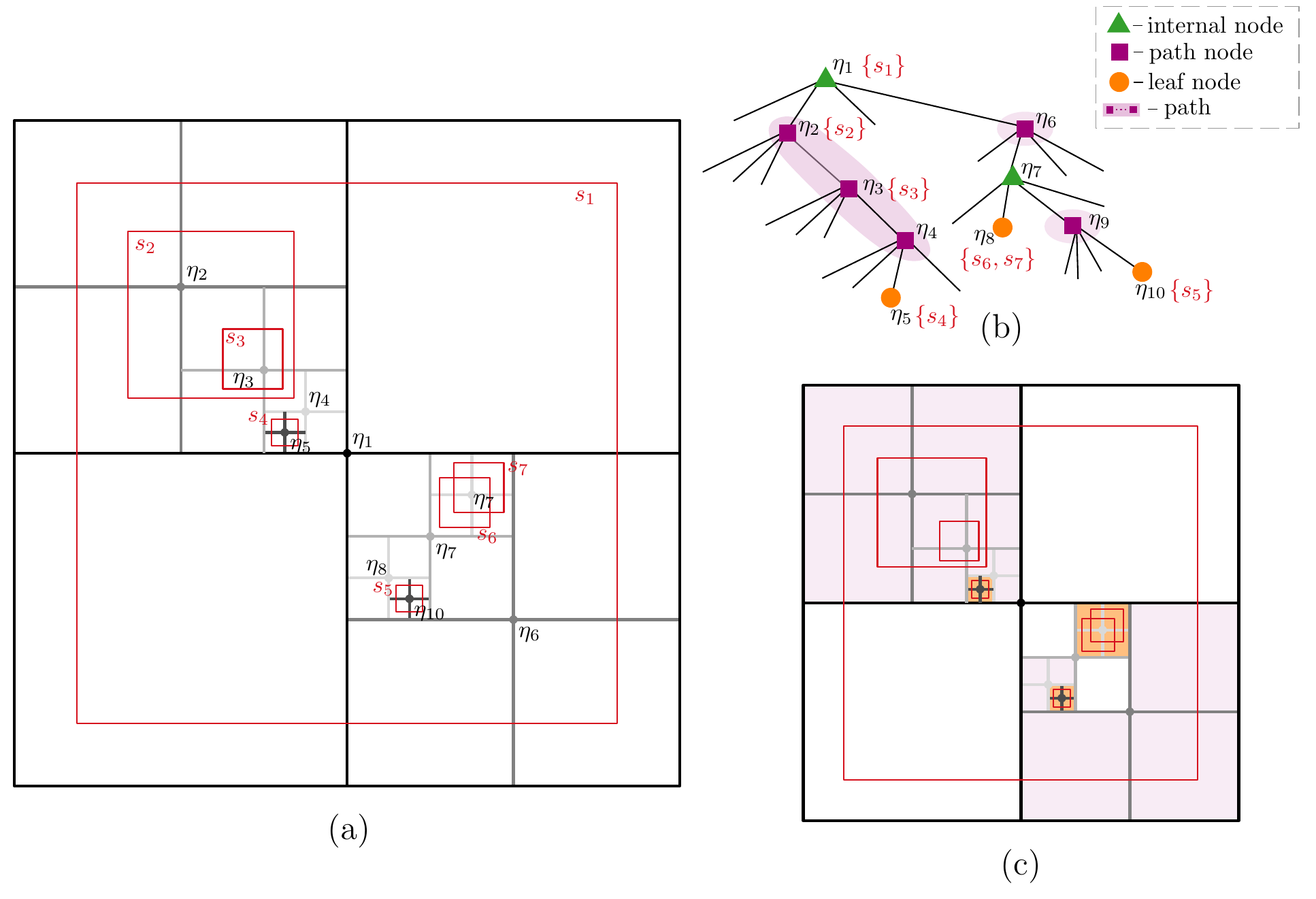}
  \end{center}
  \caption{The quadtree of the illustrated squares drawn normally (a) and as a tree (b). In (b) each node is categorized as a leaf, an internal node, or part of a monochild path. In (c) the protechted regions of each path and leaf are illustrated, and are pairwise disjoint.} \label{f:thequad}
\end{figure}




\begin{figure}
\begin{minipage}[c]{0.4\textwidth}
    \includegraphics[width=\textwidth]{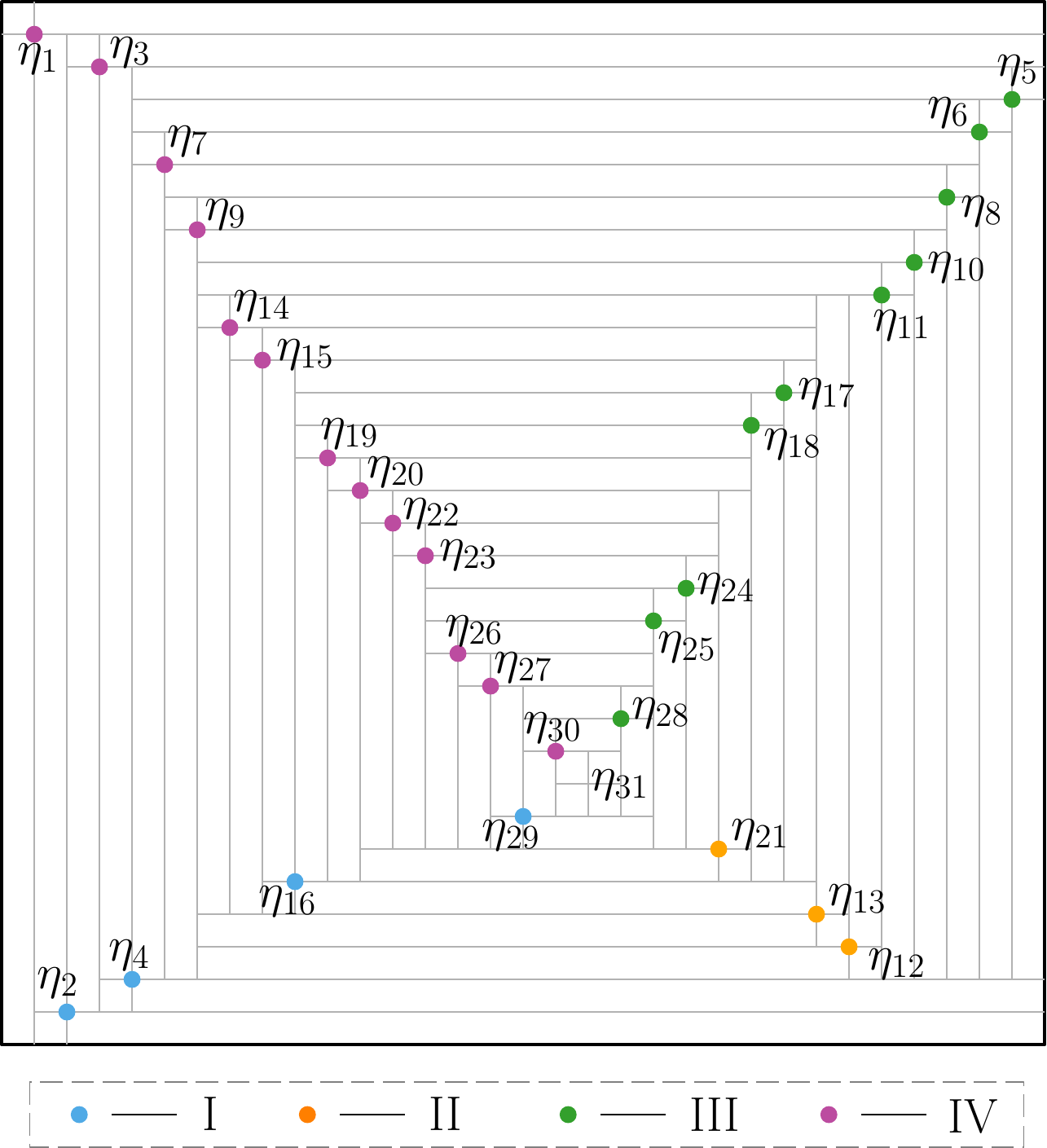}
\end{minipage} \hspace{0.15\textwidth}
\begin{minipage}[c]{0.4\textwidth}
\caption{A single monochild path in a quadtree is illustrated. Figure is not to scale, for example, if drawn to scale $\node_1$ would be in the center and the rest of the figure would be in the lower-right corner. The nodes on the path are labelled by depth, and fall into four groups based on which quadrant their child lies in. Crucially, each of these four groups is a monotone path.}
\label{f:patha}
\end{minipage}
\vspace{3pc}

\begin{minipage}[c]{0.45\textwidth}
    \includegraphics[width=\textwidth]{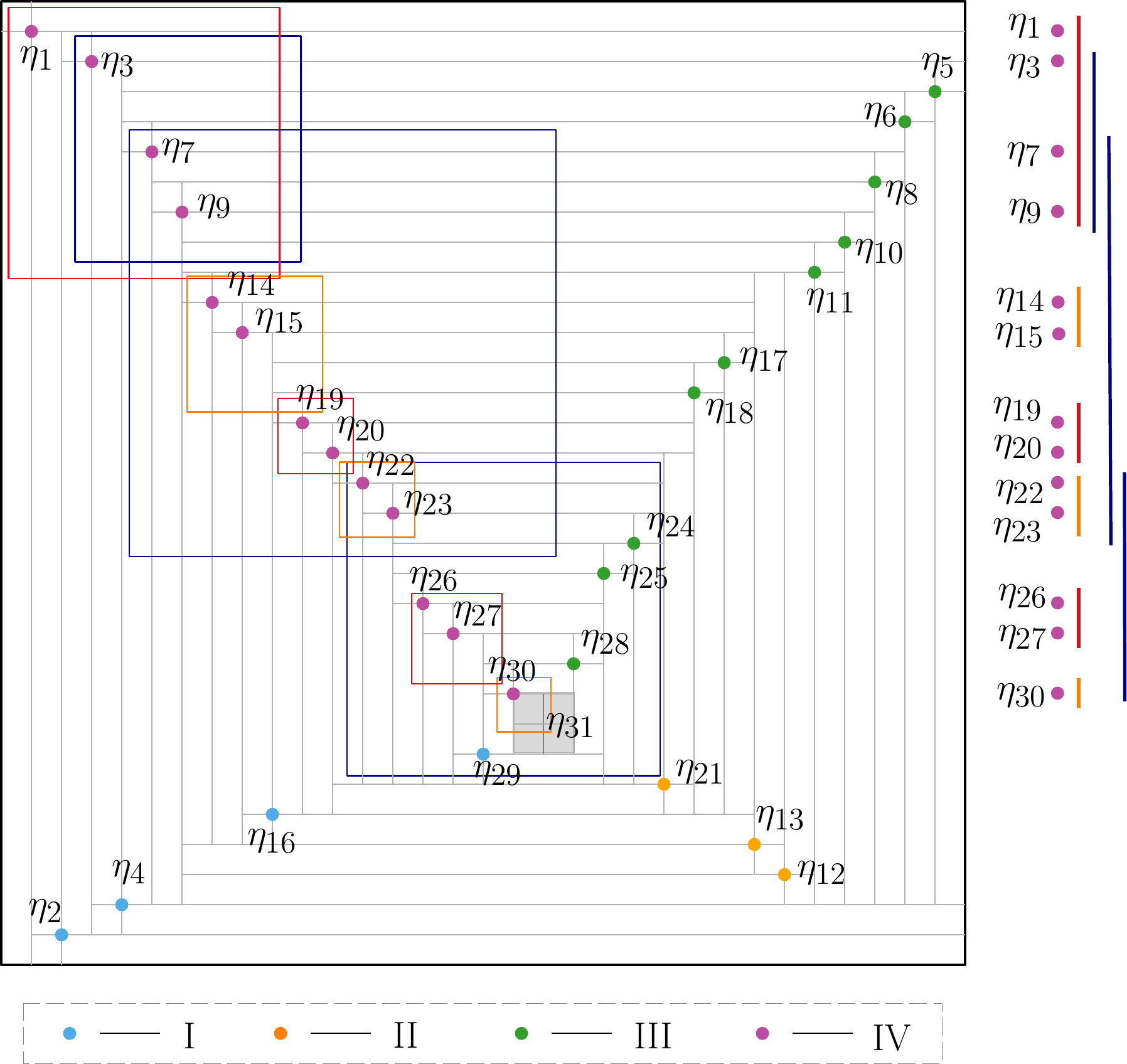}
\end{minipage}\hspace{0.1\textwidth}
\begin{minipage}[c]{0.4\textwidth}
\caption{Several squares are illustrated which are associated with the nodes of type $\qd$. They each are associated with an interval which spans the depths of the nodes of type $\qd$ that the contain the centers of. These intervals are drawn vertically on the right. Observe that the union of the orange and red intervals is an independent set of intervals, but the orange and red squares intersect. However, by taking every other interval of this union one obtains the red intervals which correspond to the red squares and which are disjoint. }
\label{f:pathb}
\end{minipage}
\end{figure}

\paragraph{Going Dynamic.} 
In order to make this basic framework dynamic we need a few additional ingredients, which are the subject of Section~\ref{sec:squares_dynamic}:

\begin{itemize}
    \item Use a link-cut structure \cite{DBLP:journals/jcss/SleatorT83} on top of the quadtree, as it is not balanced, this is needed for searching where to add a new node and various bulk pointer updates
    \item Use our dynamic interval structure within each path.
    \item Support changes to the shape of the quadtree, this can cause the paths to split and merge, and thus this may cause the splitting and merging of the underlying dynamic interval structures, which is why we needed to support these operations (see extensions of intervals, Section~\ref{sec:intervals_extend}).
    
    \item For the purposes of efficiency, all squares are stored in a four-dimensional labelled range query structure. This will allow efficient, $O(\log^5 n)$, computation of the local changes needed by the dynamic interval structure.
\end{itemize}

Those differences worsen the approximation ratio at some places.

\begin{itemize}
    \item In step~\ref{item:monochild_lose4} of the description above we said that for the static structure, we divide the monochild path into four monotone subpaths and by loosing a factor of 4, we pick among them the one whose maximum independent set has the largest size. However, in the dynamic case, this path might change very frequently; as the monotone subpath maximum independent set is unstable, we do not change from using the independent set from one subpath until it falls to being less than half of the maximum. This causes the running time bound to be amortized instead of worst-case and increases the bound by a factor of 2. That is, instead of losing a factor of 4 by focusing on monotone paths, we loose 8.
    \item In step~\ref{item:monotone_lose2} for the static case, we lose a factor of 2 due to picking every other square from the independent intervals. Dynamically we need more flexibility, we will ensure that there is between 1-3 squares between each one that was taken, and we show how a red-black tree can simply serve this purpose. Thus, given a $c$-approximate dynamic independent set of intervals structure (supporting splits and merges) we get a solution for monotone paths of size at least $4c \opt - 3$.
\end{itemize}

 Putting everything together in a similar way as in the static case, we get for monochild paths a solution of size at least $(1/32c) \opt - 3 $, i.e., having $d=32c$ and $\gamma=3$. Therefore due to step \ref{item:path_lose}, we get $(2+ d (\gamma+1))$-approximation. By replacing and using $c=2$ as the approximation factor for dynamic intervals (an easy upper bound on $1+\epsilon$) we get that our method maintains an approximate set of independent squares that is expected to be within a 4128-factor of the maximum independent set, and supports insertion and deletion in $O(\log^5 n)$ amortized time. 
 
 While 4128 seems large, it is simply a result of a combination of a steady stream of steps which incur losses of a factor of usually 2 or 4.
We note that we have chosen clarity of presentation over optimizing the constant of approximation, had we made the opposite choice, factors of two could be reduced to $1+\epsilon$. However, this is not the case everywhere, and the constant-factor losses having to do with using centered squares and not using any squares associated with internal nodes are inherent in our approach.

There is also nothing in our structure that would prevent implementation. It has many layers of abstraction, but each is simple, and probably the hardest thing to code would be the link-cut trees \cite{DBLP:journals/jcss/SleatorT83} if one could not find an implementation of this swiss army knife of operations on unbalanced trees (see \cite{DBLP:journals/jea/TarjanW09} for a discussion of the implementation issues in link-cut trees and related structures).

%% file: 2-dynamicIntervals.tex
\section{Dynamic Independent Set of Intervals}
\label{sec:intervals_details}

As it is clear from discussion of Section~\ref{sec:intervals_main}, in order to maintain a $(1-\epsilon)$-approximation of the maximum independent set, it suffices to maintain an independent set which (i) is $k$-maximal and (ii) satisfies the property that no interval is contained in an interval of the independent set. This latter property is referred to as the no-containment property. In this section we describe how to maintain dynamically such an  independent set of intervals subject to insertions and deletions.

In Section~\ref{sec:intervals_def}, we start by introducing all definitions and background which will be necessary to formally define and analyse our algorithm. The formal description of our algorithm and  data structures, as well the proof of running time is in Section~\ref{sec:intervals_alg}. The proof of correctness for our insertions/deletions algorithms, which is the most technical and complicated part is in Section~\ref{sec:intervals_cor}. In Section~\ref{sec:intervals_extend} we present some extensions of our results (maintaining a $k$-valid independent set under splits and merges) which will be used in Section~\ref{sec:squares_dynamic} to obtain our dynamic structure for squares.

\subsection{Definitions and Background}
\label{sec:intervals_def}

We now define formally alternating paths (described in Section~\ref{sec:intervals_main}) and introduce some necessary background on them. In particular we will focus on specific alternating paths, called proper, defined below. 

\subsection*{Alternating Paths}

Let $(A,B)$ be a pair of independent sets of $S$ of sizes $t$ and $t+1$ for some $0\leq t\leq k$ such that $(I\setminus A)\cup B$ is an independent set of $S$. 
Hence such a pair is a certificate that the independent set $I$ is not $k$-maximal.
We observe that any inclusionwise minimal such pair induces an alternating path: a sequence of pairwise intersecting intervals belonging alternately to $B$ and $A$. 


\begin{lemma}[Alternating paths]
  \label{lem:mes}
  Let $(A,B)\in {{\binom{I}{t}}\times {S\setminus {\binom{I}{t+1}}}}$ be a pair such that $(I\setminus A)\cup B$ is an independent set, and there is no $A'\subset A$ and $B'\subset B$ such that $(A',B')$ also satisfies the property.
  Then the set $A\cup B$ induces an alternating path of length $2t+1$ in the intersection graph of $I\cup B$.
\end{lemma}
\begin{proof}
  In what follows, we identify the intervals of $I$ and $B$ with the corresponding vertices in their intersection graph.
  If $(A,B)$ is inclusionwise minimal, then its vertices must induce a connected component.
  The intersection graph of $I\cup B$ is an interval graph with clique number 2, hence its connected components are caterpillars.

  First note that in the caterpillar induced by $(A,B)$, every vertex $a\in A$ has degree at most three.
  Indeed, if $a$ has degree four or more, then it must fully contain two intervals of $B$, yielding a smaller pair $(A',B')$ with $t=1$.
  The intervals of $A$ are linearly ordered. Let us consider them in this order.
  
  If the first interval $a\in A$ is adjacent to three vertices in $B$, say $b_1,b_2,b_3$, then the interval $b_2$ must be fully contained in $a$, and we can end the alternating path
  with $a$ and $b_2$, and remove all their successors. This yields a smaller $(A',B')$, a contradiction.

  If $a$ has degree one, then it must be the case that an interval of $A$ further on the right has degree three, since otherwise $|B|\leq |A|$.
  Pick the first interval $a'$ of $A$ of degree three, adjacent to $b_1,b_2,b_3$.
  The interval $b_2$ must be fully contained in $a'$.
  Hence a smaller $(A',B')$ can be constructed by removing all predecessors of $b_2$ and $a'$, a contradiction.

  Therefore, $a$ must have degree two, with neighbors $b_1,b_2$.
  The vertices $b_1$ and $a$ are the first two in the alternating path, and we can iterate the reasoning with the next interval of $A$, if any.
\end{proof}


We will refer to such pairs $(A,B)\in \binom{I}{t}\times {S\setminus \binom{I}{t}}$ as {\em inducing an alternating path with respect to $I$}. Note that we allow $t=0$, in which case the pair has the form $(\emptyset, \{x\})$ and the alternating path has length 1.

\begin{observation}
\label{obs:alt_a_no_containment}
If $(A,B)$ is an alternating path with respect to an independent set $I$, then no interval of $A$ is strictly contained in an interval of $B$. 
\end{observation}

Note that the inverse is not true in general. The leftmost and/or rightmost interval of $B$ might be strictly contained in an interval of $A$.

We focus on a particular class of alternating paths which we call \textit{smallest}.

\begin{definition}
\label{def:smallest-path}
An alternating path $(A,B)$ with respect to an independent set $I$ is called smallest if there is no alternating path $(A',B')$ such that $A' \subset A$.
\end{definition}

We make the following key observation.

\begin{lemma}
 \label{lem:pick}
  Consider a smallest alternating path induced by $(A,B)$, $A=\{a_1,\ldots ,a_t\}$, $B=\{b_1,\ldots ,b_{t+1}\}$ for some $0\leq t\leq k$, where the intervals in each set are indexed according to their order on the real line. Then every interval $b_i$ for $2\leq i\leq t$ can be assumed to be an interval with leftmost right endpoint among all intervals with left endpoint in the range $[r(b_{i-1}), r(a_{i-1})]$. Similarly, $b_1$ can be assumed to be an interval with leftmost right endpoint among all intervals with left endpoint in the range $[r(a'),r(a_1)[$, where $a'$ is the interval on the left of $a_1$ in $I$ if it exists, or in $]-\infty, r(a_1)[$ otherwise.
\end{lemma}
\begin{proof}
If the interval $b_i$ does not have the leftmost right endpoint, then we can replace it with one that has.
For $i\leq t$, this new interval must intersect $a_i$ as well, for otherwise the pair $(A,B)$ is not smallest. 
\end{proof}

Note that the symmetric is also true: If $(A,B)$ is a smallest alternating path, then there exists a smallest alternating path $(A,B')$ which satisfies the leftmost right endpoint property, i.e, $b'_i$ is an interval with the rightmost left endpoint among all intervals with right endpoint in the range $[\ell(a_i),\ell(b_{i+1})]$; the proof is identical to the proof of Lemma~\ref{lem:pick} above by flipping the terms left and right.

Using this observation, we can proceed to the following definition.

\begin{definition}
\label{def:proper}
An alternating path $(A,B)$ is called \textit{proper} if it is a smallest alternating path and it satisfies the leftmost right endpoint property.
\end{definition}

 Clearly, by the discussion above, given a proper alternating path $(A,B)$, there exists also a smallest alternating path $(A,B')$ which satisfies the rightmost left endpoint property. 
 
 \begin{definition}
 \label{def:sibling}
 Let $(A,B)$ be a proper alternating path. The smallest alternating path $(A,B')$ that satisfies the rightmost left endpoint property is called   \textit{sibling} of  $(A,B)$.  
 \end{definition}
 
 All swaps made by our insertions and deletion algorithms will involve solely proper alternating paths or their siblings.

\subsection{Algorithm and Data Structures}
\label{sec:intervals_alg}

Here we get more closely on the details of the algorithm presented informally in Section~\ref{sec:intervals_main}.

\subsubsection*{The Interval Query Data Structure}

We will use a data structure which supports standard operations like membership queries, insert and delete in time $O(\log n)$. Moreover we need to answer queries of the following type: 
Given $a,b$, find an interval having the leftmost right endpoint, among all intervals whose left endpoint lies in the range $[a,b[$.
We refer to these queries as {\em leftmost right endpoint} queries. The symmetric queries (among all intervals whose right endpoint lies in $[a,b[$, find the one having the rightmost left endpoint) are referred as {\em rightmost left endpoint queries}.

\begin{lemma}[Interval Query Data Structure (IQDS)]
  \label{lem:ors}
  There exists a data structure storing a set of intervals $S$ and supporting:
  \begin{itemize}
  \item {\bf Insertions and deletions}: Insert an interval $x$ in $S$/ delete an interval $x$ from $S$.
  \item {\bf Leftmost right endpoint queries.} $\jop{Report-Leftmost}(a,b)$: Among intervals $y$ with $\ell(y) \in (a,b)$, report the one with the leftmost right endpoint (or return NULL).
  \item {\bf Rightmost left endpoint queries.} $ \jop{Report-Rightmost}(a,b)$: Among all intervals $y$ with $r(y) \in (a,b)$, report the one with the rightmost left endpoint.
  \item {\bf Endpoint Queries.} Given an interval $x$, return its left and right endpoints.
  \item {\bf Merge}:  Given two such data structures containing sets of intervals $S_1$ and $S_2$, and a number $t$ such that $\ell(s)\leq t$ for all $s\in S_1$
  and $\ell(s)> t$ for all $s\in S_2$, construct a new data structure containing $S_1\cup S_2$,
  \item {\bf Split}: Given a number $t$, split the data structure into two, one containing $S_1 :=\{s: \ell (s)\leq t\}$, and one containing $S\setminus S_1$,
  \end{itemize}
  in $O(\log n)$ time per operation in the worst case.
\end{lemma}

\begin{proof}
We resort to augmented red-black trees, as described in Cormen et al.~\cite{CLRS09}. 
The keys are the left endpoints of the intervals, and we maintain an additional information at each node: the value of the leftmost endpoint of an interval in the subtree rooted at the node.
This additional information is maintained at a constant overhead cost.
Leftmost right endpoint queries are answered by examining the $O(\log n)$ roots of the subtrees corresponding to the searched range.
The structure can be duplicated to handle the symmetric rightmost left endpoint queries.
\end{proof}

\paragraph{Remarks.} Before proceeding to presenting our algorithms using the data structure of Lemma~\ref{lem:ors}, we make some remarks:

\begin{enumerate}
    \item In fact our data structure can be implemented in a comparison-based model where the only operations allowed are comparisons between endpoints of intervals. In particular, leftmost right endpoint queries (and symmetrically rightmost left endpoint queries) are used only for $a$ and $b$ being endpoints of intervals of $S$. Here, we present them as getting as input arbitrary coordinates just for simplicity of exposition. 
    
    \item For the context of this section, it is sufficient to use augmented red-black trees to support those operations in time $O(\log n)$. However, later we would need to use the intervals data structure as a tool to support independent set of squares, this will not be enough. The details will be described in Section~\ref{sec:squares_dynamic}.
    
    \item The split and merge operations are only needed to make our extensions to squares work (see Section~\ref{sec:intervals_extend}). The reader interested in intervals may ignore them.
\end{enumerate}

We will maintain two such data structures, one for storing the set of all intervals $S$ and one storing the current independent set $I$. 

\paragraph{Alternating paths in time $O(k \log n)$}
We show that using such a data structure, we can find alternating paths of size at most $k$ in time $O(k \cdot \log n)$. In particular, we are going to have the following procedure:
\begin{itemize}
    \item  $\jop{Find-Alternating-Path-Right} (I, k,(a,b)$): Find an alternating path, with respect to the independent set $I$, of size at most $k$, where the leftmost interval has left endpoint in $(a,b)$. This alternating path will satisfy the leftmost right endpoint property. 
\end{itemize}

The other, completely symmetric, procedure $\jop{Find-Alternating-Path-Left}$, does the same thing, only with left and right (and left endpoints and right endpoints) reversed. 
It therefore suffices to describe only $\jop{Find-Alternating-Path-Right}$. We let $A\gets\emptyset$ and $B\gets \emptyset$, and proceed as follows. Let $\NEXT$ be the leftmost interval of $I$ to the right of $b$ (if exists). If $\NEXT = \NULL$, let $r(\NEXT) = \infty$. 

\begin{enumerate}
  \item \label{step:min} Among all intervals in $S\setminus I$ with left endpoint in $[a, b[$, if any, let $y$ be the one such that $r(y)$ is minimum.
  \item If such a $y$ exists, then:
      \begin{enumerate}
      \item If $r(y) < \ell(\NEXT)$ then $B \leftarrow B \cup \lrbrace{y}$. Return $(A,B)$.
      
      \item If
        \begin{itemize}
        \item $\ell(\NEXT) \leq r(y)<r(\NEXT)$, 
        \item {\em and} $|A|<k$,
        \end{itemize}
        then $A \leftarrow A \cup \lrbrace{\NEXT}$, $B = B \cup \lrbrace{y}$ and iterate from step~\ref{step:min} with $(a,b)$ replaced by $(r(y),r(\NEXT))$ and $\NEXT$ replaced by the first interval of $I$ that follows on its right.
      \item Otherwise return fail.
      \end{enumerate}
      \item Otherwise return fail.
\end{enumerate}

By construction, this procedure performs at most $k$ iterations where in each iteration the only operations required are leftmost right endpoint queries and finding the next interval in independent set $I$, which can both be done in time $O(\log n)$. Therefore the overall running time is always $O(k \log n)$.

\paragraph{Some auxiliary operations} Sometimes we might need to transform an alternating path $(A,B)$ satisfying the leftmost right endpoint property to another path $(A,B')$ (possibly with $B = B' $) which satisfies the rightmost left end point property (or vice versa). We show that our data structure supports this in time $O(|A| \cdot \log n)$.

Let $I$ be a $k$-maximal independent set and let $A = \lrbrace{a_1,\dotsc,a_t}$ and $B = \lrbrace{b_1,\dotsc,b_{t+1}}$, such that $(A,B)$ is a smallest alternating path satisfying the rightmost left endpoint property. We will show how to transform it into a path $(A,B')$ satisfying the rightmost left endpoint property. 

The main idea is to start from $b_1$ and for all $i=1,\dotsc,t+1$, replace $b_i$ with another interval $b'_i$ which intersects both $a_{i-1}$ and $a_i$ and has the leftmost right endpoint property. 

Let $x \in I$ be the interval of $I$ to the left of $a_1$ (if any). We start by finding the interval with the leftmost right endpoint, among all intervals with left endpoint in $(r(x), \ell(a_1) )$ (set $r(x)$ to -$\infty$ if $x$ does not exist). This interval will be $b'_1$. Note that it might be possible that $b'_1 = b_1$. We continue in the same way for all $i \leq t+1$. Once interval $b'_{i-1}$ is fixed we answer the query Report-Leftmost$(r(b'_{i-1}),r(a_{i-1}))$ and the outcome will be the new interval $b'_i$. Overall we answer $t$ leftmost right endpoint queries, thus the total running time is $O(t \cdot \log n)$.

Note that in the algorithm above, all leftmost right endpoint queries will return for sure an interval and will never be NULL; this is because the interval $b_i$ satisfies the requirements, so there exists at least one interval to report. Moreover, there is the possibility that in step $i$, the interval $b'_i$ ends before interval $a_i$ starts. We will make sure that our algorithms use this procedure in instances which this does not happen (proven in Lemmata~\ref{lem:higher_exchange} and~\ref{lem:j-to-j}). 
\subsubsection*{Description of Algorithms}

We now describe our algorithms in pseudocode using our data structure and the operations it supports. 

Whenever we use $L$ or $R$ to denote alternating paths, we implicitly assume that those are defined by sets $(L_A,L_B)$, such that $L_A \subseteq I$ and $L_B \subseteq S \setminus I$ (resp. $(R_A,R_B)$). Whenever we say that we perform the exchange defined from alternating path $L$ (resp. $R$) we mean that we set $I \leftarrow (I \setminus L_A) \cup L_B$ (resp. $ I \leftarrow (I \setminus R_A) \cup R_B $). 

\paragraph{Insertions.} Interval $x$ gets inserted. Let $a_{\ell}$ be the interval of $I$ containing $\ell(x)$ (NULL if such interval does not exist) and $a_r$ the one containing $r(x)$.

\begin{enumerate}
\item If both $a_\ell$ and $a_r$ are NULL, then
\begin{enumerate}
\item \label{case:insert-one} If no interval of $I$ lies between $\ell(x)$ and $r(x)$ (that is, $x$ can be added), then $I \leftarrow I \cup \lrbrace{x}$.
\end{enumerate}

\item If both $a_\ell$ and $a_r$ are defined, then:
\begin{enumerate}
\item \label{case:strict_contain} If $a_{\ell} = a_r$, hence if $x$ is strictly contained in interval $a := a_\ell=a_r \in I$, then:

\begin{itemize}
\item Replace $a$ by $x$: $I \leftarrow (I \setminus \lrbrace{a})\cup x$.
\item $R \leftarrow \jop{Find-Alternating-Path-Right} (I,k,(r(x),r(a))$. If $R \neq \emptyset$, do this exchange.
\item $L \leftarrow \jop{Find-Alternating-Path-Left} (I,k,(\ell(a),\ell(x))$. If $L \neq \emptyset$, do this exchange.
\end{itemize}

\item \label{case:insert_alt_both} If $a_{\ell}$, $a_r$ are two consecutive intervals of $I$, then try to find an alternating path containing $x$:

\begin{itemize}
\item   $R \leftarrow \jop{Find-Alternating-Path-Right} ((I,k-2,(r(x),r(a)))$. 
\item If $R \neq \emptyset$, then set $L \leftarrow \jop{Find-Alternating-Path-Left} ((I,k-2-|R|,(\ell(x),\ell(a))$). \\ If both $L= (L_A,L_B)$ and $R = (R_A,R_B)$ are nonempty, then:

\begin{itemize}
\item Set $A \leftarrow L_A \cup \lrbrace{a_{\ell},a_r} \cup R_A$ and $B \leftarrow L_B \cup \lrbrace{x} \cup R_B$. $(A,B)$ is an alternating path of size at most $k$. Do this exchange.
\end{itemize}  
\end{itemize}

\end{enumerate}

\item \label{case:insert_alt_one} If only $a_r$ exists (the case where only $a_{\ell}$ exists is symmetric), then try to find an alternating path of size at most $k-1$ to the right:

\begin{itemize}
\item $R \leftarrow \jop{Find-Alternating-Path-Right} ((I,k-1,(r(x),r(a))$). 
\item If $R= (R_A,R_B)$ non empty, then set $A \leftarrow R_A \cup \lrbrace{a_r}$, $B \leftarrow R_B \cup \lrbrace{x}$. \\ $(A,B)$ is an alternating path. Do this exchange. 
\end{itemize} 
\end{enumerate}

\paragraph{Deletions.}Interval $x$ gets deleted. If $x \notin I$, which can be checked in time $O(\log n)$, then we do nothing. So we focus on the case $x \in I$.  Let $a_{\ell}$ be the interval of $I$ to the left of $x$ (if it exists) and $a_r$ the interval to the right of $I$ (if it exists). We first delete $x$ and then search for alternating paths to the right and left of $x$:

\medskip

\noindent $R \leftarrow \jop{Find-Alternating-Path-Right} ((I,k,(r(a_{\ell}),\ell(a_r)))$). 

\medskip

\noindent $L \leftarrow \jop{Find-Alternating-Path-Left} ((I,k,(r(a_{\ell}),\ell(a_r)))$). 

\medskip

\noindent $L$ has the rightmost left endpoint property. If nonempty, we replace $L$ by its sibling which satisfies the leftmost right endpoint property, as explained above. 

\begin{enumerate}
\item If $L$ and $R$ are nonempty, then check whether they can be merged, that is, whether the right endpoint of the rightmost interval of $L_B$, say $r(L)$, is to the left of the left endpoint of the leftmost interval of $R_B$, $r(L) < \ell(R)$.

\begin{enumerate}
\item \label{case:del-both-sides} If yes, then do the exchanges defined by $R$ and $L$.
 
\item \label{case:del-both-one} Otherwise do the exchange defined either from $L$ or $R$ (arbitrarily)
\end{enumerate} 

\item \label{case:del-one-side} If only one of $L$ and $R$ is nonempty, do this exchange.

\item If both $L$ and $R$ are empty, then search for an alternating path including an interval $y$ containing $x$: 
$y \leftarrow \jop{Report-Leftmost}(r(a_{\ell}),\ell(x))$) ($y$ contains $x$).
\begin{enumerate}
\item \label{case:del-superset-one} If $r(y) < \ell(a_r)$ ($y$ can be added), then $I \leftarrow I \cup \lrbrace{y}$.
\item \label{case:del-superset-path} Otherwise, check for alternating paths including intervals strictly containing $x$ (if any): Let $a_1,\dotsc,a_{k}$ be the $k$ intervals of $I$ to the right of $x$, ordered from left to right (note $a_1=a_r$). If some interval does not exist, set it to $\NULL$. Let also $a_0 = x$. For $i=1$ to $k$, use $\jop{Find-Alternating-Path-Left} (I,k,(r(a_{i-1}),\ell(a_i)))$ to search to the left for an alternating path of length at most $k$. Whenever a path $(A,B)$ is found, do this exchange and stop.
\end{enumerate}
\end{enumerate}

\paragraph{Running time.} It is easy to see that  for insertion all operations used require time $O(k \log n)$ and for deletion $O(k^2 \log n)$; this increase in deletion time comes solely due to case (\ref{case:del-superset-path}) where we need to search at most $k$ times for alternating paths of size at most $k$, which requires $O(k \log n)$ time.

\subsection{Correctness}
\label{sec:intervals_cor}

We now prove correctness of our algorithms. Recall that by Definition~\ref{def:valid} a $k$-valid independent set of intervals is $k$-maximal and satisfies the no-containment property. We show that our algorithms always maintain a $k$-valid independent set of intervals.

\paragraph{Some easy observations.}
We begin with some easy, yet useful, observations.

\begin{observation}
\label{obs:insert_one}
Let $I$ be a $k$-valid independent set of $S$. If an interval $x$ gets inserted such that $I \cup \lrbrace{x}$ is an independent set, then $I \cup \lrbrace{x}$ is also $k$-valid. 
\end{observation}

\begin{observation}
\label{obs:delete_notinset}
If an interval $x \notin I$ gets deleted, then $I$ remains $k$-valid.
\end{observation}

\begin{observation}
\label{obs:del_superset}
Let $I$ be a $k$-valid independent set. Assume that for an interval $x \in I$ there exists $y \in S$, such that $y$ contains $x$ and $(I \setminus \lrbrace{x}) \cup \lrbrace{y}$ is also an independent set. Then, $(I \setminus \lrbrace{x}) \cup \lrbrace{y}$ is $k$-maximal.
\end{observation}

\paragraph{Main Technical Lemmas.} It turns out that the most crucial technical part of our approach is the following two lemmas, which are used both to the insertion and deletion algorithms. The first lemma has to do with $j$ to $j+1$ exchanges and the second with the $j$ to $j$ exchanges.

\begin{lemma}
\label{lem:higher_exchange}
Let $I$ be a $k$-valid independent set of intervals. Assume there exists a proper alternating path $(A,B)$, such that $|A| = t$, $|B|=t+1$ for $k < t \leq 2k+1$. Then, $(I \setminus A) \cup B$ is also a $k$-valid independent set. 
\end{lemma}


To state the second lemma we need the following definition. 

\begin{definition}
Let $I$ be a $k$-valid independent set. A set $B = \lrbrace{b_1, \dotsc, b_j} \subseteq S \setminus I$ is called a  left/right \textit{substitute} of a set $A = \lrbrace{a_1,\dotsc,a_j} \subseteq I$ if the following  holds: 

\begin{enumerate}
 \item There is no way to extend $A$ and $B$ to create alternating paths of size $t$ to $t+1$, for any  $t \leq 2k$.
 \item Left substitute: If interval $a_j$ was not there, then $(A \setminus \lrbrace{a_j},B)$ would be a proper alternating path. Symmetrically for right substitute, if $a_1$ was not there, then $(A \setminus \lrbrace{a_1},B)$ would be a proper alternating path. 
\end{enumerate}
\end{definition}

Another important lemma, concerning exchanges with the same number of intervals.

\begin{lemma}
\label{lem:j-to-j}
Let $I$ be a $k$-valid independent set. Let $A \subseteq I$ and $B$ a (left or right) substitute of $A$. Then, $(I \setminus A) \cup B$ is a $k$-valid independent set.
\end{lemma}

Proofs of lemmata~\ref{lem:higher_exchange} and~\ref{lem:j-to-j} are deferred to the end of this subsection. We first show how they can be combined with the observations above to show correctness of our dynamic algorithm. 

\vspace{0.16cm}

\noindent \textbf{Correctness of the Insertion Algorithm.}  We need to perform a case analysis depending on the change made by our algorithm after each insertion. However, in all cases our approach is the same: we show that the overall change is equivalent to (i) either a $j$-to-$j+1$ exchange for $j \leq 2k+1$ or a $j$-to-$j$ substitution before insertion of $x$, plus (ii) adding $x$ in the independent set. The resulting independent set remains valid after step (i) due to Lemma~\ref{lem:higher_exchange} or~\ref{lem:j-to-j} respectively and after step (ii) using observation~\ref{obs:insert_one}.

 We now begin the case analysis. First, observe that we need only to consider the case where the algorithm performs exchanges. If no exchanges are made, then it is easy to see that $I$ remains $k$-valid: both $k$-maximality and no-containment can only be violated due to $x$ and if this is the case we fall into one of the cases where the algorithm makes changes. Thus we assume that the algorithm does some change.

 In case the new interval $x$ does not intersect any other interval of $I$ and gets inserted (case~\ref{case:insert-one} of the algorithm) then the new independent set is $k$-valid due to Observation~\ref{obs:insert_one}. In case where the inserted interval $x$ is strictly contained in an interval $a \in I$, which corresponds to case~\ref{case:strict_contain} in the pseudocode of Section~\ref{sec:intervals_alg}  (case 1 in the description of Section~\ref{sec:intervals_main}),  three subcases might occur:

\begin{enumerate}
\item Alternating paths were found in both directions: $L = (L_A,L_B) $ and $R = (R_A,R_B)$ (see Figure~\ref{fig:insert_subset}). Let $A = L_A \cup \lrbrace{a} \cup R_A$ and $B = L_B \cup R_B$. Observe that $(A,B)$ is an alternating path of size $j \leq 2k+1$ in the intersection graph of $S$ before the insertion of $x$. Thus, the overall change is equivalent to (i) doing a $j$ to $j+1$ exchange in the previous graph, for $j \leq 2k+1$, then (ii) adding $x$. Thus using Lemma~\ref{lem:higher_exchange} and Observation~\ref{obs:insert_one}, we get that the new independent set is $k$-valid.

\item An alternating path was found only in one direction: Assume that it is found only to the left, i.e., $L = (L_A,L_B) \neq \emptyset$ and $R = \emptyset$ (see Figure~\ref{fig:insert_1b}). Note that before $x$ was inserted, $L_B$ was a left substitute of $L_A \cup \lrbrace{a}$. Thus the overall change made from our algorithm is equivalent to (i) performing a substitution of $L_A \cup \lrbrace{a} $ by $L_B$ in the previous graph, then (ii) adding $x$ in $I$. $I$ remains $k$-valid after step (i) due to Lemma~\ref{lem:j-to-j} and after step (ii) due to Observation~\ref{obs:insert_one}.

\item No alternating path is found neither to the left nor to the right: $L = R = \emptyset$. Here it is easy to show that the new independent set is $k$-valid; clearly it satisfies the no-containment property. It remains to show the $k$-maximality. Assume for contradiction that there exists an alternating path $(A,B)$ of size at most $k$; this alternating path should involve $x$ (otherwise $I$ was not $k$-maximal which contradicts the induction hypothesis) and since $x$ is subset of $a$, it should involve $a$, thus it was an alternating path before insertion of $x$, contradiction.
\end{enumerate}    

\begin{figure}[ht]
    \centering
    \includegraphics[scale=1]{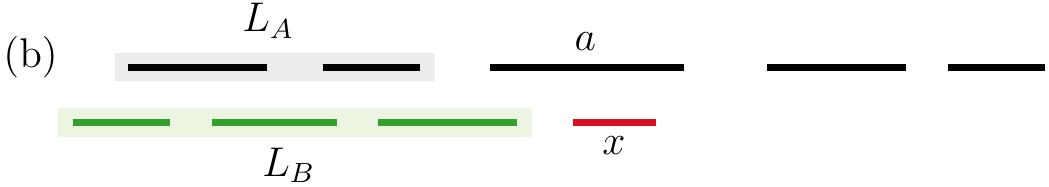}
    \caption{$x$ is contained in $a$, an alternating path $(L_A,L_B)$ is found.}
    \label{fig:insert_1b}
\end{figure}

It remains to show correctness for the cases where an alternating path involving $x$ is found and an exchange is made, that is, cases \ref{case:insert_alt_both} and \ref{case:insert_alt_one} of the insertion algorithm. The two cases are similar. In case \ref{case:insert_alt_both} (an alternating path extends both to the left and to the right of $x$ --see Figure~\ref{fig:insertion_alternating_both}), let $(L_A,L_B)$ be the alternating path found in the left and $(R_A,R_B)$ the one found at the right of $x$. Note that before insertion of $x$, $L_B$ was a left substitute of $L_A \cup \lrbrace{a_\ell}$ and $R_B$ was a right substitute of $ R_A \cup \lrbrace{a_r} $. Thus the the overall change made by the algorithm, removing $L_A \cup \lrbrace{a_\ell} \cup R_A \cup \lrbrace{a_r}$ from $I$ and adding $L_B \cup \lrbrace{x} \cup R_B$, is equivalent to (i) performing two substitutions before insertion of $x$ and (ii) adding $x$; thus by Lemma~\ref{lem:j-to-j} and Observation~\ref{obs:insert_one} we get that the new independent set is $k$-valid. In case~\ref{case:insert_alt_one}, the analysis is the same, just $L$ or $R$ is empty and $a_\ell$ or $a_r$ respectively is null, thus the same arguments hold. 

\begin{figure}[ht]
    \centering
    \includegraphics[scale=1]{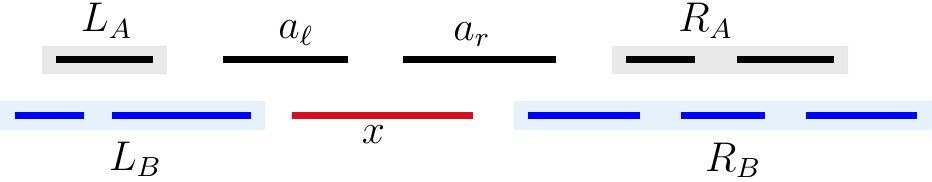}
    \caption{After insertion of $x$, an alternating path $(A,B)$ is formed, where $A = L_A \cup \lrbrace{a_{\ell},a_r} \cup R_A$ and $B = L_B \cup \lrbrace{x} \cup R_B $}
    \label{fig:insertion_alternating_both}
\end{figure}

\paragraph{Correctness of the Deletion Algorithm.} Recall that if the deleted interval $x$ is not in the current independent set $I$, then we do not make any change and by Observation~\ref{obs:delete_notinset} $I$ remains $k$-valid. So we focus on the case $x \in I$. Same as for insertion, it is easy to show that is the algorithm does not make any change other than deleting $x$, then $I$ remains $k$-maximal. We proceed to a case analysis, assuming algorithm did some change.  

\begin{enumerate}
    \item Alternating paths found in both directions and they can be merged (Case \ref{case:del-both-sides} of the deletion algorithm). In this case we have two alternating paths $L = (L_A,L_B)$ and $R = (R_A,R_B)$ (see Figure~\ref{fig:del_alternating_both}) with $|L_A|, |R_A| \leq k$. Let $A = L_A \cup \lrbrace{x} \cup \lrbrace{R_A}$ and $B = L_B \cup R_B$. Observe that, before deletion of $x$, $(A,B)$ was an alternating path of size $|L_A| + |R_A| +1 \leq 2k+1$.  The exchange made by our algorithm (deleting $x$, removing $L_A, R_A$ from $I$ and adding $L_B, R_B$ to $I$) is equivalent to performing the exchange $(A,B)$ before deletion of $x$; then when $x$ is deleted, $I$ is not affected (by Observation~\ref{obs:delete_notinset}). By Lemma~\ref{lem:higher_exchange} we get that the new independent set is $k$-valid.

    \item An exchange is performed only to the left (right) of $x$ (cases \ref{case:del-both-one} and \ref{case:del-one-side} of the deletion algorithm). We show the case of left; the one for right is symmetric. $L = (L_A,L_B)$ is an alternating path on the left of $x$. Note that before deletion of $x$, $L_B$ was a left substitute of $L_A \cup \lrbrace{x}$. Thus the performed exchange is equivalent to a (i) performing a $j$-to-$j$ substitution before the deletion of $x$, then (ii) deleting $x$ from $S$. In step (i) we remain $k$-valid due to Lemma~\ref{lem:j-to-j} and in step (ii) due to Observation~\ref{obs:delete_notinset}.
    \item If an interval $y$ containing $x$ gets added to $I$ after deletion of $x$ (case \ref{case:del-superset-one} of deletion algorithm), then clearly $I$ satisfies the no-containment property: this is because the interval $y$ we use is the one with leftmost right endpoint among intervals containing $x$. Moreover, the new independent set $I' = I \setminus \lrbrace{x} \cup \lrbrace{y}$ is $k$-maximal due to Observation~\ref{obs:del_superset}. Overall, $I'$ is a $k$-valid independent set. 
    \item In case we find an alternating path $(A,B)$ including an interval $y$ containing $x$ (case \ref{case:del-superset-path} of deletion algorithm), let $L_A \subseteq A$ be the intervals of $A$ to the left of $x$ and $R_A \subseteq A$ the ones to the right of $x$. Similarly let $L_B$ and $R_B$ be the intervals of $B$ to the left/right of $y$ (see Figure~\ref{fig:del_superset_alter}). Note that before the deletion of $x$,  $L_B$ is a left substitute of $L_A$ and $R_B$ is a right substitute of $R_A$.  Thus the exchange made by the deletion algorithm is equivalent to (i) substituting $L_A$ by $L_B$ and $R_A$ by $R_B$ before the deletion of $x$ and (ii) after the deletion of $x$ replacing it by $y$. After the substitutions of step (i) we remain $k$-valid due to Lemma~\ref{lem:j-to-j} and for step (ii) we use Observation~\ref{obs:del_superset}.   
\end{enumerate}


\begin{figure}[ht]
    \centering
    \includegraphics[scale=1]{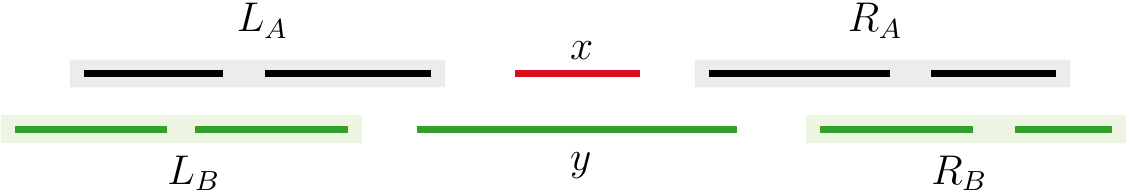}
    \caption{Case \ref{case:del-superset-path} of deletion algorithm: After deletion of $x$, a new alternating path $(A,B)$ is formed with $A = L_A \cup R_A$ and $B = L_B \cup \lrbrace{y} \cup R_B$.}
    \label{fig:del_superset_alter}
\end{figure}

\paragraph{Missing proofs.} In the remainder of this section, we give the full proofs of lemmata~\ref{lem:higher_exchange} and~\ref{lem:j-to-j} which were omitted earlier.

\medskip

\noindent {\bf  Lemma~\ref{lem:higher_exchange}}  {\em (restated) Let $I$ be a $k$-valid independent set of intervals. Assume there exists a proper alternating path $(A,B)$, such that $|A| = t$, $|B|=t+1$ for $k < t \leq 2k+1$. Then, $(I \setminus A) \cup B$ is also a $k$-valid independent set.}

\begin{proof}
We prove the lemma by contradiction. First we show that the no-containment property is true for $(I \setminus A) \cup B$. We then proceed to $k$-maximality. All proofs are shown using a contradiction argument. Let $A = \lrbrace{a_1,\dotsc,a_{t}}$ and $B = \lrbrace{b_1,\dotsc,b_{t+1}}$.

\medskip

\noindent \textit{No containment:} Assume for contradiction that there exists an interval $y \in S \setminus ((I \setminus A) \cup B)$ that is strictly contained in an interval $x \in (I \setminus A) \cup B$. Clearly, $x \in B$, since for all intervals of $I \setminus A$, the no-containment property is true (because $I$ is $k$-valid). Moreover, $y \notin A$, since $(A,B)$ is an alternating path, thus by Observation~\ref{obs:alt_a_no_containment} no interval of $A$ is strictly contained in an interval of $B$. Thus $y \in S \setminus (I \cup B) $. Overall, we have $y \in S \setminus (I \cup B) $ and $x \in B$ such that $y$ is strictly contained in $x$. 

 Let $i$ be the integer such that $b_i=x$; clearly $1 \leq i \leq t+1$. There are 4 cases to consider depending on how $y$ intersects with $a_{i-1}$ and $a_i$\footnote{Corner cases: If $i=1$ then $a_{i-1} = a_0$ does not exist; similarly, if $i=t+1$, then $a_i = a_{t+1}$ does not exist. In case an interval $a_{i-1}$ or $a_i$ does not exist, we simply assume that it exists and does not intersect with $y$.}, illustrated in Figure~\ref{fig:no-containment}.
\begin{enumerate}
\item $y$ does not intersect with none of $a_{i-1}$,  $a_i$. This contradicts $k$-maximality of $I$, since $I \cup \lrbrace{y}$ would be an independent set.

\item $y$ strictly contained in $a_{i-1}$ or $a_i$. This contradicts the fact that $I$ is a $k$-valid independent set (no-containment violated).

\item $y$ intersects only one interval of $A$ but it is not strictly contained in it. Assume it intersects $a_{i-1}$;  we construct a contradicting alternating path from left to right (proof for $a_i$ will be symmetric, i.e., constructing a contradicting alternating path from right to left). Then, $b_1, a_1, \dotsc,b_{i-1}, a_{i-1}, y$ is an alternating path of size $i-1 \leq t$, contradicting that $(A,B)$ is a smallest alternating path. Corner case: if $i = t+1$, then this contradiction does not hold: the new alternating path has the same size. But in that case, $r(y) < r(x)$, meaning that $x= b_i$ is not the interval with the leftmost right endpoint that could be added in the alternating path, thus $(A,B)$ is not proper. Contradiction. For the case $y$ intersects $a_i$, the same corner case appears if $i=1$; same way, this will contradict that $(A,B)$ satisfies the leftmost right endpoint property.

\item $y$ intersects both $a_{i-1}$ and $a_i$. In that case, $y$ could replace $b_i$ in the alternating path; contradicts the fact that $(A,B)$ is a proper alternating path: here $r(y) < r(b_i)$, yet $b_i$ was included in the alternating path. 
\end{enumerate}

\begin{figure}[ht]
    \centering
    \includegraphics[scale=1]{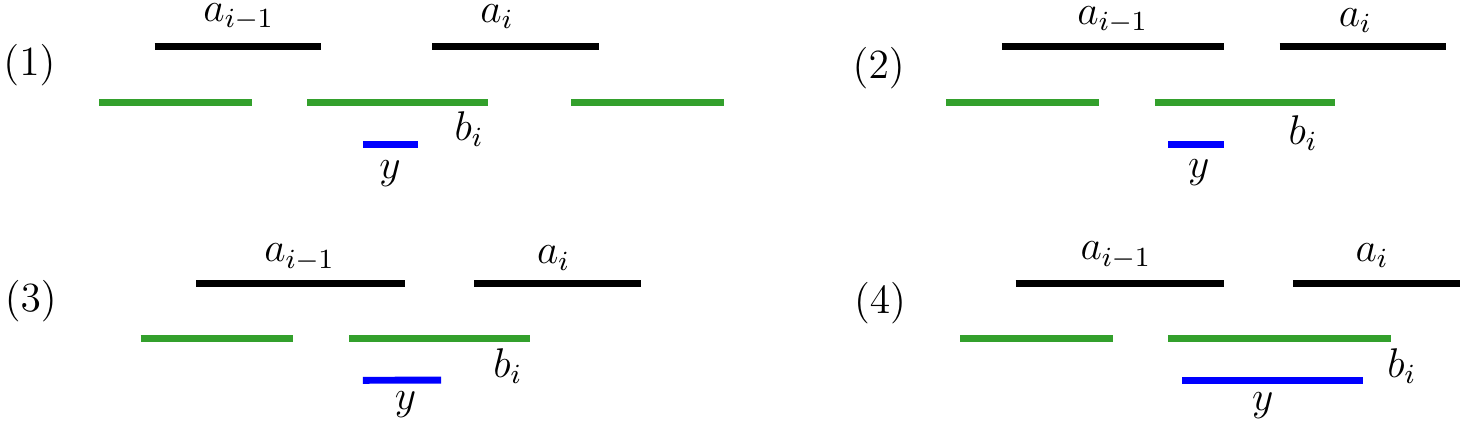}
    \caption{Obtaining contradiction in all cases for the no-containment property.}
    \label{fig:no-containment}
\end{figure}

Overall, in all cases we obtained a contradiction, implying that $(I \setminus A) \cup B$ satisfies the no-containment property.

\medskip

\textit{$k$-maximality:} We now show that $(I \setminus A) \cup B$ is a $k$-maximal independent set. Assume for contradiction that there exists a pair $(C,D)$ of size at most $k$ that induces an alternating path with respect to $(I\setminus A)\cup B$.  We will show that this contradicts either that $I$ is $k$-maximal or that $(A,B)$ is a proper alternating path. Let $C=\{c_1,\ldots ,c_{t'}\}$ and $D=\{d_1,\ldots ,d_{t'+1}\}$, with $t'\leq k$.

First observe that $C\cap B\not=\emptyset$; this can be easily shown by contradiction: If $C \cap B = \emptyset $, then $(C,D)$ is an alternating path with respect to $I$ of size at most $k$, contradicting that $I$ is $k$-maximal\footnote{Note that here we use crucially that $I$ satisfies the no-containment property. For an arbitrary $k$-maximal set with the no-containment property, this is not true, and such alternating paths could exist.}. Since $C \cap B \neq \emptyset$, we have that $C$ is non-empty.

Since $|C| \leq k$ and $|B| >k$, $C$ cannot be a strict superset of $B$. Either $C$ will be a contiguous subsequence of $B$ or it will extend it in one direction (left or right). As a result, one extreme interval of $C$ (either the leftmost or the rightmost) will belong to $B$. We give the proof for the case that the leftmost interval of $C$, namely $c_1$, belongs to $B$. In case $c_1 \notin B$, then $c_{t'} \in B$ and the proof is essentially the same by considering the mirror images of the intervals and obtaining the contradiction for the sibling alternating path $(A,B')$ that satisfies the rightmost left endpoint property (see Definition~\ref{def:sibling}). 

From now on we focus on the case where $c_1 \in B$.
There exists some $1 \leq i \leq t+1$ such that $b_i = c_1$. Since $(A,B)$ induces an alternating path, there are (at most) two intervals $a_{i-1}$ and $a_i$ intersecting $b_i = c_1$ (in case $i=1$ there is only one interval, $a_i = a_1$ and in case $i=t+1$ there exists only $a_{i-1} = a_t$).

We consider the intersection pattern of $a_{i-1},a_i,d_1$.  Note that since $(I \setminus A) \cup B$ satisfies the no-containment property, we have that $d_1$ can not be strictly contained in $c_1 = b_i$.

Note that if $i >1$, then $d_1$ should intersect with $a_{i-1}$: they both contain the left endpoint of $c_1 = b_i$. Moreover, it must be that $d_1 \neq a_{i-1}$, since $a_{i-1}$ contains the point $r(b_{i-1})$, but $d_1$ does not. In case $i=1$, then $a_{i-1} = a_0$ does not exist; for convenience in the proof we assume that $a_{i-1} = a_0$ exists and does not intersect with $d_1$. We need to consider two separate cases depending on the intersection between $d_1$ and $a_i$. 

\medskip

\noindent \textbf{Case 1: $d_1$ does not intersect $a_i$.} We distinguish between two subcases.

\begin{enumerate}
\item \textit{$d_1$ does not intersect $a_{i-1}$}. Recall this can happen only if $i=1$. In that case, we have that $d_1$ does not intersect any interval of $I$, thus $I \cup \lrbrace{d_1}$ is an independent set, therefore $I$ is not maximal, contradiction. 

\item \textit{$d_1$ intersects $a_{i-1}$} (see Figure~\ref{fig:higher_exchange_contr}): In that case, $b_1,a_1, \dotsc, b_{i-1},a_{i-1},d_1$ is an alternating path. Equivalently, for $A' = \lrbrace{a_1,\dotsc,a_{i-1}}$ and $B' = \lrbrace{b_1,\dotsc,b_{i-1}, d_1}$, $(A',B')$  is an alternating path. Since $A' \subset A$, we get that the alternating path $(A,B)$ is not proper. Contradiction. 
\end{enumerate}

\medskip

\noindent \textbf{Case 2: $d_1$ intersects $a_i$} (see Figure~\ref{fig:higher_exchange_contr}). Here we do not need any subcases. Note that $r(d_1) < r(c_1) = r(b_i)$. Thus, if $\ell(d_1) > \ell(c_1) = \ell(b_i)$, then $d_1$ is a strict subset of $c_1  =b_i$, which contradicts the no-containment property of $(I \setminus A) \cup B$. So it must be the case that $\ell(d_1) < \ell(c_1) = \ell(b_i)$. But then, in the alternating path $(A,B)$, the interval $b_i$ could have been replaced by $d_1$; this contradicts the assumption that $(A,B)$ is a proper alternating path, since it does not satisfy the leftmost right endpoint property.

\medskip

\begin{figure}[ht]
    \centering
    \includegraphics[scale=1]{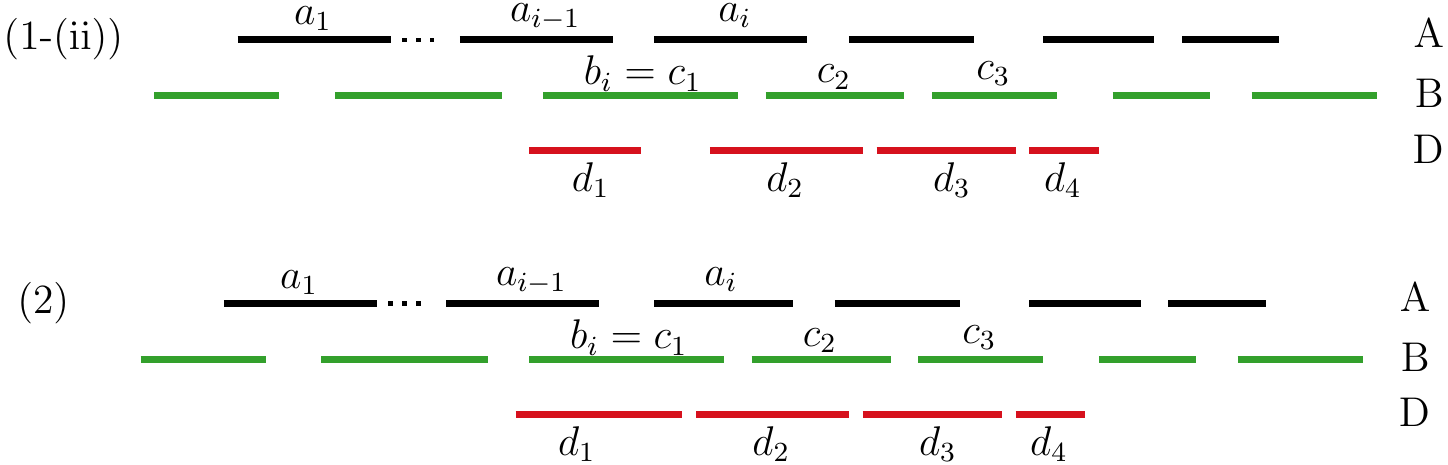}
    \caption{Showing the contradiction. On top, case 1(ii): If $d_1$ does not intersect $a_i$, then $b_1,a_1, \dotsc, a_{i-1}, d_1$ is an alternating path. Down, case 2: If $d_1$ intersects $a_i$, then it could replace $b_i$ and give an alternating path satisfying the leftmost right endpoint property.}
    \label{fig:higher_exchange_contr}
\end{figure}

We crucially note that the proof holds even if $A \cap D \neq \emptyset$: in all cases the only interval of $D$ used to obtain contradiction was $d_1$; since $d_1 \notin A$, as explained above, then the proof holds even if $d_j \in A$ for some $j > 1$.\qedhere

\end{proof}

We conclude with the proof of Lemma~\ref{lem:j-to-j}.

\medskip 

\noindent {\bf  Lemma~\ref{lem:j-to-j}}  {\em (restated)  Let $I$ be a $k$-valid independent set. Let $A \subseteq I$ and $B$ a (left or right) substitute of $A$. Then, $(I \setminus A) \cup B$ is a $k$-valid independent set.}

\begin{proof}
The proof of the no-containment property is the same as in Lemma~\ref{lem:higher_exchange}, by considering four cases and proving contradiction to all of them. The only difference is that the corner case with $i=t+1$ in case 3 cannot appear.

\textit{k-maximality:} Without loss of genrality, we only give the proof for the case $B$ is a left substitute of $A$. Part of the proof carries over from Lemma~\ref{lem:higher_exchange}. Suppose for contradiction that $\iab$ is not $k$-maximal and there exists an alternating path $(C,D)$ for $C=\{c_1,\ldots ,c_{t'}\}$ and $D=\{d_1,\ldots ,d_{t'+1}\}$, with $t'\leq k$.

We note that, in contrast to Lemma~\ref{lem:higher_exchange}, now it is not obvious that $C \cap B \neq \emptyset$ (see Figure~\ref{fig:C_disj_B}). Thus we first give the proof for this case and later we consider the case $C \cap B = \emptyset$.

We focus on the case $C \cap B \neq \emptyset$. We distinguish between two sub-cases, depending on whether $c_1 \in B$ or not. 

\medskip

\noindent \textbf{Case 1:  $c_1 \in B$}. Note that in this case, either $C$ is a strict subset of $B$, or $C$ extends $B$ to the right. Let $c_1 = b_i$. The proof is same as the proof of Lemma~\ref{lem:higher_exchange} based on intersections between $d_1$ and $a_{i}, a_{i-1}$ and showing the exact same contradiction in all cases.  

\noindent \textbf{Case 2: $c_1 \notin B$.} Note that in that case, $C$ is either a strict superset of $B$, or extends $B$ to the left. Observe that $c_1$
is on the left side of intervals of $B$. Let $b_1=c_i$. Observation: For all $1 \leq t \leq \min \lrbrace{j,t'+1-i} $, interval $d_{i+t}$ intersects $a_{t}$ (they both contain the right endpoint of $b_{t}$).  Let $t \leq j-1$ be the smallest index such that $d_{i+t}$ does not intersect $a_{t+1}$. We claim that such $t$ always exists. Then, $d_1,c_1,\dotsc,d_i,a_1,d_{i+1}, \dotsc, d_{i+t}$ is an alternating path of size less than $k$; equivalently $C' = \lrbrace{c_1, \dotsc,c_{i-1},a_1, \dotsc, a_{t}}$ and $D' = \lrbrace{d_1,\dotsc,d_{i+t}}$ is an alternating path with respect to $I$, of size $i+t-1 \leq i+t'+1-i-1 = t' < k$, a contradiction. It remains to show that such a $t$ always exist. To this end, we distinguish between two subcases to conclude the proof:

\begin{enumerate}
\item Case $A \cap D \neq \emptyset$: Let $j'$ be the smallest index such that $d_{j'} \in A$. Since for all $t \leq \min \lrbrace{j,t'+1-i}$ each interval $d_{i+t}$ intersects interval $a_t$, we get that $d_{j'}$ intersects $a_{j'-i}$, thus $d_{j'} = a_{j'-i}$.  That means, interval $d_{j'-1}$, does not intersect $a_{j'-i}$. Thus for $t = j' - i -1$, we have that interval $d_{i+t}$ does not intersect $a_{t+1}$.

\item Case $A \cap D = \emptyset$. In that case we need some further case analysis. 

\begin{enumerate}
\item $t'+1-i \geq j$: Note that in this case, $C$ is a strict superset of $B$. Assume such $t$ does not exist. Then, $((C \setminus B)\cup A ,D) $ is an alternating path of size at most $k$, contradicting that $I$ is $k$-maximal.

\item $t'+1-i < j$: Note that in this case, $C$ extends $B$ on its left. Assume such $t$ does not exist. Then, we have that interval  $d_{t'} = d_{i+(t'-i)}$ intersects both $a_{t'-i}$ and $a_{t'-i+1}$, therefore  $\ell(d_{t'+1}) > r(d_{t'}) = r(d_{i+(t'-i)}) \geq \ell(a_{t'-i+1})$. Also, interval $b_{t'-i+2}$ exists (since $t'-i+2 \leq j$) and has $\ell(b_{t'-i+2}) < r(a_{t'-i+1}) $. Also $r(d_{t'+1}) < \ell(b_{t'-i+2}) $ (because the alternating path ends at $d_{t'+1}$), therefore  $r(d_{t'+1}) < r(a_{t'-i+1})$. 

Overall we have that  $\ell(d_{t'+1}) > \ell(a_{t'-i+1})$ and $r(d_{t'+1}) < r(a_{t'-i+1})$, i.e, $d_{t'+1}$ is strictly contained in $a_{t'-i+1}$, contradicting that $I$ is a $k$-valid independent set. 
\end{enumerate}

\end{enumerate}

\begin{figure}[ht]
    \centering
    \includegraphics[scale=1]{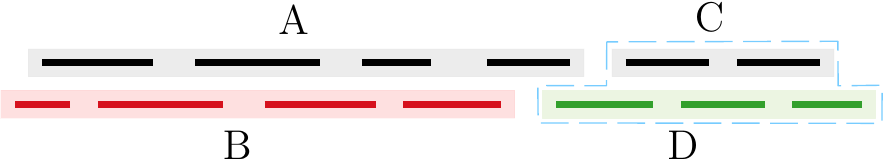}
    \caption{$B$ is a left substitute of $A$, and it might be that $(C,D)$ is an alternating path. But in that case, $(A \cup C,B \cup D)$ is an alternating path.}
    \label{fig:C_disj_B}
\end{figure}

It remains to consider the case $C \cap B = \emptyset$ (Figure~\ref{fig:C_disj_B}) and obtain again a contradiction. First, show that the only way this can happen is if $C$ contains consecutive intervals to the right of $A$. But in this case, observe that $(A \cup C, B \cup D)$ is an alternating path of size $j+t' \leq 2k$, which contradicts that $B$ is a substitute of $A$. \qedhere

\end{proof}

\subsection{Extensions}
\label{sec:intervals_extend}

We now extend our data structure to also support, in the same running time $O(k^2 \log n)$, the following operations:

\begin{enumerate}
    \item {\bf Merge:} Given two sets of intervals $S_1$ and $S_2$ and such that for all $x \in S_1$ $\ell(x) \leq t$ and for all $y \in S_2$, $\ell(y) > t$ and two $k$-valid independent sets $I_1$ and $I_2$ of $S_1$ and $S_2$ respectively, get a $k$-valid independent set $I$ of $S= S_1 \cup S_2$.
    \item {\bf Split:} Given a set of intervals $S$, a $k$-valid independent set $I$ of $S$, and a value $t$, split $S$ into $S_1$ and $S_2$ such that $x \in S_1$ $\ell(x) \leq t$ and for all $y \in S_2$, $\ell(y) > t$ and produce $k$-valid independent sets $I_1, I_2$ of $S_1$ and $S_2$ respectively. 
    \item {\bf Clip($t$).}  Assume we store a set $S$ of intervals and a $k$-valid independent set $I$ and let $t$ be a point to the left of the rightmost left endpoint of all intervals of $S$. This operation shrinks all intervals $x$ with $r(x) > t$, such that $r(x) = t$. 
\end{enumerate}

Furthermore we show that some types of changes in the input set $S$ do not affect our solution.

\begin{itemize}
\item {\bf Extend($y$):} Assume we store a set $S$ of intervals and a $k$-valid independent set $I$. Then, if an interval $y \in S \setminus I$ gets replaced by $y'$ such that $y$ is a strict subset of $y'$, then $I$ remains $k$-valid.
\end{itemize}

The operations on the interval query data structure can be done using red-black trees as explained in Lemma~\ref{lem:ors}. The non-trivial part is to show how to maintain $k$-valid independent sets under splits and merges. For example, when merging, the leftmost interval of $I_2$ might intersect (or even be strictly contained in) the rightmost interval of $I_1$.
We show that we can reduce those operations to a constant number of insertions and deletions.

\paragraph{1. Merge:} We now describe our merge algorithm. We start by inserting a fake tiny interval $b$ in $S_1$, with $\ell(b) >t$ such that $b$ is strictly contained in any interval of $L$ it intersects, and it does not intersect the leftmost interval of $S_2$. Our algorithm can be described as follows.

\begin{enumerate}
    \item \label{merge:ins_b} Insert $b$ in $S_1$. Update $I_1$ to $I'_1$. 
    \item \label{merge:merge} $S \leftarrow S_1 \cup S_2$ and $I \leftarrow I'_1 \cup I_2 $.
    \item \label{merge:del_b} Delete $b$ from $S$. Update $I$ to $I'$.
\end{enumerate}

\paragraph{Running time.} Note that this algorithm runs in time $O(k^2 \log n)$, where $n = |S|$. This is because steps~\ref{merge:ins_b} and~\ref{merge:merge} require $O(k \log n)$ due to our insertion algorithm and Lemma~\ref{lem:ors} respectively, and step~\ref{merge:del_b} requires time $O(k^2 \log n)$ using our deletion algorithm from Section~\ref{sec:intervals_alg}.

\paragraph{Correctness.} We show that the merge algorithm indeed produces a $k$-valid independent set of $S$. Since our insertion algorithm from Section~\ref{sec:intervals_alg} maintains a $k$-valid independent set, in particular it satisfies the no-containment property, we get that after step~\ref{merge:ins_b}, $b \in I'_1$. Therefore, after step~\ref{merge:merge}, $I = I'_1 \cup I_2$ is an independent set of $S$, since $b$ ensures that no overlap exists. We want to show that $I$ is also $k$-valid. Towards proving this, we make one observation. 

\begin{observation}
\label{obs:bomb}
For any interval $x \in S_1 \cup S_2$, the endpoints  $\ell(x)$ and $r(x)$ do not intersect $b$.
\end{observation}


We now proceed to our basic lemma, showing that the new independent set is $k$-valid. 

\begin{lemma}
\label{lem:merge_valid}
The independent set $I = I'_1 \cup I_2$ obtained in step~\ref{merge:merge} is $k$-valid
\end{lemma}

\begin{proof}
$k$-maximality: If there exists an alternating path, it should contain $b$. But due to Observation~\ref{obs:bomb} no endpoint of any interval intersects $b$, thus there cannot be such an alternating path.

No-Containment: No interval of $I'_1$ contains an interval of $S_1$. No interval of $I_2$ contains an interval of $S_2$. By construction, an interval of $I_2$ cannot contain an interval of $S_1$, since the left endpoints are on different sides of $t$. Similarly, an interval of $I'_1 \setminus \lrbrace{b}$ does not contain an interval of $S_2$. Finally, $b$ does not contain any interval of $S_2$ due to Observation~\ref{obs:bomb}.
\end{proof}

Now it is easy to see that the algorithm outputs a $k$-valid independent set: After step~\ref{merge:merge}, $I = I'_1 \cup I_2$ is a $k$-valid independent set of $S_1 \cup S_2 \cup \lrbrace{b}$. Thus after deletion of $b$, the new independent set $I'$ is a $k$-valid independent set of $S = S_1 \cup S_2$, since our deletion algorithm from Section~\ref{sec:intervals_alg} maintains a $k$-valid independent set.

\paragraph{2. Split:} We now proceed to the split operation. It is almost dual to the merge and the ideas used are very similar. Recall we maintain a set of intervals $S$ and a $k$-valid independent set $I$, and we want to split to $S_1$ and $S_2$ and corresponding independent sets $I_1$ and $I_2$, such that all intervals $x \in S_1$ have $\ell(x) \leq t$ and all $y \in S_2$ have $\ell(x) > t$. 

We introduce a fake tiny interval $b$ such that $\ell(b) > t$ which is strictly contained in all intervals $x$ with $\ell(x) \leq t$ and $r(x) >t$ and $r(b)$ is smaller than the leftmost left endpoint larger than $t$. The algorithm is the following

\begin{enumerate}
    \item \label{split:insert} Insert $b$ in $S$. $I'$ is the new independent set.
    \item \label{split:split} Split $S$ into $S_1 \cup \lrbrace{b}$ and $S_2$. Split $I'$ into $I'_1$ and $I_2$, where $b$ is the rightmost interval of $I'_1$.
    \item \label{split:delete} Delete $b$ from $S_1 \cup \lrbrace{b}$ and update $I'_1$ to $I_1$ using our deletion algorithm.  
\end{enumerate}
 
\paragraph{Running time.} Like merge, this algorithm runs in worst-case $O(k^2 \log n)$ time, due to the bounds from Lemma~\ref{lem:ors} and the running time of our insertion and deletion algorithms.
 
\paragraph{Correctness.} The proof of correctness is  similar to that of the merge operation. After step~\ref{split:insert}, $I'$ is a $k$-valid independent set of $S\cup\lrbrace{b}$. It remains to realize that after step~\ref{split:split}, $I'_1$ and $I_2$ are $k$-valid independent sets of $S_1 \cup \lrbrace{b}$ and $S_2$ respectively. 
Then, the result for $I_2$ is immediate and for $I_1$ it comes from correctness of our deletion algorithm of Section~\ref{sec:intervals_alg}.

\paragraph{3. Clip:} Given a set $S$ of intervals, a $k$-valid independent set $I$ and $t$ be a point to the left of the rightmost left endpoint of all intervals of $S$, we shrink all intervals $x$ with $r(x) > t$, such that $r(x) = t$. 

We show that the change in the independent set $I$ can be supported in time $O(\log n)$. Let $\tau$ be the interval of $S$ the rightmost left endpoint. Note that $\ell(\tau) < t$.

We begin with an observation, which is essentially a corollary of Lemma~\ref{lem:higher_exchange} from Section~\ref{sec:intervals_cor}. 

\begin{observation}
\label{obs:valid_leftmost}
Let $I$ be a $k$-valid independent set maintained by the IQDS structure. Then, using the IQDS structure, we can modify $I$ to contain the leftmost interval $\tau$ of $S$ and remain $k$-valid, in time $O(k \log n)$. 
\end{observation}

\begin{proof}
Let $x$ be the righthmost interval of $I$. If $x=\tau$ then $\tau \in I$ and no change is needed. We focus thus on the case $x \neq \tau$. 

Set $I' \leftarrow (I \setminus \lrbrace{x}) \cup \tau $. This might create an alternating path of size exactly $k$ to the left of $\tau$. Search for such alternating path $(A,B)$ and if exists, do the swap, i.e, set $I'' \leftarrow (I' \setminus A) \cup B$. Clearly the runtime using IQDS is $O(k \log n)$

It is easy to see that the new independent set $I''$ is $k$-valid: If no alternating path found, then the change is a right 1-to-1 substitution, thus by Lemma~\ref{lem:j-to-j}, $I''$ is $k$-valid. If an alternating path $(A,B)$ was found, then the overall change corresponds to an alternating path with respect to $I$, of size $k+1$: the alternating path is $(A \cup \lrbrace{x},B \cup \lrbrace{\tau})$. By Lemma~\ref{lem:higher_exchange}, this exchange produces a $k$-valid independent set; thus $I''$ is $k$-valid. \end{proof}

Using Observation~\ref{obs:valid_leftmost}, we can show that the operation CLIP($x$) maintains a $k$-valid independent set. We first make sure that $\tau \in I$; if not we modify $I$ using the procedure described above (in time $O(k \log n)$). 

Then it is easy to see that even after shrinking the intervals, $I$ remains a $k$-valid independent set. The no-containment property holds trivially, since no interval has a larger left endpoint. $k$-maximality is also easy: since no interval has its left endpoint inside $\tau$, there cannot be any alternating path involving $\tau$. Alternating paths without $\tau$ cannot exist, since then they should have existed before, contradicting that $I$ is $k$-valid.

\paragraph{4. Extend:} We conclude by showing that given a $k$-valid independent set of intervals $I$ of $S$, if an interval $y \in S \setminus I$ gets replaced by a strict superset $y'$, then $I$ remains $k$-valid.

\begin{lemma}
\label{lem:insert_superset}
Let $I$ be a $k$-valid independent set of a set $S$ of intervals. Then, if an interval $y \in S \setminus I$ gets replaced by $y'$ such that $y'$ strictly contains $y$, then $I$ remains $k$-valid.
\end{lemma}

\begin{proof}

The no-containment is clearly preserved. $I$ satisfies this property, and the new interval contains the previous. Since $y$ was not contained in any interval $x \in I$, then $y'$ is not contained either.

The $k$-maximality property can be proven by contradiction. Suppose that after replacing $y$ by $y'$, there exists an alternating path $(A,B)$ of size $t \leq k$. Let $A = \lrbrace{a_1,\dotsc,a_t}$ and $B = \lrbrace{b_1,\dotsc,b_{t+1}}$. There exists some $1 \leq i \leq t+1$ such that $b_i = y'$. Let us first focus in the case $2 \leq i \leq t$. Interval $b_i$ should intersect both $a_{i-1}$ and $a_i$. Since $b_i = y'$ contains $y$, we examine the intersections between $y$ and $a_{i-1},a_i$.

\noindent \textbf{Case 1: $y$ does not intersect any of $a_{i-1},a_i$.} Then, $y$ could be added to $I$ and produce an independent set, thus $I$ was not maximal with respect to $S$, contradiction.

\noindent \textbf{Case 2: $y$ intersects only one of $a_{i-1}$ or $a_i$}. We focus on the case intersecting $a_{i-1}$ and the other is symmetric. We have that $b_1,a_1,\dotsc,a_{i-1},y$ is an alternating path of size $i-1 \leq t \leq k $, contradiction.

\noindent \textbf{Case 3: $y$ intersects both $a_{i-1}$ and $a_i$.} Then, $(A,(B \setminus b_i) \cup y )$ is an alternating path of $S$ with respect to $I$ of size $t \leq k$, contradiction, since $I$ is $k$-maximal.

It remains to consider the corner case that $i=1$ or $i=t+1$. There, $b_i = y'$ intersects only one interval ($a_1$ or $a_t$ respectively). Thus the only cases appearing are the cases (1) and (2) above, and we show the contradiction using the same arguments. 
\end{proof}

%% file: 3-quadtreeApproach.tex


\pagebreak

\section{Static Squares: The Quadtree Approach}\label{s:quadtree}

In this section we turn our attention to squares. We present a $O(1)$-approximate solution for the (static) independent set problem where all input objects are squares. Although such results (or even PTAS) are already known, we present our approach for the static case, while developing the structural observations that will become the invariants when we address dynamization in Section~\ref{sec:squares_dynamic}.

For this section, the input $S$ is a set of $n$ axis-aligned squares in the plane. Let $\Iopt{S}$ denote a maximum independent set of $S$. The maximum size of an independent set of $S$
will be denoted by $\OPT{S}=|\Iopt{S}|$.
Given a set $\Sinput$, our goal is to compute a set $\Iapprox \subseteq \Sinput$ which is an independent set of squares and where $E[|\Iapprox|] \geq c\cdot \OPT{\Sinput}$ for some absolute constant $c>0$. 
%

%

We will use as a black box a solution for the 1-D problem of given a set of intervals, to compute a $c$-approximate maximum independent set of these intervals. We will show that using such a solution, we can obtain a $O(c)$-approximation for squares.

\paragraphh{Random quadtrees.}
We assume all squares in $\Sinput$ are inside the unit square $[0,1]^2$.
Let $\Qinf$ be the infinite quadtree where the root node is a square centered on a random point in the the unit square and with a random side length in $[1,2]$. 
Everything that follows is implicitly parameterized by this choice of random quadtree $\Qinf$ and $\Sinput$. We use $\node$, possibly subscripted, to denote a node of the quadtree and $\Square{\node}$ to denote $\node$'s defining square. 
We will use $s$, possibly subscripted to denote a square in $\Sinput$.

\paragraphh{Centered squares.}
Given a square $s$, let $\Node{s}$ be the smallest quadtree node of $\Qinf$ that completely contains $s$, we say that $s$ and $\Node{s}$ are associated with each other\footnote{We note that computing the coordinates of $\Node{s}$ from $s$ is the only operation on the input that we need beyond performing comparisons. This operation can be implemented with a binary logarithm, a floor, and a binary exponentiation.}. 
A square $s$ is said to be $\emph{centered}$ if $s$ contains the center point of its associated node, $\Square{\Node{s}}$. See Figure~\ref{f:centered}.
We use $\Centered$ to denote the set subset of $\Sinput$ where the squares are centered.

\paragraph{Outline.} We will present the approach in a similar way as explained at a high-level in Section~\ref{sec:outline_squares}. 

 \begin{enumerate}
     \item  First, we will show that by loosing a factor of $16$, we can restrict our attention to centered squares (Lemma~\ref{l:bds}).
     
    \item  In Section~\ref{sec:paths_suffice} we focus on the subtree of $\Qinf$ including nodes associated with centered squares and their ancestors, denoted by $\Qtree$ and show that given a linearly approximate solution for paths of $\Qtree$, we can get a $O(1)$-approximate solution for $\Qtree$. (Lemma~\ref{l:combine}). 
    
     \item  In Section~\ref{sec:paths_8c} we show how to decompose each path into four monotone subpaths, and that by loosing a factor of 4, it suffices to solve the problem in each of the monotone subpaths and use only the largest of the four independent sets.
     
     \item Last, we show that an approximate independent set of centered squares in monotone subpaths reduces to the problem of the maximum independent set of intervals.
     (Lemma~\ref{l:indsetsize}). 
 \end{enumerate}

\subsection{Centered Squares}

We begin with showing that by losing a $O(1)$ factor, we can focus on centered squares and search for an approximate maximum independent set in $\Centered$ rather than $\Sinput$ itself.

\begin{lemma} \label{l:bds}
The maximum size of an independent set of the centered squares $\Centered$ is expected to be at least $\frac{1}{16}$ that of $\Sinput$: $E[\OPT{\Centered}] \geq \frac{1}{16}\OPT{\Sinput}$.
\end{lemma}

\begin{proof}
Given a square $s$ of size $k$, let $\ell$ be the size of the smallest quadtree cell that is larger then $k$. Observe that $\ell$ is uniformly distributed from $k$ to $2k$. Let $\node$ be the node of the quadtree of size $\ell$ that contains the lower-left corner of $s$.  
We know $s$ is centered if its lower-left corner lies in the the lower-left quadrant of $\node$ and if $s$ lies entirely in $\node$. The first, that the lower-left corner of $s$ is in the lower left quadrant of $\node$, happens with probability $\frac{1}{4}$. Given this, we need to know if the $x$ extent of $s$ is in the node, this can be done by checking of where the square begins relative to the left of the node plus its size is smaller than the width of the node; that is, whether its $x$ coordinate, which uniformly in $[0,\frac{l}{2}]$ relative to the left of $\node$ plus its $x$-extent, which is $k$, is less than the width of the square $\ell$, which is uniform in $[k,2k]$. 
This happens with probability $\frac{1}{2}$, and independently with probability $\frac{1}{2}$ for the $y$-extent as well. Multiplying, this gives a probability of at least $\frac{1}{16}$ that $s$ in centered in $\node$.

For any $s \in \Sinput$, let $i(s)$ be 1 if $s \in \Centered$, otherwise, $i(s)=0$; 
from the previous paragraph $E[i(s)]\geq \frac{1}{16}$. 
Let $I$ be $\Iopt{\Sinput} \cap \Centered$, those squares from a maximum independent set that are centered. We know $E[|I|] =\sum_{s\in  \Iopt{\Sinput}}i(s)$, and thus by linearity of expectation $E[|I|] \geq \frac{\OPT{\Sinput}}{16}$.
Since $I$ is a subset of $\Iopt{\Sinput}$, it is an independent set and thus $\OPT{\Centered}\geq |I|$. Combining these gives the lemma.
\end{proof}

\subsection{From Quadtrees to Paths}
\label{sec:paths_suffice}
We now show that the essential hardness on obtaining an independent set on the quadtree relies on getting independent sets on special types of paths of the tree. 

\paragraphh{Marked nodes and the finite quadtree.}
The quadtree nodes in $\cup_{s\in \Centered} \Node{s}$ are said to be \emph{marked} and are denoted as $\Marked$.
Let $\Squares{\node}$ be the inverse of the $\Node{s}$ function, that is, given a node of the quadtree $\node$, it returns the set of squares $s$ such that $\Node{s}=\node$, these are the squares associated with a node. 
We define the quadtree $\Qtree$ to be the subtree of the infinite quadtree $\Qinf$ containing the marked nodes $\Marked$ and their ancestors. %
Each node of $\Qtree$ is a leaf, an internal node (a node with more than one child), or a monochild node (a node with one child); with the exception that the root is considered to be an internal node if it has one child. See figure~\ref{f:thequad}.
We use the word \emph{path} to refer to a maximal set of connected monochild nodes in $\Qtree$. 
By definition, the node above the top node and below the bottom node in a path must exist in $\Qtree$ and will not be monochild. 
We use these nodes to denote a path: $\Path(\node_{\text{top}},\node_{\text{bottom}})$ refers to the path strictly between $\node_{\text{top}}$ and $\node_{\text{bottom}}$, where $\node_{\text{top}}$ is an ancestor of $\node_{\text{bottom}}$, neither is a monochild node, and there are only monochild nodes between
$\node_{\text{top}}$ and
$\node_{\text{bottom}}$.
Given a path $\Path=\Path(\node_{\text{top}},\node_{\text{bottom}})$, let $\Squares{\Path}$ refer to the squares associated with nodes of the path,
$ \Squares{\Path} \coloneqq \cup_{n \in \Path} \Squares{\node}$,
and $\Pathtop,\Pathbottom$ refer to the nodes $\node_{\text{top}},\node_{\text{bottom}}$ that bound $\Path$.

Let $\Qleaves$ refer to the set of leaves of $\Qtree$, $\Qinternal$ refer to the internal nodes of $\Qtree$, and $\Qpaths$ refer to the set of monochild paths of $\Qtree$. 
The nodes in these sets partition the nodes of $\Qtree$.
Observe that the size of $\Qtree$, measured in nodes, cannot be bounded as a function of $|\Centered|$ as we do not have any bound on the aspect ratio of the squares stored. However, the number of leaves, $|\Qleaves|$, internal nodes, $|\Qinternal|$, and paths, $|\Qpaths|$ are all linear in $|\Centered|$.

\paragraphh{Protected Independent Sets.} Given a path $\Path$
we say a set $\Ippath{\Path} \subseteq \Squares{\Path}$ is a protected independent set with respect to 
$\Path$ if it is an independent set, and if no square in $\Ippath{\Path}$ intersects the square $\Square{(\Pathbottom)}$. 
This definition implies that no square $\Ippath{\Path}$ can intersect any squares associated with any nodes in $\Qtree$ not on this path; that is, $\Ippath{\Path}$ is disjoint from all squares in $\Centered \setminus \Squares{\Path}$. 
It is this property that makes protected independent sets valuable.
We show that to obtain an approximate independent sets, it is sufficient to use protected independent sets, along with one square associated with each leaf:

\begin{lemma} \label{l:combine}
Let $\Iapprox$ be the subset of $\Centered$ which is the union of
\begin{itemize}
    \item An arbitrary square in $\Squares{\node}$ for each leaf $\node \in \Qleaves$
    \item For each path $\Path \in \Qpaths$, a protected independent set $\Ippath{\Path}$. We require that 
    $$\sum_{\mathclap{\Path \in \Qpaths}} | \Ippath{\Path}| \geq c_1  \sum_{\mathclap{\Path \in \Qpaths}} \OPT{\Squares{\Path}} -  c_2 | \Qpaths | ,$$ 
    for some absolute positive constants $c_1> 0$, $c_2 \geq 0$. That is, in aggregate, all of these protected independent sets must be within a linear factor of the maximum protected independent sets on all paths. 
    \item No squares in $\Squares{\node}$, for each internal node $\node \in \Qinternal$
\end{itemize}
Observe that $|\Iapprox|=|\Qleaves| +\sum_{{\Path \in \Qpaths}} |\Ippath{\Path}|$.
The set $\Iapprox$ is an independent set of squares and $|\Iapprox|\geq \frac{c_1}{2c_1+c_2+1} \OPT{\Centered}$.
\end{lemma}

\begin{proof}
First, we argue that $\Iapprox$ is an independent set. For any two squares $s_1, s_2 \in \Iapprox$, we argue they cannot intersect. This has several cases:
\begin{itemize}
    \item Both $s_1$ and $s_2$ are associated with the same leaf node $\node=\Node{s_1}=\Node{s_2} \in \Qleaves$. This cannot happen as this means both $\node_1$ and $\node_2$ are in $\Squares{\node}$, but only one element from $\Squares{\node}$ is in $\Iapprox$ by the construction.
    \item Both $\Node{s_1}$ and $\Node{s_2}$ are nodes, neither of which is the same as or the ancestor of the other. In a quadtree, squares of nodes which are not ancestors or descendants are disjoint, and thus $s_1$ and $s_2$, which are contained in the squares defining these quadtree nodes, $\Square{\Node{s_1}}$ and $\Square{\Node{s_1}}$ are disjoint.
    \item Both $\Node{s_1}$ and $\Node{s_2}$ are in the same $\Squares{\Path}$, for some $\Path \in\Qpaths$. Then they would only be in $\Iapprox$ if they were both in $\Ippath{\Path}$, which is by definition an independent set.
    \item The only remaining case is where $\Node{s_1}$ is part of some path $\Path \in \Qpaths$ and $s_2$ is a descendent of this path and thus inside the square $\Square{(\Pathbottom)}$. But, since $s_1$ is in $\Ippath{\Path}$, and $\Ippath{\Path}$ is a protected independent set, by definition $s_1$ will not intersect $\Square{(\Pathbottom)}$ and thus will not intersect $s_2$.
\end{itemize}
Second, we argue that $\Iapprox$ is an approximation. For disjoint $S_1,S_2$, $\OPT{S_1 \cup S_2} \leq \OPT{S_1} + \OPT{S_2}$. Thus we compute the independent sets of the squares associated with each part of the quadtree (leaves, internal nodes, degree-1 paths) separately:
\begin{align*}
\OPT{\Centered}  & \leq \OPT{\Qleaves} + \OPT{\Qinternal } + \sum_{\mathclap{\Path \in \Qpaths}} \OPT{\Squares{\Path}}  \\
\intertext{As all leaves are disjoint, are associated with at least on square each, and each square in a leaf intersects the center of the leaf, the optimum is always exactly one square from each leaf and $\OPT{\Qleaves}= |\Qleaves|$:}
 & \leq |\Qleaves| + \OPT{\Qinternal } + \sum_{\mathclap{\Path \in \Qpaths}} \OPT{\Squares{\Path}}  \\
\intertext{As the number of internal nodes is at most the number of leaves, and there can be most one square associated with each internal node in any independent set, $\OPT{\Qinternal }\leq |\Qinternal|\leq |\Qleaves|$:}
& \leq 2|\Qleaves| + \sum_{\mathclap{\Path \in \Qpaths}} \OPT{\Squares{\Path}}|   \\
 \intertext{In the statement of the lemma $
 \sum_{\Path \in \Qpaths} | \Ippath{\Path}| \geq c_1  \sum_{\Path \in \Qpaths} \OPT{\Squares{\Path}} -  c_2 | \Qpaths |$:}
& \leq 2|\Qleaves|  + \frac{1}{c_1} \sum_{{\Path \in \Qpaths}}  |\Ippath{\Path}|+ \frac{c_2}{c_1}|\Qpaths| \\ 
\intertext{As there are no more paths than leaves:}
& \leq \lrp{2+\frac{c_2}{c_1}}|\Qleaves|  +\frac{1}{c_1} \sum_{{\Path \in \Qpaths}}  |\Ippath{\Path}| \\ 
& \leq \lrp{2+\frac{c_2+1}{c_1}}\lrp{ |\Qleaves|  +\sum_{\mathclap{\Path \in \Qpaths}}  |\Ippath{\Path}| } \\ 
\intertext{By the definition of $\Iapprox$ in the statement of the lemma:}
& \leq \lrp{2+\frac{c_2+1}{c_1}} |\Iapprox|  = \frac{1}{\frac{c_1}{2c_1+c_2+1}} |\Iapprox| 
\end{align*}\qedhere
\end{proof}

\subsection{Paths}
\label{sec:paths_8c}

In the previous subsection, the results depended on obtaining an approximate independent set for the squares associated with each path. In this subsection we will show how this can be solved, given a solution to the 1-D problem of computing the approximate independent set of intervals.

\paragraphh{Monotone paths. }
In this section we assume a monochild path $\Path \in \Qpaths$. 
We assume paths are ordered from the highest to lowest node in the tree.
Let $\Depth{\node}$ represent the depth of a node in $\Qtree$, and we use $\Depth{s}$ as a shorthand for $\Depth{\Node{s}}$.

As each node on a path in $\Qpaths$ has only one child, label each node in the path with the quadrant of its child ($\qa, \qb, \qc, \qd$). The label of a square in $s \in \Squares{\Path}$ is the label of $\Node{s}$ in $\Path$. 
We partition the nodes of $\Path$ into subpaths $\Pathq{\qa}, \Pathq{\qb}, \Pathq{\qc}$ and $\Pathq{\qd}$. See~Figure~\ref{f:patha}.
We use 
$\Pathextend{q}$ to refer to $\Pathq{q}$ with $\Pathbottom$ appended on the end, we call this an \emph{extended subpath} (we will use $q$ to make statements that apply to an arbitrary quadrant).

All of $\Pathq{q}$ and $\Pathextend{q}$ are referred to as \emph{monotone} (extended) subpaths as for every pair of nodes $\node_1,\node_2$ in a subpath, if $\node_1$ appears before $\node_2$ on the subpath, then $\node_2$ is in quadrant $q$ of $\node_1$. 

Our general strategy will be to solve the independent set problem for each of the four quadrants and take the largest, which only will cause the loss of a factor of four:

\begin{fact} \label{f:max4}
As each $\Path$ in $\Qpaths$ is partitioned into $\Pathq{q}$ for 
$q \in \{\qa,\qb,\qc,\qd\}$, we know that 
$$\max_{q \in \{\qa,\qb,\qc,\qd\} } \OPT{\Squares{\Pathq{q}}} \geq \frac{1}{4} \OPT{\Squares{\Path}}$$ 
and 
$$\max_{q \in \{\qa,\qb,\qc,\qd\} } \sum_{\Path \in \Qpaths} \OPT{\Squares{\Pathq{q}}} \geq \frac{1}{4} \sum_{\Path \in \Qpaths}  \OPT{\Squares{\Path}}.$$
\end{fact}



\old{
\begin{figure}
    \centering
    \includegraphics{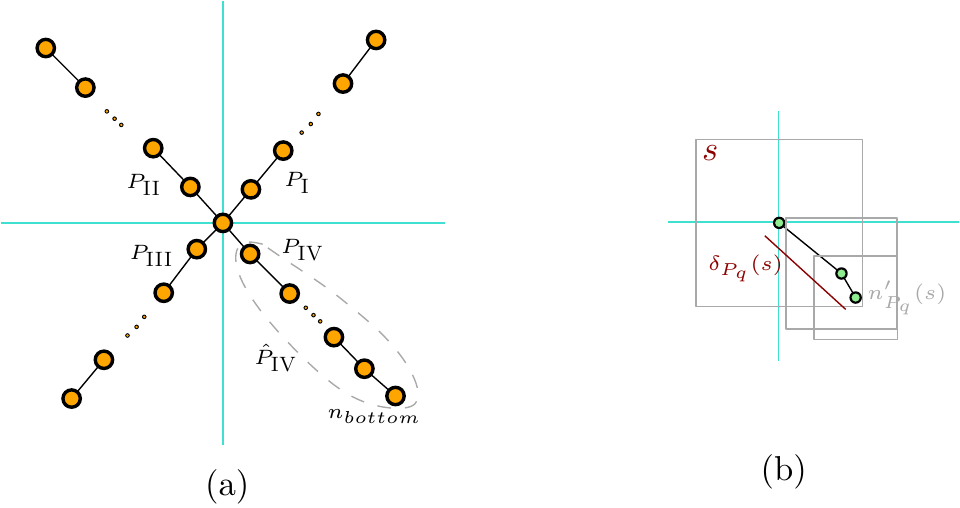}
    \caption{(a) An illustration of the monotone paths (b)} 
    
    \label{fig:quad-tree-paths}
\end{figure}
}

\paragraphh{From monotone paths to intervals. }
Given a monotone extended subpath $\Pathq{q}$, and square $s$, where $s \in \Squares{\Pathq{q}}$, we will be interested in which are the nodes $\node '$ in $\Pathq{q}$ where $s$ intersects the centers $\Squarecenter{\node '}$. 
As $s$ is centered, we know $s$ intersects $\Squarecenter{\Node{s}}$, and by the definition of $\Node{s}$, we know $s$ will not intersect any $\Squarecenter{\node '}$ for any $\node '$ that comes before $\Node{s}$ in $\Pathq{q}$ (and thus is a parent of $\Node{s}$ in $\Qtree$). 
Let $\node '_{\Pathq{q}}(s)$ be the last node in $\Pathq{q}$ such that $s$ intersects its center.
We use $\Depthmax{\Pathq{q}}{s}$ as a shorthand for $d(\node '_{\Pathq{q}}(s))$. 
Trivially $\Depth{s} \leq \Depthmax{\Pathq{q}}{s}$.
We denote the interval $[\Depth{s},\Depthmax{\Pathq{q}}{s}]$ as $\Interval{P_q}{s}$. See Figure~\ref{f:pathb}.
We use $\Pathq{q}[d_1,d_2]$ to refer to the subpath of $\Pathq{q}$ consisting of the nodes of depths between $d_1$ and $d_2$, inclusive.
What is interesting about $s$ is that from the above we know that $s$ only intersects squares of $\Squares{\Pathq{q}}$ with depths in the range $\Interval{P_q}{s}$, but that it has number of interesting geometric properties including that intersects \emph{all} such squares:

\begin{lemma} \label{l:monotone}
Given an extended monotone subpath $\Pathq{q}$:
\begin{itemize}

\item  The centers of the nodes of monotone path $\Pathq{q}$ followed by an arbitrary point in $\Square{(\Pathbottom)}$ are monotone with respect to the $x$ and $y$ axes. 

\item A square $s \in \Squares{\Pathq{q}}$ intersects the centers $\Squarecenter{\node}$ of all nodes $\node \in \Pathq{q}$ with depths in $\Pathq{q}[\Interval{\Pathq{q}}{s}]$, and thus $s$ intersects all squares $s$ in $\Squares{\Pathq{q}}$ with depths in $\Pathq{q}[\Interval{\Pathq{q}}{s}]$.

\item Given squares $s_1,s_2 \in \Squares{\Pathq{q}}$, if the intervals $\Interval{\Pathq{q}}{s_1}$ and $\Interval{\Pathq{q}}{s_2}$ intersect, then $s_1$ and $s_2$ intersect.

\item Given squares $s_1,s_2,s_3 \in \Squares{\Pathq{q}}$, where $\Interval{\Pathq{q}}{s_1}$, $\Interval{\Pathq{q}}{s_2}$, $\Interval{\Pathq{q}}{s_3}$ are disjoint, and 
$\Interval{\Pathq{q}}{s_1}$ is to the left of $\Interval{\Pathq{q}}{s_2}$ which is to the left of $\Interval{\Pathq{q}}{s_3}$, then $s_1$ and $s_3$ are disjoint.
\item Given squares $s_1,s_2 \in \Squares{\Pathq{q}}$, $\Depth{s_1}<\Depth{s_2}$ if the intervals $\Interval{\Pathq{q}}{s_1}$ and $\Interval{\Pathq{q}}{s_2}$ are disjoint, then $s_1$ does not intersect $\Square{(\Pathbottom)}$.
\end{itemize}
\end{lemma}

\begin{proof}
The first statement holds by the definition of a monotone (extended) subpath, since all succeeding nodes are in the same quadrant relative to the centers of all preceding nodes.
The second statement holds with regards to any squares (in fact, any rectangles) and monotone sequences of points.
The third statement is trivial, given the second, as if two squares cover the same point, they intersect.
The proof of the fourth statement may be found in Figure~\ref{fig:monotone}; it requires more than the monotone nature of the centers, but also that each square is contained in a quadrant which is part of a monotone path.
The last point is really the same as the second, using the fact that the first point holds for an arbitrary point in $\Square{(\Pathbottom)}$.
\end{proof}

\begin{figure} 
    \centering
    \includegraphics{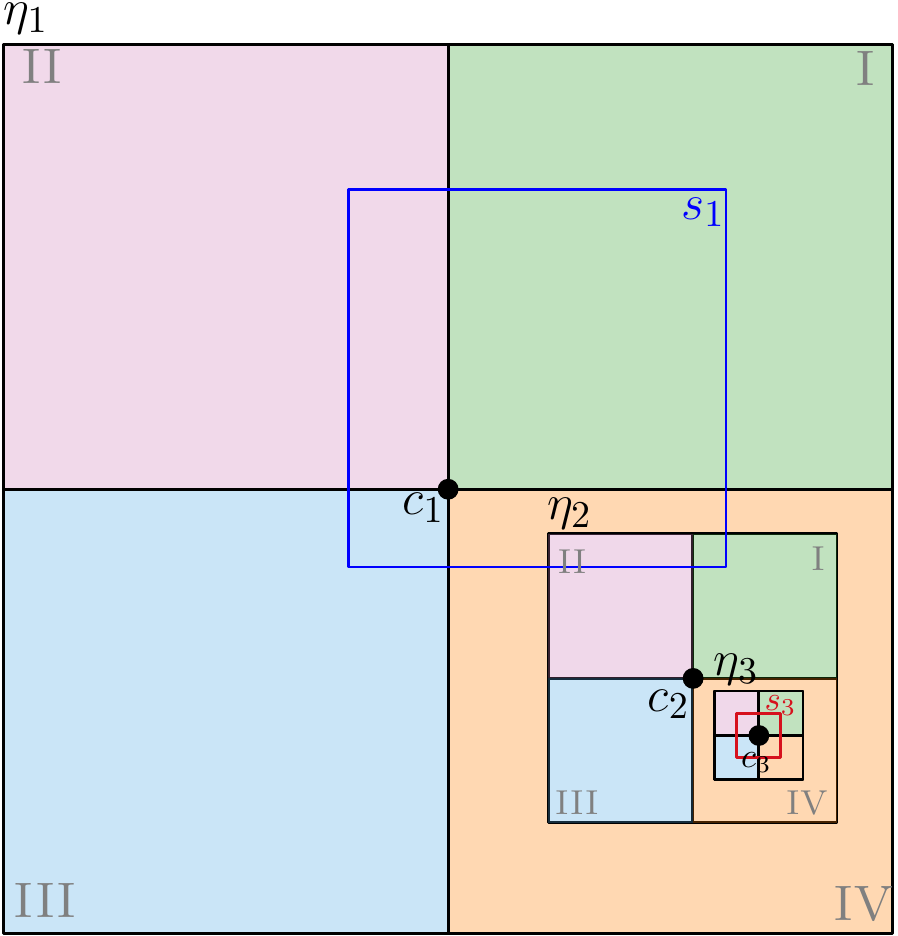}
    \caption[fuck]{
     We wish to prove the penultimate point of Lemma~\ref{l:monotone}: 
     given squares  $s_1,s_3 \in \Squares{\Pathq{q}}$ 
     where
     $\Interval{\Pathq{q}}{s_1}$,  $\Interval{\Pathq{q}}{s_2}$ are disjoint, and 
 $\Interval{\Pathq{q}}{s_1}  < \Interval{\Pathq{q}}{s_3}$.
 Let $\node_1 \coloneqq \Depth{s_1}$, $\node_3 \coloneqq \Depth{s_3}$ and let $\node_2$ be a node between $\node_1$ and $\node_3$ on $\Pathq{q}$.
     Observe that the depths of the nodes are increasing: $\Depth{\node_1}<\Depth{\node_2}<\Depth{\node_3}$.
     Assume w.l.o.g., that $q=\qd$. 
     Three quadtree cells, $\node_1, \node_2, \node_3$, that are a subsequence of $\Pathq{\qd}$, are illustrated with their centers $c_1=\Squarecenter{\node_1},c_2=\Squarecenter{\node_2},c_1=\Squarecenter{\node_3}$. 
         Observe that cell $\node_2$ is contained in the lower-right quadrant of $\node_1$, and $\node_3$ is in the lower-right quadrant of $\node_2$. 
 Square $s_1$ and $s_3$ are illustrated. Since $\Interval{\Pathq{q}}{s_1} < \Interval{\Pathq{q}}{s_2}$, we know that $s_1$ does not intersect $c_2$.
  Additionally $s_3$ by definition is contained in $\node_3$.
         The basic geometric fact we wish to illustrate, from which the penultimate point of Lemma~\ref{l:monotone} follows, is that since $s_1$ does not intersect $c_2$, then it cannot intersect the red square $s_3$. This follows as relative to $c_2$, $c_1$ is to the upper-left, and any point in $\node_3$ is to the lower right, and any axis-aligned square including points both to the upper-left of a point and the lower-right, must include the point itself. Thus, as the blue square $s_1$ must by definition include $c_1$, if it is to intersect the red square which lies entirely in $\node_3$, it must intersect $c_2$.
        }
    \label{fig:monotone}
\end{figure}

We call some $I \subseteq \Squares{\Pathq{q}}$ \emph{interval independent} if the intervals $\{ \Interval{\Pathq{q}}{s}  | s \in I\}$ are independent. We would like to connect the notion of interval independent squares with independent squares. However, there exist sets of squares which are interval independent but yet have intersecting squares, see Figure\ref{f:pathb} for an example. 

Let $\Ipath{\Pathq{q}}$ be an interval independent subset of $\Squares{\Pathq{q}}$. Let
 $\Half{\Pathq{q}}{\Ipath{\Pathq{q}}}$ be a subset formed by $\Ipath{\Pathq{q}}$ taking every other element out of $\Ipath{\Pathq{q}}$, starting from the last (deepest) one. As $\Ipath{\Pathq{q}}$ is interval independent, the depths of the squares are distinct, and $\Half{\Pathq{q}}{\Ipath{\Pathq{q}}}$ is uniquely defined.
 Clearly,  $|\Half{\Pathq{q}}{\Ipath{\Pathq{q}}}| \geq \frac{1}{2}|\Ipath{\Pathq{q}}|-1 $
 
 \begin{lemma}\label{l:isprotected}
Given an interval-independent set $\Ipath{\Pathq{q}} \subseteq  \Squares{\Pathq{q}}$, the set $\Half{\Pathq{q}}{\Ipath{\Pathq{q}}}$ is a protected independent set.
 \end{lemma}
 
 \begin{proof}
 First, we argue that this is an independent set. As we took the squares corresponding to every other interval, by the penultimate point of Lemma ~\ref{l:monotone}, they are independent. 
From the last point of Lemma~\ref{l:monotone}, the only square in $\Squares{\Pathq{q}}$ that could intersect $\Pathbottom$ is the last one, and by definition this is not included in $\Half{\Pathq{q}}{\Ipath{\Pathq{q}}}$.
 \end{proof}
 
 \begin{lemma} \label{l:indsetsize}
Given $P_q$, let $\Iintopt{\Pathq{q}}$   be a maximum interval-disjoint subset of $\Squares{\Pathq{q}}$. 
Let $\Iintapprox{\Pathq{q}}$ be an interval-disjoint subset of $\Squares{\Pathq{q}}$   where $c|\Iintapprox{\Pathq{q}}|\geq |\Iintopt{\Pathq{q}}|$ for some $c\geq 1$.
Then, $|\Half{\Pathq{q}}{\Iintapprox{\Pathq{q}}}| \geq  \frac{1}{2c}\OPT{\Pathq{q}}-1$.
 \end{lemma}
 
 \begin{proof}
 
 \begin{align*}
     |\Half{\Pathq{q}}{\Iintapprox{\Pathq{q}}}| &  \geq \frac{1}{2} |\Iintapprox{\Pathq{q}}|-1 & \text{From taking every other element}
     \\
     & \geq \frac{1}{2c}|\Iintopt{\Pathq{q}}|-1 & \text{Given}
     \\
     & \geq \frac{1}{2c}\OPT{\Pathq{q}}-1
 \end{align*}
 The last line follows from the third point of Lemma~\ref{l:monotone}, which implies that if a subset of $\Squares{\Pathq{q}}$ is interval-disjoint, then it independent; thus the size of the maximum interval-disjoint subset of $\Squares{\Pathq{q}}$, $|\Iintopt{\Pathq{q}}|$, is a lower bound on the size $\OPT{\Pathq{q}}$ of a maximum independent set of $\Squares{\Pathq{q}}$. 
 \end{proof}

\subsection{Summary}

\begin{theorem} \label{t:static}
Given a $c$-approximation algorithm for maximum independent set of intervals, one obtain an expected $256c+32$-approximate randomized algorithm for maximum independent set of squares.
\end{theorem}

\begin{proof}
For each $P$ in $\Qpaths$, and $q \in \{\qa,\qb,\qc,\qd\}$ we have the intervals $\{\Interval{\Pathq{q}}{s}| s \in \Squares{\Pathq{q}}\}$. 
Give each of these sets of intervals to the assumed solution for an approximate independent set of intervals, and let
$\Iintapprox{\Pathq{q}}$ be the squares that give rise to the solution.
Now by removing every other element of $\Iintapprox{\Pathq{q}}$, beginning with the last, we obtain $\Half{\Pathq{q}}{\Iintapprox{\Pathq{q}}}$. 
From Lemma~\ref{l:isprotected} we know that $\Iintapprox{\Pathq{q}}$ is a protected independent set of squares, not just intervals. As for the size of $\Iintapprox{\Pathq{q}}$:

\begin{align*}
  &  \max_{q \in \{\qa,\qb,\qc,\qd\} } \sum_{\Path \in \Qpaths} |\Half{\Pathq{q}}{\Iintapprox{\Pathq{q}}}|
\\
&\geq \max_{q \in \{\qa,\qb,\qc,\qd\} } \sum_{\Path \in \Qpaths} \lrp{\frac{1}{2c}\OPT{\Pathq{q}}-1 } & \text{Lemma~\ref{l:indsetsize}}
\\
 &\geq \lrp{ \max_{q \in \{\qa,\qb,\qc,\qd\} } \sum_{\Path \in \Qpaths}  \frac{1}{2c}\OPT{\Pathq{q}}}-|\Qpaths|
 \\
& 
\geq  \frac{1}{8c} \sum_{\Path \in \Qpaths} I_{\text{opt}}(\Path) -|\Qpaths| & \text{Fact \ref{f:max4}}
\end{align*}

Thus, for 
$$q_{\max} \coloneqq \underset{q \in \{\qa,\qb,\qc,\qd\}}{\operatorname{argmax}} \sum_{\Path \in \Qpaths} |\Half{\Pathq{q}}{\Iintapprox{\Pathq{q}}}|,$$ 
the sets $\Half{\Pathq{q}}{\Iintapprox{\Pathq{q}}}$ will satisfy the requirements of Lemma~\ref{l:combine} with $c_1=\frac{1}{8c} $ and $c_2=1$. Thus, the sets $\Half{\Pathq{q}}{\Iintapprox{\Pathq{q}}}$ along with an arbitrary square associated with each node $\Qleaves$ by Lemma~\ref{l:combine} is an independent set of squares of size at least at least $\frac{\frac{1}{8c}}{2\cdot \frac{1}{8c}+1+1}\OPT{\Centered}=\frac{1}{2+16c} \OPT{\Centered}$; from Lemma~\ref{l:bds} this is expected to be at least $\frac{1}{2+16c}  \cdot \frac{1}{16} \OPT{S}=  \frac{1}{256c+32}  \OPT{S}$.
\end{proof}

%% file: 4-dynamicSquares.tex
\section{Dynamization}
\label{sec:squares_dynamic}

\newcommand{\jti}[1]{\textit{#1}}

\newcommand{\ttt}{\tabto{1pc}}
\newcommand{\tttt}{\tabto{2pc}}
\newcommand{\ttttt}{\tabto{3pc}}
\newcommand{\tttttt}{\tabto{4pc}}
\newcommand{\ttttttt}{\tabto{5pc}}
\newcommand{\tttttttt}{\tabto{6pc}}
\newcommand{\this}{\text{this}}
\newcommand{\jtc}[1]{\tabto{2.7in}\parbox[t]{2.7in}{\textit{// #1}}}

\newcommand{\Ss}[1]{\jmark{S}^{\jmark{\text{search}}}_{#1}}
\newcommand{\Si}[1]{\jmark{S}^{\jmark{\text{interval}}}_{#1}}
\newcommand{\Sp}{\jmark{S}^{\jmark{\text{path}}}}
\newcommand{\Sq}[1]{\jmark{S}^{\jmark{\text{qtree}}}_{#1}}
\newcommand{\Sm}{S^{\jmark{\text{qman}}}}
\newcommand{\St}{S^{\jmark{\text{top}}}}

In the dynamic structure, squares are inserted and deleted into an initially empty set of squares, and any changes to an approximate maximal independent set are reported back. Thus, the user of the dynamic structure, by keeping track of all changes that have been reported, will know the current contents of the independent set.

This data structure is a dynamization of the quadtree approach presented in section~\ref{s:quadtree}.
We present the data structure in three parts: 

\begin{enumerate}
    \item The \textit{quadtree structure} is the main structure and which efficiently stores the quadtree and its decomposition into paths, leaves, and internal nodes, and brings together the independent sets of the various parts as previously described in Lemma~\ref{l:combine}. 
 
    \item The \textit{path structure} represents a path, and there will be one such structure maintained by the quadtree structure for each path in $\Qpaths$.This structure translates each square stored into an interval in the same manner as the static structure, and uses the dynamic interval structure of section~\ref{sec:intervals_details} for each of the four monotone paths to dynamize maintaining an approximately optimal set of intervals. This requires swapping out the IQDS of \ref{lem:ors} for one that is compatible with the intervals generated from our quadtree approach; the details of this are presented in Section~\ref{sec:squaresiqds}.

    \item The \textit{search structure} is the secret sauce of the efficiency of our entire method, and allows the dynamic interval structure to query the intervals of squares which are stored in this structure implicitly.
\end{enumerate}

We proceed in a bottom-up fashion here, beginning with the search structure and culminating with the quadtree structure, with the details of the IQDS at the end.
\subsection{Search Structure}

\paragraphh{Overview. }
This structure is the only part of the more complex structure which does not directly correspond to something defined in Section~\ref{s:quadtree}. Rather, it plays a supporting role and is vital for speed reasons. There will be one instance of the search structure, which is globally accessible.

Logically, the searching structure stores a collection of squares which can be modified via $\jop{Insert}$ and $\jop{Delete}$ operations. Each square in the collection has a mark, which will be one of the quadrant labels $\qa,\qb,\qc,\qd$ or $\jop{None}$. 
Given a square $s$, the quadtree structure will maintain that $s$'s mark will be $\jop{None}$ if $\Node{s}$ is currently an internal node or a leaf in $\Qtree$, and will be the quadrant of the child if $\Node{s}$ is monochild.
There is one query operation, and one operation to change the marked state of some of the squares, which will be used by the query structure when a change in the quadtree necessitates a change in the mark of some squares.

We use $s.t, s.b, s.l, s.r$ to refer to the coordinates of the four sides of a square $s$, and $d$ to refer to one of the directions in $\{t,b,l,r\}$.

\paragraphh{Operations: } Formally, the operations supported are:

\begin{itemize}
    \item $\jop{Init}$: Makes new empty structure
    \item $\jop{Insert}(s,m)$: Inserts square $s$ with mark $m$ into the structure
    \item $\jop{Delete}(s)$: Deletes a square $s$ from the structure
    \item $\jop{RangeMark}(s_1,s_2,m)$: Gives mark $m$ to all squares $s$ stored such that $s_1.d \leq s.d \leq s_2.d$ for $d\in \{t,b,l,r\}$.
    \item $s=\jop{RangeSearch}(s_1,s_2,m)$ Finds and returns a/the marked square with mark $m$ stored such that $s_1.d \leq s.d \leq s_2.d$ for all $d\in \{t,b,l,r\}$ or reports that there is no such square.
\end{itemize}

\paragraphh{Invariants:} The quadtree structure will ensure that the following holds at the end of every operation: All squares in $\Centered$ must be stored in this search structure, and thus the insertion and deletion operations must be called when there are changes to $\Centered$.
For each square $s$, its mark is the quadrant of the child if $\Node{s}$ has is monochild (if $\Node{s}$ is on some $\Path \in \Qpaths$) and $\jop{None}$ otherwise (if $\Node{s}$ is a leaf or an internal node).
As the structure of the quadtree changes, the quadtree structure will need to call $\jop{RangeMark/RangeUnmark}$ to ensure that this invariant is kept up to date.
For example,
to mark in all squares in $\Squares{\node}$ with mark $m$ can be done with one call to $\jop{RangeMark}(s_1,s_2,m)$ where $s_1$ is the lower-left quadrant of the square of $\Square{\node}$ and $s_2$ is the upper-right quadrant of the square of $\Square{\node}$.

\paragraphh{Implementation. } 
Each square is stored as a 4-D point in a standard range tree structure~\cite{DBLP:journals/ipl/Bentley79}. The range tree is augmented with a possible note with a mark on each node (in this paragraph a node is used to refer to a node of a range tree, not of the quadtree) indicating that all points in its subtree should have that mark. 
For any point,  its mark is determined by its highest ancestor with such a note.
Any time a non-leaf node with a note is touched, its note is removed and pushed down to its children, overwriting any note on its children. 
Additionally, all nodes have an indication of whether all nodes squares in this subtree have the same mark, according to the information in the subtree (ignoring any notes in the ancestors).
With such standard augmentation, standard range query operations can be executed, such operations may be limited to a particular mark, and all points in a range can have their mark changed.

\paragraphh{Runtime.} All operations take time $O(\log^4 n)$ using the standard range tree analysis. We do not bother with fractional cascading~\cite{DBLP:journals/algorithmica/ChazelleG86,DBLP:journals/algorithmica/ChazelleG86a} as while this may shave a log, this complicates the simple description of the implementation of the marks.

\subsection{Path Structure}

The path structure is a structure which represents a path $P \in \Qpaths$ and which maintains an approximate protected independent set of the squares associated with nodes of the path. As such there will be one instance of this structure for every path. Though the details are numerous, the idea behind the path structure is simple: it stores the squares on a path, supports modifications of the path such as split and merge, and translates the squares into intervals which it passes on to our structure for dynamic intervals. 

The path structure is called by the quadtree structure every time a square is inserted or deleted that is associated with a node on the path. It is also called by the quadtree structure whenever a split, merge, extend, or contract operation needs to be performed due to the paths changing due to a structural change in the quadtree. In this way the quadtree structure ensures that there is always one path structure for each path in $\Qpaths$. 
We require that the quadtree structure updates the search structure before the path structures, and  the path structure has access to the search structure.

The ADT of the path structure is as follows:
\begin{itemize}
    \item $ \jop{Init}(\node_{\text{top}},\node_{\text{bottom}})$: Creates a new path structure for the path from $\node_{\text{top}}$ to $\node_{\text{bottom}}$.
    \item $\Delta I = \jop{Insert}(s)$: This is called by the quadtree structure whenever a square $s$ is inserted that is associated with a node on this path. Changes to the independent set are reported.
    \item $\Delta I = \jop{Delete}(s)$: This is called by the quadtree structure whenever a square $s$ is deleted that is associated with a node on this path. Changes to the independent set are reported.
    \item $(\Delta I,\Path_\text{new}) = \jop{Split}(\node)$: Given a node $\node$ on this $\Path$, splits this structure into two. This structure will represent the path defined from the former $\Pathtop$ to $\node$ and the new path will represent the path defined from $\node$ to the former $\Pathbottom$. The node $\node$ itself will no longer belong to a path. Changes to the protected independent sets are reported, that is the difference between the independent set before this operation and the union of the two paths'  structures protected independent at the completion of this operation. 
    \item $\Delta I = \jop{Merge}(\Path')$: Given another instance of an monotone path structure $P'$ which is adjacent below this one, that is, where $\Pathbottom=\Pathtop'$, the two paths and node between them, $\Pathbottom=\Pathtop'$, are combined into this structure and $\Path'$ becomes invalid. Changes to the union of protected independent sets before the operation as compared to the single one after are reported.
    \item $\Delta I = \jop{Extend/Contract}(\node)$: Set $\Pathbottom$ to $\node$. There must be no marked nodes between the old and the new $\Pathbottom$. Extend refers to $\node$ being below the current $\Pathbottom$, thus making the path larger, and contract refers to $\node$ being on the current path, thus making it smaller.  Changes to the independent set are reported.
\end{itemize}

Note that in a \jop{Extend}, \jop{Contract}, \jop{Split}, and \jop{Merge}, there is one possibly marked node that is added or removed from the path(s) that was not on a path before. However, the squares associated with this node are not passed in as a parameter as there could be arbitrarily many, but the implementation of these operations uses the search structure to access them; this works as at most one of the squares associated with this node can be added or removed to the independent set, and searching for this one square can be done with the search structure. 

\paragraphh{Implementation overview.}
The implementation works by following the same logic of the static case. To review, in the static case the path $\Path$ was partitioned into four monotone subpaths $\Pathq{q}$, for $q \in \{\qa,\qb,\qc,\qd \}$. 
In each of these subpaths, squares $s$ associated with nodes on the monotone subpath $\Pathq{q}$ were associated with intervals $\Interval{\Pathq{q}}{s}\coloneqq [\Depth{s},\Depthmax{\Pathq{q}}{s}]$ corresponding to the depths of the centers of the nodes on the subpath that they spanned. 
An approximately maximal set of squares whose intervals were independent, $\Iintapprox{\Pathq{q}}$, was then obtained. 
This set is not necessarily a protected independent sets of squares, but we showed it could easily be transformed into one by removing every other square, starting from the deepest, to yield $\Iapprox(\Pathq{q})$. 
Finally, the independent set chosen, $\Iapprox(\Path)$ was the largest of the four $\Iapprox(\Pathq{q})$.

In order to make this dynamic several minor changes are needed. The first is that to maintain the four approximate independent sets of intervals, $\Iintapprox{\Pathq{q}}$, the dynamic intervals structure of section~\ref{sec:intervals_details} is used.
The second is that the condition that the chosen independent set be that largest of $\Iintapprox{\Pathq{q}}$ is too strict as we need to report changes in the independent set, and if the largest oscillates frequently between two large sets, reporting these changes will become unacceptably expensive.
So, instead, we require that the chosen independent set be within a factor of two of the largest $\Iintapprox{\Pathq{q}}$, which allows an easy amortization of the cost of switching sets at the expense of losing let another factor of two in our constant.
The third change is that taking every-other interval from $\Iintapprox{\Pathq{q}}$ to yield $\Iapprox(\Pathq{q})$ is too strict to be maintained dynamically. We thus adopt a more relaxed approach where there are between one and three unchosen intervals between the chosen intervals, thus losing another factor of two in the approximation.
The fourth change is that the paths themselves are not static, and may be split, merged, extended or contracted as the quadtree shape changes. This is easy to support given that the dynamic independent interval structure support splits and merges.

Specifically, we store the four interval-independent subsets of $\Squares{\Pathq{q}}$, $\Iintapprox{\Pathq{q}}$, in four red-black trees, sorted by depth. We call the nodes of non-maximal black depth \emph{chosen}. It is a basic fact of red-black trees that there will be between one and three non-chosen nodes between chosen nodes, and the minimum and maximum nodes will not be chosen. The squares in the chosen nodes are thus, by the same logic as lemma~\ref{l:isprotected}, a protected independent set. It is also useful to note that in every standard red-black tree operations (insert, delete, split, merge) only a constant number of nodes can have their chosen/non-chosen status change.

One of the four red-black trees is marked as \emph{active}, and it is the chosen squares of this tree which form the protected independent set  $\Iapprox(\Path)$ seen by the user of this structure. We will maintain that the active red-black tree is the largest of the four, or at least within a factor two of the largest. Just before finishing the execution of all operations, this invariant is checked, and if it no longer holds, the active red-black tree is changed to the largest of the four, and all the chosen squares in the previously active red-black tree are reported as removed and all the chosen squares in the newly active red-black tree are reported as added to the independent set. This will result in periodic large changes in the independent set, but these changes are easily shown to be constant amortized in the same manner as the classic array resizing problem.

As for the implementation of the operations, this comes down to maintaining the four approximate interval-independent subsets in red-black trees, and four dynamic interval structures.
In Lemma~\ref{lem:ors}, the dynamic interval structure listed the queries that it needs to be able to preform on intervals, the IQDS operations, and we will show in Section~\ref{l:squaresiqds} how to implement them. 

The $\jop{Extend}$ and $\jop{Merge}$ operations require creating a new interval structure to be merged into the existing one(s). For the $\jop{Extend}$ operation, this corresponds to the squares of a newly added node. For the  $\jop{Merge}$ operation, this corresponds to the squares of the node defining the top of the one path and the bottom of the other, and is thus part of neither. Fortunately, in both cases, an arbitrary square from the node can be chosen to be the independent set, and this is maximal. During these operations care must be taken when splitting or merging the dynamic interval structures associated with the paths as during these operations the intervals may change. In section~\ref{s:splitmergeintervals} we show the details of how these splits and merges are executed.

We now summarize how each operation is executed. 

\noindent To execute $\Delta I = \jop{Insert}(s)$:

\newcommand{\Je}{\item If quadrant $q$ is not the active one set $\Delta I = []$, as changes in non-active independent sets are recorded but not returned.}
    \newcommand{\Jf}{
    \item If the red-black tree of the active quadrant no longer has size within a factor-2 of the largest of the four red-black trees, append to $\Delta I$ $(s,\jop{Delete})$ for all squares $s$ in marked nodes in the active quadrant, change the active quadrant to that of the red-black tree with largest size, and append to $\Delta I$ $(s,\jop{Insert})$ for squares $s$ in marked nodes in the new active red-black tree.
        \item Return $\Delta I$.}
 
 \newcommand{\Jd}{This returns changes to the independent sets of intervals $\Iintapprox{\Pathq{q}}$, and the red-black tree storing $\Iintapprox{\Pathq{q}}$ is updated accordingly. 
    Append to $\Delta I$ insertions and deletions of squares to reflect any changes to the squares stored in the marked nodes in the red-black tree.}

 \newcommand{\Jdd}{This returns changes to the independent sets of intervals $\Iintapprox{\Pathq{q}}$, and the red-black trees storing $\Iintapprox{\Pathq{q}}$ are updated accordingly.
    Append to $\Delta I_q$ insertions and deletions of squares to reflect any change in the squares in the red/black tree.}

\begin{itemize}[noitemsep]
    \item Set $\Delta I = []$
    \item Let $q$ be the quadrant of the child of $\Node{s}$ which we know is monochild. 
    \item Let $\Interval{\Pathq{q}}{s}\coloneqq [\Depth{s},\Depthmax{\Pathq{q}}{s}]$. This is an IQDS operation and the implementation is discussed in Section~\ref{sec:squaresiqds}.
    \item Insert $\Interval{\Pathq{q}}{s}$ into the quadrant $q$ dynamic interval structure.  \Jd
    \Je
    \Jf
\end{itemize}
    
\noindent
To execute $\Delta I = \jop{Delete}(s)$:
\begin{itemize}[noitemsep]
    \item Set $\Delta I = []$
    \item Let $q$ be the quadrant of the child of $\Node{s}$, which we know is monochild. 
    \item Let $\Interval{\Pathq{q}}{s}\coloneqq [\Depth{s},\Depthmax{\Pathq{q}}{s}]$. This is an IQDS operation and the implementation is discussed in Section~\ref{sec:squaresiqds}.
    \item Delete $\Interval{\Pathq{q}}{s}$ from the quadrant $q$ dynamic interval structure. 
    \Jd
     \Je
\Jf

\end{itemize}

\noindent
The steps to execute $\Delta I = \jop{Extend}(\node)$:

\begin{itemize}[noitemsep]
   \item Set $\Delta I_q = []$, for $q \in \{\qa,\qb,\qc,\qd\}$.
    \item Let $q$ be the quadrant of $\Pathbottom$ that $\node$ lies in. 
    \item Let $s$ be a the first square in the linked list storing $\Squares{\node}$
    \item Let $S'$ be a new interval data structure containing the interval of $s$ with respect to the new path $\Pathtop$ to $\node$ and quadrant $q$. This is an IQDS query.
    \item Add $s$ to the red-black tree storing $\Iintapprox{\Pathq{q}}$. 
     Append to $\Delta I$ insertions and deletions of squares to reflect any change in the squares stored in marked node in the red-black tree.
    \item Use the method in section \ref{s:splitmergeintervals} to merge the new data structure $S'$ with the existing interval data structure. 
\Jd
    \item Set $\Delta I = \Delta I_q$, where $q$ is the active set.

\Jf

\end{itemize}

\noindent
The steps to execute $\Delta I = \jop{Contract}(\node)$:

\begin{itemize}[noitemsep]
   \item Set $\Delta I_q = []$, for $q \in \{\qa,\qb,\qc,\qd\}$.
    \item Let $q$ be the quadrant of $\node$ that $\Pathbottom$ lies in. 
    \item If there is a square $s$ in red-black tree $q$ storing $\Iintapprox{\Pathq{q}}$ associated with node $\node$, remove it. Append to $\Delta I$ insertions and deletions of squares to reflect any change in the marked squares in the red/black tree, and the removal of $s$.
    \item Use the method in section \ref{s:splitmergeintervals} to split the existing interval data structure at the depth of at least $\node$. Discard the right structure. 
    \Jd
    \item Set $\Delta I = \Delta I_q$, where $q$ is the active set.

     \Jf   

\end{itemize}

\noindent 
To execute $\Delta I = \jop{Merge}(\Path')$:

\begin{itemize}[noitemsep]
    \item Set $\Delta I_q = []$, for $q \in \{\qa,\qb,\qc,\qd\}$.
    \item Let $q$ be the quadrant of $\Pathbottom$ that $\Path'$ lies in. 
    \item Let $s$ be a square in $\Squares{\Pathbottom}$, if it is not empty.
    \item Let $S'$ be a new data interval data structure containing the interval of $s$ with respect to the new path $\Pathtop$ to $\node$ and quadrant $q$.
    \item Let $S''$ be the interval data structure of $\Path'$
    \item Add $s$ to the red-black tree storing $\Iintapprox{\Pathq{q}}$. Merge the four RB trees of this path with those of $P'$.
    Append to $\Delta I_q$ insertions and deletions of squares of marked nodes in the quadrant-$q$ red-black tree to reflect any change in the marked squares in the red-black trees.
    
    \item Use the method in section \ref{s:splitmergeintervals} to merge the quadrant-$q$ dynamic interval structure with new data structure $S'$ and then merge all four interval data structure with those of $S''$. 
    \Jdd
        \item Set $\Delta I = \Delta I_q$, where $q$ is the active set.
    \Jf

\end{itemize}

\noindent 
To execute $(\Delta I,\Path_\text{new}) = \jop{Split}(\node)$:

\begin{itemize}[noitemsep]
    \item Set $\Delta I_q = []$, for $q \in \{\qa,\qb,\qc,\qd\}$.
    \item Use the method in section \ref{s:splitmergeintervals} to split the intervals of the four dynamic interval structures each into three groups, based on those with depths at least, at most, and equal to that of $\node$. Discard the middle structure, and place the at least group into a newly created path structure $\Path_\text{new}$. 
    This returns changes to the independent sets of intervals $\Iintapprox{\Pathq{q}}$, and the red-black trees storing $\Iintapprox{\Pathq{q}}$ are updated accordingly. Append to $\Delta I_q$ insertions and deletions of squares to reflect any change in the squares in the red-black trees.
    \item  Delete from the red-black trees the (possibly) single square associated with $\node$. 
    Split the four RB trees of this path based on $\node$ and move the larger part to $P'$.
    \Jdd
    \item Set $\Delta I = \Delta I_q$, where $q$ is the active set.
   \Jf
\end{itemize}

\begin{lemma} \label{l:dynpath}
For each path structure:
\begin{itemize} 
\item For each $q$ the path structure maintains a protected subset $\Iintapprox{\Pathq{q}}$ of $\Squares{\Pathq{q}}$ where
$\Iintapprox{\Pathq{q}} \geq  \frac{1}{8} \OPT{\Squares{\Pathq{q}}} - 3$. 
\item 
The path structure maintains a protected subset $\Iapprox(\Path)$ of $\Squares{\Path}$ where
$\Iapprox(\Path) \geq  \frac{1}{64} \OPT{\Squares{\Path}} - 3$. 
\item The runtime of each operation is at most $O(\log^5 n)$ amortized.
\end{itemize}
\end{lemma}

\begin{proof}
With the exception of the change of the active red-black tree, the implementation only requires a constant number of operations and calls to the operations of the interval structure. In turn, the interval structure only makes a constant number of calls to the IQDS operations listed in Lemma~\ref{lem:ors}. These are executed with calls to the search structure in time $O(\log^5 n)$ each, as described in section~\ref{sec:squaresiqds}.

The dynamic interval structures only report a constant number of changes of the approximate maximum set of intervals per operation. Thus, by standard amortization arguments, the large cost incurred by switching the active red-black tree is only constant amortized per operation (use as a potential function the difference between the size of the largest red black tree and the active one times a constant).

Adjusting for the use of the red-black tree instead of every-other node gives a set of size at least $\frac{1}{4c} \OPT{\Squares{\Pathq{q}}} - 3$ assuming a $c$-approximate independent set returned by the dynamic path structure, as it will take at least every fourth square, missing up to three at the end. 
As for the approximation factor for dynamic intervals, using Lemma~\ref{l:indsetsize} with $c=2$ (an easy upper bound on $1+\epsilon$) gives the first point.

From the first point to the second point, we are converting from one quadrant to the best of the four, which from Fact~\ref{f:max4} says we lose a factor of four, but as in the dynamic case we are actually only taking something which is within a factor two of the best of the four, so lose a factor of 16 when moving from the first to the second point.
\end{proof}

\subsection{Quadtree Structure}
The quadtree structure is the main data structure. It maintains an approximate maximum independent set of squares under the insertion and deletion operations. In these operations it returns any changes to the independent set.

In order to do this, the quadtree structure contains and maintains one path structure for each path in $\Qpaths$. The independent set it logicicaly maintains follows the static case and is the union of the approximate maximal protected independent sets from these path structures and a single square associated with each leaf in $\Qleaves$.

As such, the quadtree structure needs to make sure that there is one path structure for each path $\Path$ in $\Qpaths$, and that they are informed via their insert and delete methods of any changes in the sets $\Squares{\Pathq{q}}$. Additionally, as squares are inserted and deleted, the shape of the quadtree $\Qtree$ and thus the set $\Qpaths$ may change, and thus the quadtree structure needs to call the various operations on the path structures to reflect any such changes in the composition of $\Qpaths$.

Also, as the structure of the quadtree changes, $\Qleaves$ may change and thus any changes to the part of the independent set composed of one element of $\Squares{\node}$ for each $\node\in \Qleaves$ needs to be reported.

Finally, the quadtree structure needs to maintain the invariants of the searching structure. This needs to be called for the insertion and deletion of any squares in $\Centered$, which is easy enough. 
However, each square in the search structure is marked based on whether $\node(s)$ is a monochild, and if so, which quadrant the child is. 
As the shape of a quadtree changes, this could result in the marked status of many squares changing, and the range marking method of  searching structure need to be called.

\paragraphh{Operations: }
\begin{itemize}
    \item $\Delta I = \jop{Insert}(s)$: Inserts square $s$ and any changes to the independent set are reported
    \item $\Delta I = \jop{Delete}(s)$: Deletes square$s$ and any changes to the independent set are reported
\end{itemize}

\paragraphh{Implementation overview:} 
The implementation of the quadtree structure is the most complex part of the data structuring required, but at the same time the most standard. This is because in order to maintain the path structures and the leaves, of course one needs to know the shape of the quadtree. As the size of $\Qtree$, measured in nodes, is unbounded, it is stored in the standard compressed way where unmarked monochild nodes are contracted; such a tree has at most linear size in the marked nodes. This compressed quadtreee is explicitly stored. In each node, the squares associated with each node are stored in a linked list. 

We need to be able to identify, given a newly inserted square, where is its node. Is it in the existing quadtree, and if not where should it be added? As the quadtree is not necessarily balanced, and we are seeking polylogarithmic time, additional data structuring on top the quadtree is needed. Secondly, for each node in the quadtree we wish to logically maintain a pointer to a path structure to the representing the path it is on, or is the bottom of. However, if these pointers are maintained in the obvious way by having a single pointer from each node, a single structural change in the quadtree could cause a linear number of these pointers to change (see 3(c) Figure~\ref{fig:insertRect}).  However, fortunately, the solution to both of these algorithmic issues is the Link-Cut tree of Sleator and Tarjan \cite{DBLP:journals/jcss/SleatorT83}. We maintain a link-cut tree structure on top of the compressed quadtree. This allows the searching of the quadtree, and the obtaining and changing of the path pointers to be effectuated in logarithmic time. 

\begin{figure}
    {\centering
    \includegraphics[scale=1]{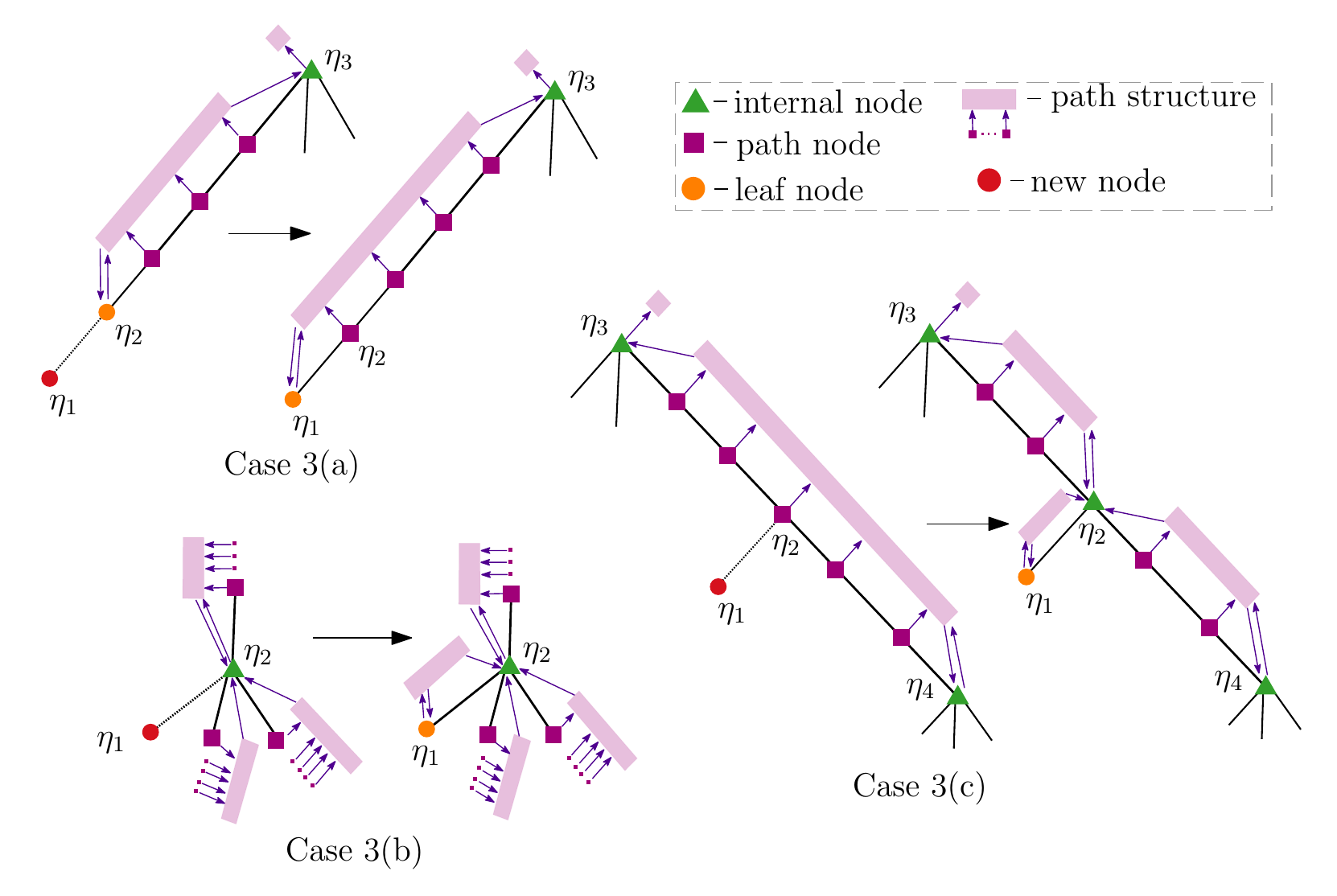}}
    \caption[fuck]{Insertion into the quadtree. The new node is $\node_1$, the LCA of the new node and the existing tree is $\node_2$. Three cases are illustrated, depending on whether $\node_2$ is a leaf, internal node, or monochild node. All cases begin with the new square being added to the searching structure with mark $\jop{None}$.

Each path structure is is illustrated, which represents a maximal sequence of monochild nodes. Each path structure has pointers to the two nodes $\Pathtop$ and $\Pathbottom$ in between which the nodes on the path lie. All nodes have a pointer to the path that they lie on or are $\Pathbottom$ of.
        }
    \label{fig:insertRect}
\end{figure}

\paragraphh{Insertion.} We now can describe in detail the insertion process of on new square $s$. See Figure~\ref{fig:insertRect} for an overview.

First the coordinates of $\node(s)$ are computed, which can be done with arithmetic and a discrete logarithm. Then it is checked whether $s$ is centered, whether it contains the center of $\node(s)$, if it is not the insertion procedure returns having done nothing.

Then, $\node(s)$ is either in the compressed quadtree or not. If not, this could be because it needs to be added as a new leaf, or it could be because it is in $\Qtree$ but as a monochild unmarked node and thus has been compressed. In any case the search is done using the link cut tree. Link-cut trees support a so-called oracle search in $O(\log n)$ time \cite{DBLP:journals/algorithmica/AronovBDGILS18}. That is, given some query, and a node in the tree, if one can compute in constant time whether the node is the answer to the query, and if not which connected component of the tree minus the node the query lies in, then the search runs in time $O(\log n)$. For searching for a node in a quadtree, this oracle is a simple geometric computation, as given a node a query node is in a child if it lies entirely in one of the node's quadrants, with the quadrant number indicating the child, else it lies in the part of the tree attached to the node's parent.

To summarize $\node(s)$ is searched for in the quadtree and the outcome is one of the following:
\begin{enumerate}[noitemsep,nolistsep]
    \item The node $\node(s)$ is already present in the compressed quadtree.
    \item The node $\node(s)$ is a compressed node, that is it is in $\Qtree$ but not in the compressed quadtree as it is monochild and is currently unmarked. Let $\node_1$ (above) and $\node_2$ (below) denote the nodes in the compressed quadtree bookending the insertion point.
    \item The node $\node(s)$ is not in the quadtree $\Qtree$. It should be attached to $\node_3$ which is$\ldots$
    \begin{enumerate}[noitemsep,nolistsep]
        \item a leaf.
        \item an node of $\Qinternal$.
        \item a marked monochild node between $\node_4$ (above) and $\node_5$ (below).
        \item an unmarked monochild node (a compressed node). Thus $\node_3$ is in $\Qtree$ but not the compressed quadtree, where it is between $\node_4$ (above) and $\node_5$ (below).
    \end{enumerate}
\end{enumerate}

Once the case has been established the changes need to carried out in three phases: 
Changes to the searching structure, structural changes to the paths and quadtree, and finally changes to the independent set caused by changes of the leaves. The order of the first two is crucial as the path operations require that changes to the searching structure to have already been effectuated.

For the searching structure, the new square needs to be added by calling $\jop{Insert}(s,m)$ with the appropriate mark. If $\node(s)$ is a new node, $\node(s)$ will be a leaf, and it should be added with $m=\jop{None}$, otherwise it should be added with the type of the node it is being added to. 

Additionally, squares that will have the type of their node change need to have the mark changed in the search structure. This can happen to the squares in $\Squares{\node_3}$ in case 3(a) and (c); in the first $\node_3$ changes from being a leaf to a monochild node and thus the mark on the squares in $\Squares{\node_3}$ must change from $\jop{None}$ to one of $\qa,\qb,\qc,\qd$ depending on how $\node(s)$ is attached to $\node_3$. In 3(c) $\node_3$ changes from being a monochild node to an internal node, thus the mark on the nodes in $\Squares{\node_3}$ must change to $\jop{None}$. In both cases a single call to $\jop{ChangeMarkOnSquaresOfNode}(\node_3,m)$ will suffice.

Now, structural changes need to be performed. 
In case 2, the node $\node(s)$ need to be added between $\node_1$ and $\node_2$ with a path pointer identical to the node below, $\node_2$.
In case 3(d), the node $\node_3$ needs to be added between $\node_4$ and $\node_5$, with a path pointer identical to the node below, $\node_5$.
Then in case 3, $\node(s)$ is created as a leaf attached to $\node_3$. In case 3(a), the path which had $\node_3$ as it bottom node needs to be extended to include $\node(s)$, this is done by calling $\jop{Extend}(\node(s))$ on the path of $\node_3$ and having the path pointer of $\node(s)$ point to this path. In cases 3(b-d) the new leaf $\node(s)$ has its path pointer pointing to a new empty path structure created with $\jop{Init}(\node_3,\Node{s})$ and that represents the empty path between the leaf $\node(s)$ and the now-internal node $\node_3$. 

In cases 2(c-d), the addition of $\node(s)$ causes $\node_3$ to switch from being a monochild node on a path, to an internal node. This requires that the path be split. Let $\node_6$ be the bottom node of the path of $\node_3$. We call $\jop{Split}(\node_3)$ on $\node_3$'s path and receive a new path structure $P_{\text{new}}$. We set the path pointers on all nodes from $\node_5$ to $\node_6$ to $P_{\text{new}}$. This is where the link-cut tree's management of the path pointers is crucial, as there could be many nodes from $\node_5$ to $\node_6$, but the link cut tree can change them all logically in logarithmic time. 

With the structural changes complete, the square $s$ is added to the end of the linked list of $\node(s)$. If $\node(s)$ is a monochild node, which could be the case in case (1) and is definitely true in case (2), then we are required to call $\jop{Insert}(s)$ on $\node(s)$'s path.

In all of the path operations, any reported changes to the the independent sets reported by the path structure must saved and returned. 

The third stage is to report any changes to the independent set as a result of the leaves of the quadtree changing, as one of the two components of the independent set is the first node in the linked list of all leaf nodes. In cases $3 (a-d)$ $\node(s)$ is a new leaf containing only $s$, so $s$ is reported as being added to the independent set. In case $3 (a)$, $\node_3$ was a leaf but is a leaf no longer, so the first square in $\node_3$ is reported as being deleted (Note that the $\jop{Extend}$ operation may result in this square being re-added to this set, this is intentional and is okay).

\old{
\begin{figure}
    {\centering
    \includegraphics[scale=1]{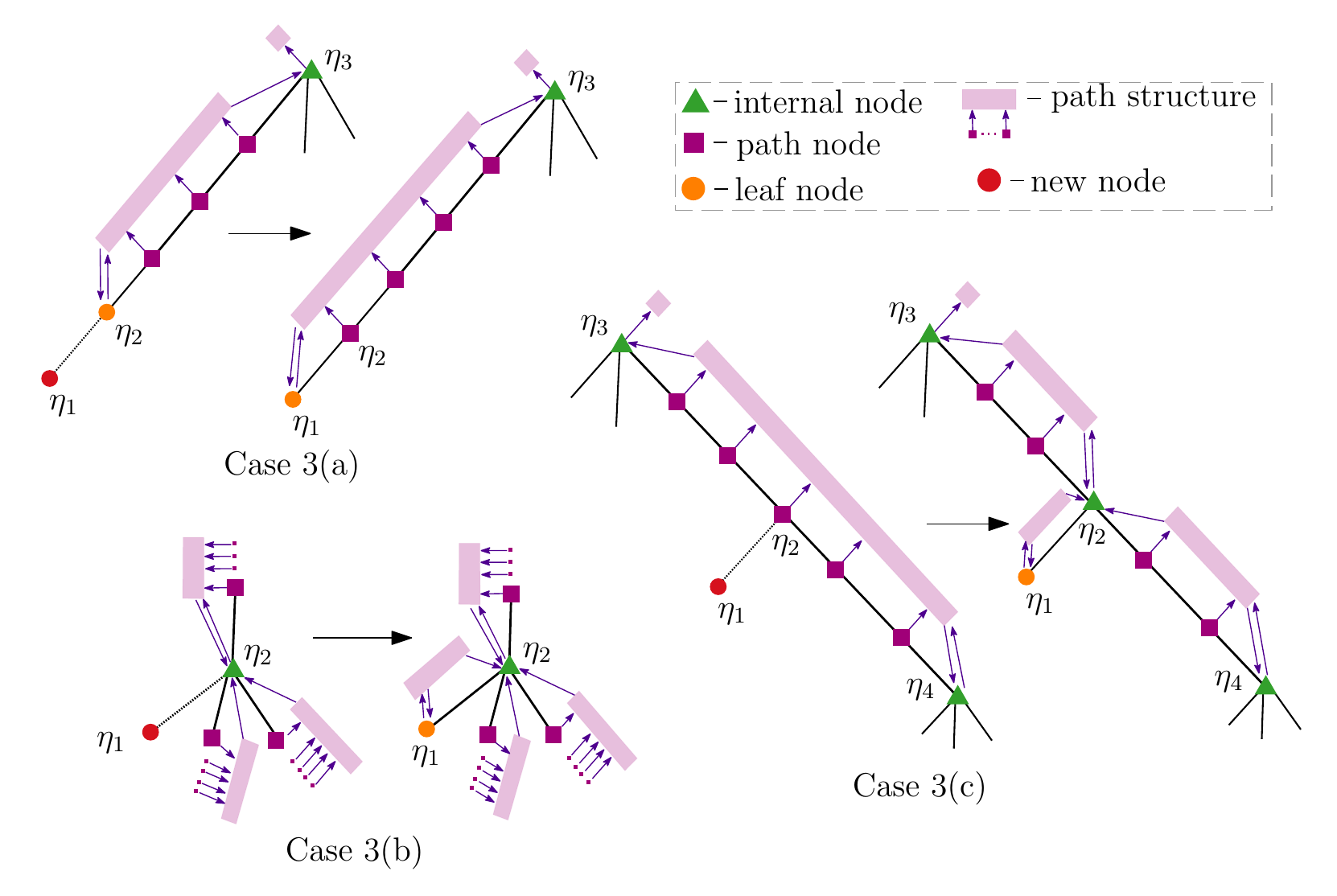}}
    \caption[fuck]{Insertion into the quadtree. The new node is $\node_1$, the LCA of the new node and the existing tree is $\node_2$. Three cases are illustrated, depending on whether $\node_2$ is a leaf, internal node, or monochild node. All cases begin with the new square being added to the searching structure with mark $\jop{None}$.

Each path structure is is illustrated, which represents a maximal sequence of monochild nodes. Each path structure has pointers to the two nodes $\Pathtop$ and $\Pathbottom$ in between which the nodes on the path lie. All nodes have a pointer to the path that they lie on or are $\Pathbottom$ of.
        }
    \label{fig:insertRect}
\end{figure}
}
\paragraphh{Deletions.}
Deletions are handled in a largely symmetric fashion which we now describe. 
As in deletions, the square $s$ is tested to see if it is centered, and if so nothing is done. Otherwise $\Node{s}$ is located in the quadtree.

\begin{enumerate}[noitemsep,nolistsep]
    \item The linked list containing $\Squares{\Node{s}}$ contains $s$ and other squares
    \item The linked list containing $\Squares{\Node{s}}$ only contains $s$
        \begin{enumerate}[noitemsep,nolistsep]
            \item $\Squares{\Node{s}}$ is a leaf with parent $\node_1$
            \begin{enumerate}[noitemsep,nolistsep]
                \item Node $\node_1$ is an monochild node.
                \item Node $\node_1$ is an internal node with two children and $\Squares{\node_2}$ is empty.
                \item Node $\node_1$ is an internal node with two children $\Squares{\node_2}$ is nonempty.
                \item Node $\node_1$ is an internal node with more than two children.
            \end{enumerate}
            \item $\Squares{\Node{s}}$ is an monochild node on a path
            \item $\Squares{\Node{s}}$ is an internal node
        \end{enumerate}    
\end{enumerate}

Once the case has been established the changes need to carried out in three phases: 
Changes to the searching structure, structural changes to the paths and quadtree, and finally changes to the independent set caused by changes of the leaves. The order of the first two is crucial as the path operations require that changes to the searching structure to have already been effectuated.

For the searching structure, the square $s$ needs to be removed by calling $\jop{Delete}(s)$.
Additionally, squares that will have the type of their node change need to have the mark changed in the search structure. This will happen in case 2(a)(i) where $\node_1$ will become a leaf and thus all squares in $\Squares{\node_1}$ need to have their mark changed to $\jop{None}$. It will also happen in case 2(a)(iii) where $\node_1$ will change from being an internal node to monochild and thus $\Squares{\node_1}$ need to have their mark changed to one of $\qa,\qb,\qc,\qd$ depending on which quadrant the other child of $\node_1$ is in.

Now, structural changes need to be performed. First $s$ is removed from the linked list of $\Node{s}$. In (1) and 2(a)(iii) the quadtree structure remains untouched. In 2(a)(iii) and 2(b) a node on a path loses its last square and thus becomes unmarked and becomes compressed; the compressed quadtree is called to reflect this, but the paths do not change.
In 2(a)(i), as $n_1$ is now a leaf, $\jop{Contract}(n_1)$ is called on its path. The most complicated cases are 2(a)(ii-iii) where the deletion causes $\node_1$ to stop being an internal node and it is merged via $\jop{Merge}$ with the paths above and below.

The third stage is to report any changes to the independent set as a result of the leaves of the quadtree changing, as one of the two components of the independent set is the first node in the linked list of all leaf nodes. This only occurs in 2(a)(ii) where $\node_1$ becomes a leaf and so the first square of $\Squares{\node_1}$ should be reported as being added to the independent set.

\paragraphh{Main result.} 

\begin{theorem}
The quadtree structure maintains a dynamic set of squares under insertion and deletion in $O(\log^5 n)$ time and reports changes to an independent subset of these squares whose size is expected to be a $4128=O(1)$-approximate factor approximation of the maximum independent set.
\end{theorem}
\begin{proof}
This structure only does a constant number of operations, where each is an operation on a link-cut tree, a path structure, or a search structure. These all have $O(\log^5 n)$ amortized cost. 

In lemma~\ref{l:dynpath} we showed that we can maintain an approximate protected independent set each path that was within a linear factor of optimal, with constants $c_1=\frac{1}{64}$ and $c_2=3$.
By lemma~\ref{l:combine} the approximation factor for centered squares was obtained as a function  of these constants (which were better in the static case). Specifically, Lemma~\ref{l:combine} says that we have an $\frac{c_1}{2c_1+c_2+1}=258$-approximation for centered squares.
Finally, as in theorem~\ref{t:static}, lemma~\ref{l:bds} is applied to convert the approximation factor on centered squares to an expected one on all squares, at a cost of a factor of 16.
Thus our overall approximation factor is at most 4128 in expectation. 
\end{proof}


\subsection{The Interval Query Data Structure (IQDS) for Squares}
\label{sec:squaresiqds}

In this section, we prove the following:

\begin{lemma} \label{l:squaresiqds}
We can support the following operations in the time indicated:
\begin{itemize}
    \item $\jop{Get-Interval}(\Path,q,s)$: Given a path $\Path \in \Qpaths$ and $q$, and a square $s \in \Squares{\Pathq{q}}$, report the interval $\Interval{\Pathq{q}}{s}$ in time $O(\log n)$
    \item $\jop{Report-Leftmost}(\Path,q,\ell_1,\ell_2)$: Given a path $\Path$ and quadrant $q$, among all squares in  $s \in \Squares{\Pathq{q}}$ with left endpoint of its interval $\Interval{\Pathq{q}}{s}$ in $(\ell_1,\ell_2)$ report the one with minimal right endpoint in time $O(\log^5 n)$.
    \item $\jop{Report-Rightmost}(\Path,q,r_1,r_2)$: Given a path $\Path$ and quadrant $q$, among all squares in  $s \in \Squares{\Pathq{q}}$ with right endpoint of its interval $\Interval{\Pathq{q}}{s}$ in $(r_1,r_2)$ report the one with maximal left endpoint in time $O(\log^5 n)$.
\end{itemize}
\end{lemma}

The above shows that the queries of Interval Query Data Structure (IQDS) can be answered for each dynamic interval structure, of which there is one associated with each monotone path $\Pathq{q}$. The path and quadrant serves as the identifier of which interval data structure the query is to be executed on.

This is a key to the efficiency of our method. We are able to store all of the squares in one data structure, the search structure, so that the intervals of squares can be easily computed on the fly given the top and bottom and quadrant number of the monotone path they lie on. In this way, seemingly impossible changes like when during a structural change to the quadtree causes two paths to merge, and many intervals associated with squares grow, are easily handled as the search structure already has enough information to query the intervals of any monotone path that is consistent with the squares stored without any changes.

It also means that split, merge, insert, and delete need not be directly implemented in the IQDS as used for dynamic squares, they can return having done nothing. This is because in implementing the IQDS here we have access to the search structure, which is maintained by the quadtree structure to contain all of the squares and marks.
Thus the search structure and as well as the top and bottom of the path being queried and quadrant is the only thing needed to answer a query, and this information is in the monotone path that owns each instance of the dynamic intervals structure that calls the IQDS. 

We now go through a series of technical lemmas that prove the above Lemma~\ref{l:squaresiqds} beginning with a lemma that shows a mapping between squares on a given monotone path with intervals in a certain range and a region of 4-dimensional space.
The reader is encouraged to see Figure~\ref{fig:twoells} which provides motivation for the first technical lemma.

\begin{figure}
\begin{minipage}[c]{3in}
\includegraphics[width=\textwidth, clip=true, trim=0 0 0 1pc   ]{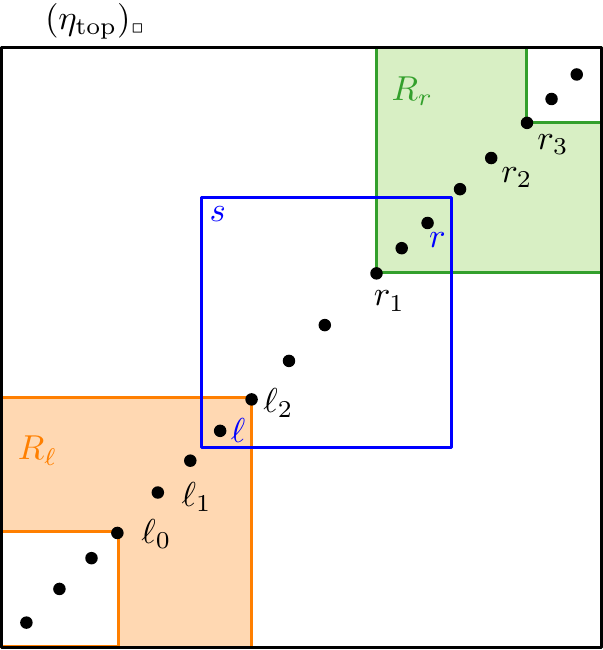}
\end{minipage}
\hspace{1pc}
\begin{minipage}[c]{3in}
\caption{The black points represent the centers of the squares of nodes of some monotone path $\Pathq{\qa}$, which are monotone and increasing in depth as you go to the upper right. The basic fact illustrated is given a square $s$, the  following two statements are equivalent: (1) The first of the centers $s$ intersects is of a node of depth $\ell$ between $\ell_1$ and $\ell_2$ and the last of these centers $s$ intersects is a node of depth $r$ is between $r_1$ and $r_2$. (2) The lower left corner of $s$ must be in the region labelled $R_\ell$ and the upper-right corner must be in the region labelled $R_r$.}
    \label{fig:twoells}
\end{minipage}
\end{figure}

\begin{lemma} \label{l:geomofintervals}
Given an $\Path \in \Qpaths$, a quadrant $q$ assumed without loss of generality to be $\qa$, and values $\ell_1 \leq \ell_2 \leq r_1 \leq r_2]$ where there are nodes on $\Pathq{q}$ with depths $\ell_1,\ell_2,r_1$ and $r_2$: 
\begin{itemize}[noitemsep,nolistsep]
\item Let $\Pathq{\qa}[d]$ be the node on the monotone path of the depth $d$, if it exists.
\item Let $\ell_0$ be the depth of the next shallowest node on $\Pathq{\qa}$ before $\Pathq{\qa}[\ell_1]$
which will be $-\infty$ if $\Pathq{\qa}[\ell_1]$ is already the shallowest node of type $\qa$ in $\Path$.
\item Let $r_3$ be the depth of next deepest node on $\Pathq{\qa}$ after $\Pathq{\qa}[r_2]$, which will be $\infty$ if $\Pathq{\qa}[r_2]$ is already the deepest node of type $\qa$ in $\Path$.
\item Let $\node_{\text{top}} \coloneqq \Path[d[\Pathtop]+1]$ be the top node on the path $\Path$, thus a child of the internal node $\Pathtop$ that serves to define $\Path.$
\item Let $R_{\ell}$ be the region that is in the lower-left quarter-plane relative to $\Squarecenter{\Path[\ell_2]}$ but not in the lower-left quarter-plane relative to $\Squarecenter{\Path[\ell_0]}$ (the second condition is omitted if $\ell_0 =  - \infty$).
\item Let $R_r$ be the region that is in the upper-left quarter-plane relative to $\Squarecenter{\Path[r_1]}$ but not in the upper-left quarter-plane relative to $\Squarecenter{\Path[r_3]}$ (the second condition is omitted if $r_3=\infty$).
\end{itemize}
\noindent Then:
\begin{itemize}[noitemsep,nolistsep]
    \item Any square $s \in \Pathq{\qa}
    $ with $\Interval{\Pathq{\qa}}{s} 
    =[\ell,r]$, where $\ell \in [\ell_1,\ell_2]$ and $r \in [r_1,r_2]$ has its lower left endpoint in $R_{\ell} \cap \Square{(\node_{\text{top}})}$ and its upper right endpoint is in $R_{r} \cap \Square{(\node_{\text{top}})}$
\item For any square $s$ if
\begin{itemize}[noitemsep,nolistsep]
\item $\Node{s}$ is a monochild node of type $\qa$ and
\item $s$'s lower left endpoint is in $R_{\ell} \cap \Square{(\node_{\text{top}})}$ and 
\item $s$'s upper right endpoint is in $R_{r} \cap \Square{(\node_{\text{top}})}$
\end{itemize}
\noindent then:
\begin{itemize}[noitemsep,nolistsep]
    \item $\Node{s}\in \Pathq{\qa}$ and
    \item $\Interval{\Pathq{\qa}}{s} = [\ell,r]$, where $\ell\in[\ell_1,\ell_2]$ and $r\in[r_1,r_2]$.
\end{itemize}
\end{itemize}
\end{lemma}

\begin{proof}
We prove the first claim first and assume we have a square $s \in \Pathq{\qa}$ with $\Interval{\Pathq{\qa}}{s}=[\ell,r]$ where $\ell \in [\ell_1,\ell_2]$ and $r \in [r_1,r_2]$.

\begin{itemize}[noitemsep,nolistsep]
\item The square $s$ is contained in $\Node{s}$ by definition. The region $\Square{\Node{s}}$ is contained in $\Square{(\node_{\text{top}})}$ as $\node_{\text{top}}$ is an ancestor (or equal to) $\Node{s}$ in the quadtree. Thus, the square $s$ is in $\Square{(\node_\text{top})}$.
\item The lower-left endpoint of $s$ is in the lower-left quadrant of $\Path[\ell]=\Node{s}$. Any centers nodes on the path $\Path$ deeper then $\Path[\ell]$ will be contained in the upper-left quadrant of $\Path[\ell]=\Node{s}$. Thus the lower-left endpoint of $s$ is in the lower left quarter-plane relative to $\Squarecenter{\Path[\ell_2]}$ as $\ell_2 \geq \ell$.
\item If $\ell_0$ is defined, $s$ can not be in the lower-left quarter-plane relative to $\ell_0$. This is because then $s$ would include $\ell_0$ (as we know $s$ includes $\ell_2$ which is to the upper-left of $\ell_0$), and this would mean that $s$ includes the center of a node higher in the quadtree than $\Node{s}$, a contradiction with $\ell = d(\Node{s})$.
\item The upper-right endpoint of $s$ is in the upper-right quadrant of $\Path[r]$. Any nodes on the path $\Pathq{\qa}$ with depth at most $r$  will contain $\Square{\Path[r]}$ in their upper-right quadrant. 
Thus the upper-right endpoint of $s$ is in the upper-right quarter-plane relative to $\Squarecenter{\Path[r_1]}$ as $r_1 \leq r$.
\item If $r_3$ is defined, $s$ can not be in the upper-left quarter-plane relative to $r_3$. 
This is because then $s$ would include $r_3$ (as we know $s$ includes $r_2$ which is to the lower-right of $r_3$), and this would mean that $s$ includes the center of a node on $\Pathq{\qa}$ deeper then $r$, a contradiction.
\end{itemize}

To prove the second claim we assume we have a square $s$ where 
$n(s)$ is a monochild node of type $\qa$,
$s$'s lower left endpoint is in $R_{\ell} \cap \Square{(\node_{\text{top}})}$ and 
$s$'s upper right endpoint is in $R_{r} \cap \Square{(\node_{\text{top}})}$.
As $s$ lies entirely in $\Square{(\node_{\text{top}})}$ then $\Node{s}$ is either on $\Path$ or is a descendent. But as the square $s$ includes $\Squarecenter{\Path [r_1]}$ ($R_r$ is entirely up the upper-right of $r_1$ and $R_{\ell}$ is entirely to the lower-left), we know that $\Node{s}$ is $\Path[r_1]$ or an ancestor. Thus $\Node{s}$ is on $\Path$, and as it is of type $\qa$, on $\Pathq{\qa}$.
From the geometry of the regions $R_1$ and $R_2$ we know that $\Interval{\Pathq{\qa}}{s} = [\ell,r]$, where $\ell\in[\ell_1,\ell_2]$ and $r\in[r_1,r_2]$. 
If $\ell< \ell_1$, $\ell_0$ must be defined and then the lower-left corner of the square would be to the lower-left of $\ell_0$ and thus not in $R_\ell$. 
If $\ell>\ell_2$, then the lower-left corner of $s$ would not be to the lower-left of $\ell_2$ and thus would not be in $R_\ell$. 
If $r<r_1$, then the upper-right corner of $s$ would not be to the upper-right of $r_1$ and thus not include $R_r$.
If $r>r_2$, $r_3$ must be defined and the upper right corner of $s$ would be to the upper-right of $r_3$ and thus not include $R_r$.
\end{proof} 

With this technical lemma in had, now we show how with small amount of manipulation, the claims of main lemma of this section, Lemma~\ref{l:squaresiqds}, can be proven.

\begin{lemma} \label{l:algogeom}
Given an $\Path \in \Qpaths$, a quadrant $q$ assumed without loss of generality to be $\qa$, and values $\ell_1 \leq \ell_2 \leq r_1 \leq r_2]$ where there are nodes on $\Pathq{q}$ with the four depths, 
one can determine in $O(\log^4 n)$ time whether there is some square $s$ with $\Node{s} \in \Pathq{q}$ and with $\Interval{\Pathq{\qa}}{s} = [\ell,r]$ where $\ell \in [\ell_1,\ell_2]$ and $r \in [r_1,r_2]$, and if so, report the square.
\end{lemma}

\begin{proof} \label{l:geomsearch}
From Lemma~\ref{l:geomofintervals}, we know that any such square must have one corner in the regions 
$R_1$ and the other in $R_2$, as defined in that lemma. Both of these ranges are easily computed, and are of constant complexity, being either rectangles or L-shapes. Thus after decomposing each L shape into two rectangles, four calls to $\jop{RangeSearch}$ in the search structure suffice to answer the query.
\end{proof}

\begin{lemma}
Given an $\Path \in \Qpaths$, a quadrant $q$, and values $\ell_1 \leq \ell_2]$, 
one can determine in $O(\log^5 n)$ time whether there is some square $s$ with $\Node{s} \in \Pathq{q}$ and with $\Interval{\Pathq{\qa}}{s} = [\ell,r]$ where $\ell \in [\ell_1,\ell_2]$ and among all such squares has minimal $r$.
\end{lemma}

\begin{lemma}
Given an $\Path \in \Qpaths$, a quadrant $q$, and values $r_1 \leq r_2]$, 
one can determine in $O(\log^5 n)$ time whether there is some square $s$ with $\Node{s} \in \Pathq{q}$ and with $\Interval{\Pathq{\qa}}{s} = [\ell,r]$ where $r \in [r_1,r_2]$ and among all such squares has maximal $\ell$.
\end{lemma}

\begin{proof}
(of both lemmas) Use Lemma~\ref{l:geomsearch} and binary search to find the minimal $r$/maximal $\ell$.
\end{proof}

\begin{lemma}
Given a path $\Path \in \Qpaths$ and quadrant $q$, and a square $s \in \Squares{\Pathq{q}}$, the values of $\Depth{s}$ and $\Depthmax{\Pathq{q}}{s}$ can be computed in $O(\log n)$ time
\end{lemma}

\begin{proof}
From $s$ we can compute $\Node{s}$, and from the size of $\Square{\Node{s}}$ the negation of the depth can be obtained by a discrete binary logarithm. For $\Depthmax{\Pathq{q}}{s}$, we need to find the depth of the deepest node in the part of $\Pathq{q}$ starting at $\Node{s}$ that $s$ intersects. This is just a binary search along the path restricted to the nodes of quadrant $q$ which can be done with the aid of the link-cut tree in $O(\log n)$ time.
\end{proof}

\subsection{Splitting and Merging of Interval Structures} \label{s:splitmergeintervals}
This subsection describes in detail what the path structures do when they need to split or merge interval structures. First we describe why this is not straightforward and provide some justification for the overall complexity of our structure.

The curious reader who has made it this far may wonder why we have gone through all the effort to make sure that the dynamic interval structure only interfaces with the intervals via the operations of Lemma~\ref{lem:ors}, and that it does not store any intervals not in the current independent set. Can't it just store in a BST all intervals that are currently in its (full, not just independent) set? Alas, the answer is no. There are two reasons why our more complex approach is needed. The first is that when a node transitions form internal to multichild, and thus its square would join a path, we would need to create a structure representing a set of possibly large size. If we had to pass all these intervals to the independent set, we would lose our runtime guarantees. We get around this by the fact that these squares/intervals are already in the search structure and as we will by changing the mark of all of them (in polylog time) they will be visible via the search structure to the interval structure.

\begin{figure}[ht]
    \centering
    \includegraphics{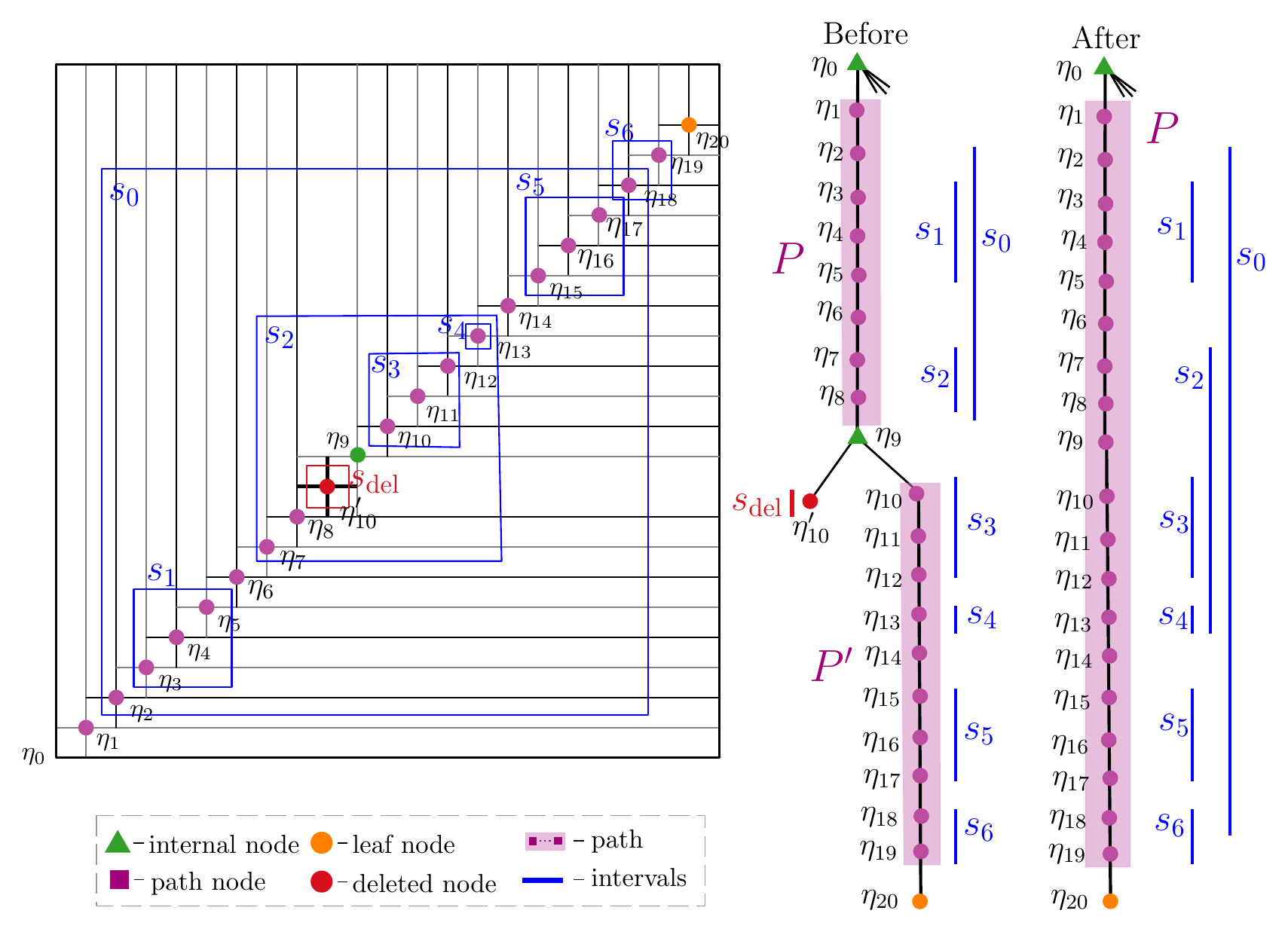}
    \caption[fuck]{Intervals change. The interval of a nodes $s$ on a path $\Pathq{\qa}$ spans the range of depths from the depth of highest node on  $\Pathq{\qa}$ that it intersects the center of to the depth of the deepest node on $\Pathq{q}$ that it intersects the center of. So, for example, in the figure, square $s_0$ has node $\node(s_0)=\node_2$ which is on the path from $\node_1$ to $\node_8$. As it intersects the centers of the nodes from $\node_2$ to $\node_8$, its interval $\Interval{\Pathq{\qa}}{s}=[\Depth{\node_3},\Depth{\node_8}]$. But, what happens if the square $s_{\text{del}}$ is deleted, which causes the leaf $\node'_10$ to be removed? The node $\node_9$ would then be monochild instead of being an internal node, and the path $P$, the node $\node_9$, and the path $P'$ will be merged into one path, illustrated on the right. But now, on this path now there are nodes beyond the former last one $\node_8$, which square $s_0$ intersects. As a result $\Interval{\Pathq{\qa}}{s}$ is now $[\Depth{\node_3},\Depth{\node_{18}}]$. Note that this expansion of an interval only happens in the limited case where an interval included already the last node on a path.
    
    Observe that if one were to view this process in reverse, starting with the after picture and inserting $s_{\text{del}}$ the effect is to take the two intervals which span the node which becomes internal and breaks the path into two, $s_2$ and $s_0$, and clip them to the depth of the bottom node of the new top path.
    }
    \label{fig:depthschange}
\end{figure}

The second subtlety is that given a square $s$, 
its interval $\Interval{\Pathq{q}}{s}\coloneqq [\Depth{s},\Depthmax{\Path}{s}]$ could change! Recall that $\Depthmax{\Pathq{q}}{s}$ is the maximum depth of the deepest node on $\Pathq{q}$ that $s$ intersects the center of. But what happens if $s$ intersects the deepest node on $\Pathq{q}$, and then because of merge, $\Path$ becomes longer? The depth $\Depthmax{\Path}{s}$ could increase.  See Figure~\ref{fig:depthschange} for worked-out example of how this can happen. Our notation has reflected this $\Interval{\Pathq{q}}{s}$ includes the monotone path $\Pathq{q}$ precisely because the interval is a function of the monotone path and can change if the monotone path that it lies on changes.

So, we must not store the intervals $[\Depth{s},\Depthmax{\Path}{s}]$ explicitly, as in a single path merge, an unbounded number of intervals could change, and what they change to will depend on the path that is merged on the bottom. This seems hopeless, until one realizes that this uncertainty as to the right endpoint of an interval is only among those intervals that intersect the last node in the path, thus only those whose right endpoint is to the right of the rightmost left endpoint. This is the magic of the search structure, given a square $s$ on a node of type $q$ on path $\Path$ and the $\Pathtop$ and $\Pathbottom$ of the path, the depth of the deepest node of type $q$ on the path that $s$ intersects can be computed, and this is very much a function of the path.

Given all of this, suppose we have two path structures $P_1$ and $P_2$ that are to be merged, where $P_1$ has a dynamic intervals structure $S_1$ and $P_2$ has dynamic interval structure $S_2$. 
Suppose $P_1$ has nodes with depths from $d_1$ to $d'_1$ and $P_2$ has depths from $d_2$ to $d'_2$, with $d_1\leq d'_1 < d_2 \leq d'_2$.

We make first make one note that when storing an interval representing depths $[a,b]$, we actually store $[a-\frac{1}{3},b+\frac{1}{3}]$ in the dynamic interval interval structure. This makes no difference to anything said so far as the intervals we store and the queries we make have have integer depths, but will allow us in the next paragraph to insert an interval which is sure to be disjoint or contained in all others.

Thus, before merging, we know that $S_1$ stores intervals with coordinates in $[d_1-\frac{1}{3},d'_1+\frac{1}{3}]$, $S_2$ stores intervals in the range $[d_2-\frac{1}{3},d'_2+\frac{1}{3}]$, and as the depths are integer with $d'_1<d_2$, these ranges are disjoint.
After the merge, some intervals in $S_1$ which had as their right endpoints as $d'_1$ may now have endpoints in $[d_2,d'_2]$. As shown in Lemma~\ref{lem:insert_superset} the dynamic interval structure can support intervals growing, so long as they are not part of the independent set. So, before merging the structures, we insert the interval $[d'_1-\frac{1}{6},d'_1+\frac{2}{6}]$ into $S_1$. This has the effect that any intervals that might be elongated are contained in the newly inserted interval and due to the $k$-valid property are not part of the independent set. It is at this point that we view the intervals as being elongated.
Then the merge operation is carried out on the structure. Then $[d'_2+\frac{1}{6},d'_2+\frac{2}{6}]$ is removed from the resultant structure.

For splitting a similar phenomenon occurs, where if we are to split an path into two paths with depths in the ranges $[d_1,d'_1]$ and $[d_2,d'_2]$, during the split any intervals which spanned these two ranges will be store in the first dynamic interval structure and will have their right endpoints clipped to $d'_1$. In Observation~\ref{obs:valid_leftmost} we have shown that this can be done.

\subsection{Hypercubes and Beyond} \label{s:hyper}

We note that while we have presented everything for squares, everything holds for hypercubes in higher dimension. 

Two $2^d$ factors are lost in the approximation for a total loss of $2^{2d}$. The first comes from the chance that a hypercube is centered.
Quadtrees naturally extend to higher dimensions, as to the notions of a monotone subpath. In dimension $d$, there will be $2^d$ different monotone subpaths so by taking the best 
of them (or an approximation of the best) will lose a factor of $2^d$ rather than the 4 of Fact~\ref{f:max4}.

As for the runtime, the searching structure will be $2d$ dimensional instead of $4$ dimensional, thus queries in this structure will cost $O(\log^{2d} n)$ instead of $O(\log^4 n)$. As we binary search in that structure, this brings the final runtime up to $O(\log^{2d+1} n)$ instead of $O(\log^5 n)$.

Alas, out methods do not extend obviously to general rectangles or circles. Our initial trick of throwing away all squares that are not centered relative to the quadtree works well for any fat object, but will fail for general rectangles.
For circles or other non-rectangular fat objects, Lemma~\ref{l:monotone}, as illustrated in Figure~\ref{l:monotone} relies very much on the fact if an object contains two points in opposite quadrants, it must contain the origin; this fact holds only for axis-aligned rectangles. Thus causes our arguments to break down for circles, and other objects of possible interest, including non-axis-aligned squares.